\documentclass[lineno]{biometrika_arxiv_class}
\pdfoutput=1

\usepackage{amsmath}
\usepackage{amssymb}
\usepackage{amsthm}
\usepackage{xcolor}
\usepackage{graphics}
\usepackage{url}

\usepackage{times}
\usepackage{bm}
\usepackage{bbm}
\usepackage{natbib}

\graphicspath{{./figures/}}

\usepackage{algorithm}
\usepackage{algorithmicx}

\renewcommand{\S}{\mathcal{S}}
\newcommand{\V}{\textup{var}}
\newcommand{\vol}{\textup{vol}}

\newcommand{\R}{\mathbb{R}}
\newcommand{\E}{\mathbb{E}}
\newcommand{\D}{\mathcal{D}}
\newcommand{\argmax}[1]{\underset{#1}{\operatorname{arg}\,\operatorname{max}}\;}

\newcommand{\power}{\text{Power}}

\newcommand{\LRT}{\text{LRT}}
\newcommand{\CF}{\text{CF}}
\newcommand{\subsplit}{\text{subsplit}}
\newcommand{\spl}{\text{split}}
\newcommand{\RIPR}{\text{RIPR}}
\newcommand{\KL}{\text{KL}}
\newcommand{\one}{\mathbbm{1}}
\renewcommand{\dim}{\text{dim}}
\renewcommand{\hat}{\widehat}

\renewcommand{\P}{\mathbb{P}}
\renewcommand{\S}{\mathcal{S}}
\renewcommand*\bar[1]{%
   \hbox{%
     \vbox{%
       \hrule height 0.5pt 
       \kern0.5ex
       \hbox{%
         \kern-0.1em
         \ensuremath{#1}%
         \kern-0.1em
       }%
     }%
   }%
} 

\newtheorem{theorem}{Theorem}
\newtheorem{lemma}{Lemma}

\begin{document}

\nolinenumbers

\jname{Biometrika}



\title{Gaussian universal likelihood ratio testing}

\author{Robin Dunn}
\affil{Novartis Pharmaceuticals Corporation, Advanced Methodology and Data Science,\\ East Hanover, New Jersey 07936, U.S.A.
\email{robin.dunn@novartis.com}}

\author{Aaditya Ramdas, Sivaraman Balakrishnan, \and Larry Wasserman}
\affil{Department of Statistics \& Data Science, Carnegie Mellon University,\\ Pittsburgh, Pennsylvania 15213, U.S.A.
\email{aramdas@stat.cmu.edu} \email{siva@stat.cmu.edu} \email{larry@stat.cmu.edu}}

\maketitle

\begin{abstract}
The classical likelihood ratio test (LRT) based on the asymptotic chi-squared distribution of the log likelihood is one of the fundamental tools of statistical inference. A recent universal LRT approach based on sample splitting provides valid hypothesis tests and confidence sets in any setting for which we can compute the split likelihood ratio statistic (or, more generally, an upper bound on the null maximum likelihood). The universal LRT is valid in finite samples and without regularity conditions. This test empowers statisticians to construct tests in settings for which no valid hypothesis test previously existed. For the simple but fundamental case of testing the population mean of $d$-dimensional Gaussian data with identity covariance matrix, the classical LRT itself applies. Thus, this setting serves as a perfect test bed to compare the classical LRT against the universal LRT. This work presents the first in-depth exploration of the size, power, and relationships between several universal LRT variants. We show that a repeated subsampling approach is the best choice in terms of size and power. For large numbers of subsamples, the repeated subsampling set is approximately spherical. We observe reasonable performance even in a high-dimensional setting, where the expected squared radius of the best universal LRT's confidence set is approximately 3/2 times the squared radius of the classical LRT's spherical confidence set. We illustrate the benefits of the universal LRT through testing a non-convex doughnut-shaped null hypothesis, where a universal inference procedure can have higher power than a standard approach.
\end{abstract}


\section{Introduction}

Suppose we have data from an unknown distribution $P_{\theta^*}$ which belongs to some set of distributions $(P_\theta : \theta \in \Theta)$. We wish to test the composite null hypothesis $H_0: \theta^* \in \Theta_0$. We use the observed data to construct a test statistic $T_n$ and reject $H_0$ if $T_n$ exceeds a level $\alpha$ critical value, which we denote $c_\alpha$. A level $\alpha$ test is valid in finite samples if 
\begin{equation}
\sup_{\theta^* \in \Theta_0} P_{\theta^*}(T_n > c_\alpha) \leq \alpha. \label{eq:level_alpha_test}
\end{equation}
The test is asymptotically valid at level $\alpha$ if 
\begin{equation}
\lim_{n\to\infty} \sup_{\theta^* \in \Theta_0} P_{\theta^*}(T_n > c_\alpha) \leq \alpha. \label{eq:asymp_level_alpha_test}
\end{equation}

We are primarily interested in universal test statistics that satisfy (\ref{eq:level_alpha_test}). For completeness, though, we will compare against a common hypothesis testing approach that satisfies (\ref{eq:level_alpha_test}) in our specific Gaussian setting of interest and satisfies (\ref{eq:asymp_level_alpha_test}) more generally.
Consider the alternative $H_1: \theta \in \Theta \: \backslash \: \Theta_0$. The generalized likelihood ratio statistic is $\mathcal{L}(\hat{\theta}) \: / \: \mathcal{L}(\hat{\theta}_0)$, where $\hat{\theta}$ is the maximum likelihood estimate in $\Theta$ and $\hat{\theta}_0$ is the maximum likelihood estimate in $\Theta_0$. Let $\dim(\Theta)$ represent the dimension of $\Theta$ in Euclidean space, and likewise for $\Theta_0$. We reject $H_0$ when $2 \log\{\mathcal{L}(\hat{\theta}) \: / \: \mathcal{L}(\hat{\theta}_0)\} > c_{\alpha, d}$, where $c_{\alpha, d}$ is the upper $\alpha$ quantile of the $\chi^2_d$ distribution and $d = \dim(\Theta) - \dim(\Theta_0)$. This construction arises from Wilks' Theorem \citep{wilks1938}, which states that $2\log\{\mathcal{L}(\hat{\theta}) \: / \: \mathcal{L}(\hat{\theta}_0)\}$ has an asymptotic $\chi^2_d$ distribution under certain regularity conditions. This will apply, for instance, when we have independent and identically distributed data from an exponential family, $\Theta_0$ is a subset of $\Theta$, and $\Theta$ and $\Theta_0$ are linear subspaces in Euclidean space \citep[Theorem 4.6]{van2000asymptotic}. 
We can invert the likelihood ratio test (LRT) to produce an asymptotically valid $100(1-\alpha)\%$ confidence region of the form:
\begin{equation*}
C_n^\LRT(\alpha) = \left\{\theta\in\Theta : 2\log \left\{\mathcal{L}(\hat{\theta}) \: / \:\mathcal{L}(\theta)\right\} \leq c_{\alpha, d} \right\}.
\end{equation*}
We reject $H_0$ if and only if $C_n^\LRT(\alpha) \cap \Theta_0 = \emptyset$, which is equivalent to rejecting $H_0$ if and only if $2 \log\{\mathcal{L}(\hat{\theta}) \: / \: \mathcal{L}(\hat{\theta}_0)\} > c_{\alpha, d}$. We refer to this testing framework as the classical LRT. Some composite nulls are irregular, meaning that Wilks' theorem does not apply and calculating a threshold can be hard due to intractable asymptotics.

The universal inference approach developed by \cite{wasserman2020universal} provides a new likelihood ratio testing framework that addresses situations where the classical LRT is not valid in finite samples, or potentially even asymptotically. This new LRT relies on sample splitting to construct a test and confidence set that are valid in finite samples and without regularity conditions. This universal inference method allows one to construct valid tests in settings for which no hypothesis test with type I error control and finite sample guarantees previously existed. The statistical literature has repeatedly emphasized the inadequacy of the asymptotic $\chi^2$ approximation in the small sample setting. Examples include \cite{bartlett1937properties}, \cite{lehmann2012likelihood}, and \cite{medeiros2017small}. Small sample sizes also pose a recurrent problem across biological science research. For instance, researchers have noted the prevalence of low-powered studies in neuroscience \citep{button2013power} and the need for clinical trial designs that account for the small sample sizes common to rare disease and pediatric population research \citep{ildstad2001small, mcmahon2016stratification}.

Many basic questions remain unanswered about the universal LRT, since its power even in very simple settings remains unknown. Further, \cite{wasserman2020universal} describe numerous settings in which the universal LRT is the first hypothesis test with finite sample validity. These settings include testing the number of components in mixture models \citep{hartigan1985, mclachlan1987, chen2009, li2010} and testing whether the underlying density satisfies the shape constraint of log-concavity \citep{cule2010}. As a precursor to studying the power in these important but as-yet intractable settings, we first study the universal LRT in the fundamental case of constructing confidence regions or hypothesis tests for the population mean $\theta^*\in\R^d$ when $Y_1, \ldots, Y_n \sim N(\theta^*, I_d)$. In this setting --- where the classical LRT is not only valid but also exact --- our results showcase the reasonable performance of the universal LRT in comparison to the classical approach. The universal LRT will still apply if the covariance matrix is unknown, but this requires fitting the maximum likelihood estimates of both the mean and covariance matrix. Furthermore, confidence regions that are spherical under the identity covariance matrix may no longer be spherical in the general covariance matrix setting. With more technical effort, it is possible to characterize the distribution of the split LRT test statistic beyond the Gaussian setting. For instance, \cite{strieder2022choice} derive a non-central split chi-squared distribution, which governs the asymptotic behavior of the split LRT statistic under local alternatives in regular settings where the classical LRT is valid. 

This work provides two main contributions: First, we provide a careful analysis of several variants of the universal LRT in the Gaussian case. We show that a repeated subsampling approach is the best choice in terms of size and power. We observe reasonable performance in a high-dimensional setting, where the expected squared radius of the best universal LRT confidence set is approximately 3/2 times the squared radius of the set constructed through the classical approach. Thus, in particular, the power of the universal approach has the same behavior in $n,d,\alpha$ as the classical approach. Second, we show an example of a hypothesis test on normally distributed data where universal LRT methods have higher power than classical testing methods. Specifically, when testing the non-convex ``doughnut'' null $H_0: \|\theta^*\|\in [0.5, 1]$ versus $H_1: \|\theta^*\|\notin [0.5,1]$ on $N(\theta^*, I_d)$ data, a universal LRT approach can have higher power than a standard approach that uses the classical LRT confidence set. A test of this form could examine, for instance, whether trial outcomes or biomarker levels are within an acceptable range.

\section{Universal LRT confidence sets}

\subsection{Universal LRT background}

\cite{wasserman2020universal} presented an alternative to the LRT that is valid in finite samples without requiring regularity conditions. Suppose we have $n$ independent and identically distributed observations $Y_1, \ldots, Y_n \sim P_{\theta^*}$, where $P_{\theta^*}$ is from a family ${(P_\theta: \theta\in\Theta)}$. Each $P_\theta$ has a density denoted by $p_\theta$. To implement the test, first partition the data into $\D_0$ and $\D_1$. Let $\hat{\theta}_1$ be an estimator constructed from $\D_1$. The parameter $\hat{\theta}_1$ could be the maximum likelihood estimate, but any parameter that is fixed given $\D_1$ is valid. Certain choices of $\hat{\theta}_1$ may be more efficient. Using the data in $\D_0$, the likelihood function is $\mathcal{L}_0(\theta) = \Pi_{Y_i\in \D_0} p_\theta(Y_i)$. Define the split LRT statistic as 
\begin{align*}
T_n(\theta) &= \mathcal{L}_0(\hat{\theta}_1) / \mathcal{L}_0(\theta).
\end{align*}
The universal confidence set for $\theta^*$ using the split LRT is
\begin{align*} 
C_n^{\spl}(\alpha) &= \{\theta \in \Theta: T_n(\theta) < 1/\alpha \}.
\end{align*}

\begin{theorem} \label{thm:valid}
$C_n^{\spl}(\alpha)$ is a valid $100(1-\alpha)\%$ confidence set for $\theta^*$. As a consequence, and equivalently, when testing an arbitrary composite null $H_0: \theta^*\in\Theta_0$ versus $H_1: \theta^*\in\Theta\:\backslash\:\Theta_0$, rejecting $H_0$ when $\Theta_0 \cap C_n^{\spl}(\alpha) = \emptyset$ provides a valid level $\alpha$ hypothesis test. The latter rule reduces to rejecting if $T_n(\hat \theta_0) \geq 1/\alpha$, where $\hat \theta_0 = \arg\max_{\theta \in \Theta_0} \mathcal{L}_0(\theta)$ is the maximum likelihood estimate under $H_0$.
\end{theorem} 
Theorem~\ref{thm:valid} is due to \cite{wasserman2020universal}. The validity of the universal test does not depend on large samples or regularity conditions. The proof establishes that $E_{\theta^*}\left\{T_n(\theta^*) \right\} \leq 1$ and then invokes Markov's inequality. See Section~\ref{sec:supp_thm} of the supplement for more details. This property on the expectation makes $T_n(\theta^*)$ an e-variable.  An instantiation of an e-variable is an e-value. For related work under varying terminology, see also the research on e-variables \citep{vovk2021values, grunwald2020safe}, betting scores \citep{shafer2021testing}, supermartingales \citep{shafer2011test, howard2020time, ignatiadis2022values}, the prediction-based-ratio protocol \citep{zhang2011asymptotically}, and the game theoretic version of e-variable-based confidence sets, which are called warranty sets, by \cite{shafer2021testing}.

The validity of $C_n^\spl(\alpha)$ only depends on 
the fact that $\E_{\theta^*}\{T_n(\theta^*)\} \leq 1$. If we consider multiple test statistics that each satisfy this condition, then the average of those test statistics will satisfy the condition as well. Therefore, the average of test statistics $T_n(\theta^*)$ across multiple data splits is also a valid test statistic. In fact, the ability to combine e-values through averaging without adjusting $\alpha$ is one benefit of e-values over p-values. E-value averaging is a common theme in the discussions of \cite{shafer2021testing}, including \cite{vovk2020note}.

\subsection{Classical test in normal setting}

Assume $Y_1, \ldots, Y_n$ are $d$-dimensional independent and identically distributed vectors drawn from $N(\theta^*, I_d)$ with $\theta^*\in\R^d$. Where $c_{\alpha, d}$ is the upper $\alpha$ quantile of the $\chi^2_d$ distribution, the classical LRT confidence set for $\theta^*$ is 
\begin{align}
C_n^\LRT(\alpha) &= \left\{\theta\in\Theta : \|\theta-\bar{Y}\|^2 \leq c_{\alpha, d} / n \right\}. \label{eq:C_n_usual}
\end{align}
See Section~\ref{sec:supp_eq} of the supplement for a derivation of (\ref{eq:C_n_usual}). In this case, $C_n^\LRT(\alpha)$ is valid in finite samples, since $n\|\theta^* - \bar{Y}\|^2$ follows a $\chi^2_d$ distribution. We compare $C_n^\LRT(\alpha)$ to the split LRT set and several universal confidence sets that are variants of the split LRT set.

\subsection{Split, cross-fit, and subsampling sets in normal setting}

First, we consider two universal LRT variants based on a single split of the data. Assume we split the $n$ observations in half, such that $\D_0$ and $\D_1$ each contain $n/2$ observations. Define $\bar{Y}_0 = (2/n)\sum_{Y_i \in \D_0} Y_i$ and $\bar{Y}_1 = (2/n)\sum_{Y_i \in \D_1} Y_i$. Then the confidence set for $\theta^*$ based on the split likelihood ratio is
\begin{align*}
C_n^{\spl}(\alpha) &= \left\{\theta\in\Theta: \exp\left(-\frac{n}{4}\|\bar{Y}_0 - \bar{Y}_1\|^2 + \frac{n}{4}\|\bar{Y}_0 - \theta\|^2\right) < \frac{1}{\alpha} \right\} \\
&= \left\{\theta \in \Theta: \|\theta - \bar{Y}_0\|^2 < (4/n)\log(1/\alpha) + \|\bar{Y}_0 - \bar{Y}_1\|^2 \right\}. \stepcounter{equation}\tag{\theequation}\label{eq:C_n_split}
\end{align*}

See Section~\ref{sec:supp_eq} of the supplement for a derivation of (\ref{eq:C_n_split}). Using the same split, we define the cross-fit statistic as $S_n(\theta) = \{T_n(\theta) + T_n^{\text{swap}}(\theta)\}/2$, where $T_n^{\text{swap}}(\theta)$ is computed by switching the roles of $\D_0$ and $\D_1$. 
Then the cross-fit confidence set is a valid $100(1-\alpha)\%$ set given by
\begin{equation*}
C_n^\CF(\alpha) = \Bigg\{\theta\in\Theta : \frac{1}{2} \exp\left(-\frac{n}{4}\|\bar{Y}_0 - \bar{Y}_1\|^2\right) \left\{\exp\left(\frac{n}{4}\|\bar{Y}_0 - \theta\|^2\right) + \exp\left(\frac{n}{4}\|\bar{Y}_1 - \theta\|^2 \right) \right\} < \frac{1}{\alpha} \Bigg\}.
\end{equation*}

The split and cross-fit sets have both statistical randomness, due to the random sampling of observations, and algorithmic randomness, due to the randomness in splitting the data into $\D_0$ and $\D_1$. In contrast, the classical LRT only has statistical randomness, since the test is deterministic for a given set of observations. We now consider a repeated subsampling approach. This universal method attempts to mitigate the algorithmic randomness from the split and cross-fit LRTs by splitting the observations many times and averaging the test statistics. Algorithm~\ref{alg:subsampling} shows how to compute the subsampling test statistic $T_n(\theta)$ at a given $\theta\in\R^d$. 

\begin{algorithm} 
\caption{Compute the subsampling test statistic $T_n(\theta)$.}
\label{alg:subsampling}
\begin{tabbing}
   \qquad \textit{Input:} $n$ independent $d$-dimensional observations $Y_1, \ldots, Y_n \sim N(\theta^*, I_d)$ ($\theta^*$ unknown),\\ 
   \qquad \qquad \quad a value of $\theta\in\R^d$, number of subsamples $B$. \\
   \qquad \textit{Output:} The subsampling test statistic $T_n(\theta)$. \\
   \qquad \enspace For $b=1$ to $b = B$: \\
   \qquad \qquad Randomly split the data into $\D_{0,b}$ and $\D_{1,b}$, each containing $n/2$ values of $Y_i$. \\
   \qquad \qquad Let $\bar{Y}_{0,b} = (2/n) \sum_{Y_i\in \D_{0,b}} Y_i$ and let $\bar{Y}_{1,b} = (2/n) \sum_{Y_i\in \D_{1,b}} Y_i$. \\
   \qquad \qquad Compute $T_{n,b}(\theta) = \exp\left(-\frac{n}{4}\| \bar{Y}_{0,b} - \bar{Y}_{1,b}\|^2 + \frac{n}{4}\| \bar{Y}_{0,b} - \theta\|^2 \right)$. \\
\qquad \enspace Output the subsampling test statistic $T_n(\theta) = B^{-1} \sum_{b=1}^B T_{n,b}(\theta)$.
\end{tabbing}
\end{algorithm}
As noted earlier,
this method is also valid. The $100(1-\alpha)\%$ subsampling confidence set is  
\begin{equation*}
C_n^{\subsplit}(\alpha) = \left\{\theta \in \Theta: \frac{1}{B} \sum_{b=1}^{B} \exp\left(-\frac{n}{4}\| \bar{Y}_{0,b} - \bar{Y}_{1,b}\|^2 + \frac{n}{4}\| \bar{Y}_{0,b} - \theta\|^2 \right) < \frac{1}{\alpha} \right\}.
\end{equation*}

Figure~\ref{fig:fixed_data_regions} shows coverage regions of the classical LRT, split LRT, cross-fit LRT, and subsampling LRT at $B = 100$ from six simulations with $\theta^* = (0, 0)$. We generate 1000 observations from $N(\theta^*, I_2)$, and we use this sample for all simulations. Hence, the variation in the split, cross-fit, and subsampling LRTs across simulations is due to algorithmic randomness. We use \texttt{ggConvexHull} for plotting because these confidence sets are all convex. See Section~\ref{sec:convex} of the  supplement for proof that these sets are convex.

The coverage regions in Figure~\ref{fig:fixed_data_regions} suggest several relationships that we will formalize. We see that the classical LRT provides the smallest confidence regions. This is not surprising since, even in finite samples, the classical LRT statistic follows a chi-squared distribution under $H_0: \theta = \theta^*$ in the Gaussian case. The volume of the cross-fit LRT set is less than or equal to the volume of the split LRT set, although the cross-fit set is not entirely contained within the split set. The split and cross-fit approaches both use a single split of the data, but there is a notable improvement from cross-fitting. The subsampling set also has less volume than the split LRT set. Recall that we construct the subsampling test statistic by performing the split LRT over repeated splits of the data and then averaging the test statistics $T_{n,b}(\theta)$. While any individual split LRT region is guaranteed to be spherical, the subsampling set is not necessarily a spherical region. For large $B$, however, we see that the subsampling region is approximately spherical. Thus, although the subsampling approach is computationally intensive, this hints that it may be possible to derive a formulaic approximation to the limiting subsampling set.

\begin{figure}
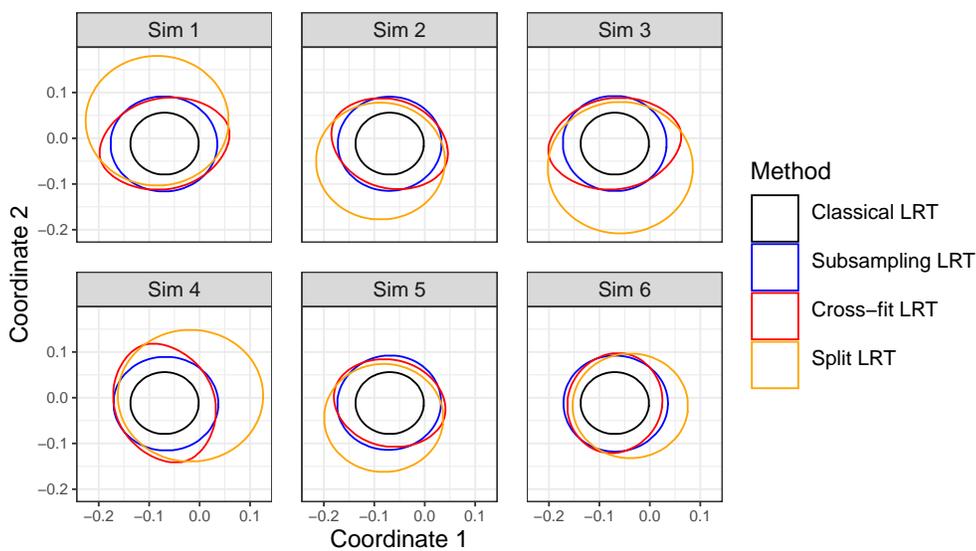

\figuresize{.65}
\figurebox{20pc}{25pc}{}[Figure1.pdf]
\caption{Coverage regions of classical LRT (black), subsampling LRT (blue), cross-fit LRT (red), and split LRT (orange) at $\alpha = 0.1$. The six simulations use the same 1000 observations from $N(\theta^*, I_2)$ under $\theta^* = (0, 0)$.}
\label{fig:fixed_data_regions}
\end{figure}

\subsection{Limit of subsampling region}

We are particularly interested in the behavior of the subsampling confidence set as $B\to\infty$. Since $B^{-1} \sum_{b=1}^B T_{n,b}(\theta) \to \E\{T_n(\theta) \mid \D\}$ as $B\to\infty$, the limiting subsampling set has no algorithmic randomness. We see hints of this in Figure~\ref{fig:fixed_data_regions}, where the subsampling set at $B=100$ does not vary much across six simulations on the same data. Theorem~\ref{thm:limiting_ratio} describes conditions for the convergence of $\E\{T_n(\theta) \mid \D\}$ over an approximation to 1. We have been suppressing the $n$ subscript when it is clear we are working with a single dataset with $n$ observations. Theorem~\ref{thm:limiting_ratio} considers a sequence of datasets, so we use the $n$ subscript to index the datasets.

\begin{theorem} \label{thm:limiting_ratio}
Assume we have a sequence of datasets $(\D_{n})_{n\in 2\mathbb{N}}$, where $\D_n = \{Y_{n1}, \ldots, Y_{nn}\}$ and each $Y_{ni}$ is an independent observation from $N(\theta^*, I_d)$. Let $\D_{0,n}$ be a sample of $n/2$ observations from $\D_n$, and let $\D_{1,n} = \D_n \backslash \D_{0,n}$. Define $\bar{Y}_n = (1/n)\sum_{i=1}^n Y_{ni}$, $\bar{Y}_{0,n} = (2/n) \sum_{Y_{ni} \in \D_{0,n}} Y_{ni}$, and $\bar{Y}_{1,n} = (2/n) \sum_{Y_{ni} \in \D_{1,n}} Y_{ni}$. Let $c > 0$, and let $(\theta_n)$ be a sequence that satisfies $\|\bar{Y}_n - \theta_n\| \leq c/\sqrt{n}$ for all $n$. Then
\begin{align}
\E\{T_n(\theta_n) \mid \D_n\} \: / \: \left\{\exp\left(\frac{3n}{10} \|\bar{Y}_n - \theta_n\|^2 \right) \left(\frac{2}{5}\right)^{d/2} \right\} &= 1 + o_P(1).
\end{align}
\end{theorem} 

In words,
the subsampling statistic is approximately given by
$R(\theta)^{3/5}(2/5)^{d/2}$
where
$R(\theta)= \mathcal{L}(\hat\theta)/\mathcal{L}(\theta)$
is the usual likelihood ratio statistic.

Section~\ref{sec:supp_thm} of the supplement contains a proof of Theorem~\ref{thm:limiting_ratio}. The proof relies critically on the finite sample central limit theorems from \cite{hajek1960limiting} and \cite{li2017general} and on the Portmanteau Theorem proof techniques from \cite{van2000asymptotic}.
 
Since
\begin{align}
\E\{T_n(\theta) \mid \D\} \approx \exp\left(\frac{3n}{10} \|\bar{Y}-\theta\|^2 \right) \left(\frac{2}{5}\right)^{d/2}, \label{eq:ddimapprox}
\end{align} 
the subsampling confidence region is approximately
\begin{align*}
C_n^{\subsplit}(\alpha) &= \left\{\theta \in \Theta: \lim_{B\to\infty} \frac{1}{B} \sum_{b=1}^{B} \exp\left(-\frac{n}{4}\| \bar{Y}_{0,b} - \bar{Y}_{1,b}\|^2 + \frac{n}{4}\| \bar{Y}_{0,b} - \theta\|^2 \right) < \frac{1}{\alpha} \right\} \\
&\approx \left\{\theta\in\Theta: \|\bar{Y}-\theta\|^2 < \frac{10}{3n} \log\left( \frac{(5/2)^{d/2}}{\alpha} \right)  \right\}. \stepcounter{equation}\tag{\theequation} \label{eq:subsplit_radius}
\end{align*}

The approximations in (\ref{eq:ddimapprox}) and (\ref{eq:subsplit_radius}) only formally hold in the setting described in Theorem~\ref{thm:limiting_ratio}. Still, Figure~\ref{fig:expect_analytical_compare} validates (\ref{eq:ddimapprox}) as a reasonable approximation. We simulate one sample $Y_1, \ldots, Y_n \sim N(0, I_d)$ at $d=1$ and $d=20$, where $n = 1000$. We consider $\theta$ values of the form $\theta = c\vec{1}$. Through $B = 100,000$ subsampling simulations, we estimate 
\begin{align*}
\E\{T_n(\theta) \mid \D\} &\approx \frac{1}{B} \sum_{b=1}^{B} \exp\left(-\frac{n}{4}\| \bar{Y}_{0,b} - \bar{Y}_{1,b}\|^2 + \frac{n}{4} \| \bar{Y}_{0,b} - \theta\|^2 \right).
\end{align*}
The black dots represent this average at each value of $c$, and the red curve traces out $\exp((3n/10) \|\bar{Y} - \theta\|^2) (2/5)^{d/2}$ from (\ref{eq:ddimapprox}). Except for the most difficult setting of $(d = 20, n = 10)$, the simulated and analytical estimates align well. At $\alpha = 0.1$, the confidence region includes all values of $\theta$ such that the test statistic is at most $1/0.1$. The horizontal dashed black line represents this value. Thus, test statistics constructed from the simulated and analytical approaches would produce similar confidence regions.

\begin{figure}
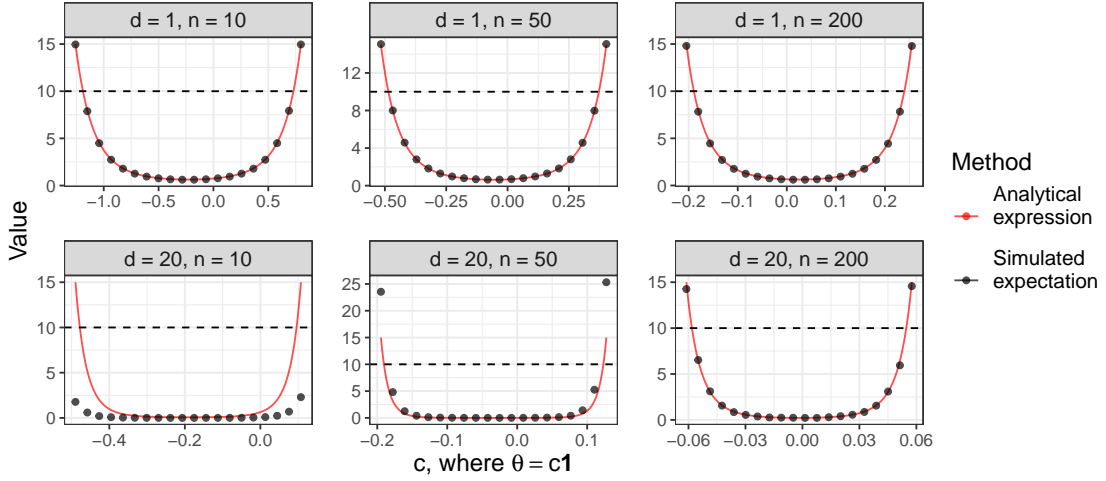

\figuresize{.65}
\figurebox{20pc}{25pc}{}[Figure2.pdf]
\caption{Analytical (red curve) and simulated (black dots) approximations of the limiting test statistic $\lim_{B\to\infty} \frac{1}{B}\sum_{b=1}^B T_{n,b}(\theta)$ at various dimensions $d$ and numbers of observations $n$. The test points equal $\theta = c\vec{1}$ for various $c$. The horizontal dashed black line at $1/0.1$ is the cutoff for an $\alpha = 0.1$ confidence region.}
\label{fig:expect_analytical_compare}
\end{figure}

\section{Comparison of universal LRT sets}

\subsection{Optimal split proportions}

We have been assuming that the universal LRTs place $n/2$ observations in $\D_0$ and $n/2$ observations in $\D_1$. The statement $\E_{\theta^*}\{T_n(\theta^*) \} \leq 1$ holds regardless of the proportion of observations in $\D_0$ versus $\D_1$, though. Let $p_0$ denote the proportion of observations that we place in $\D_0$. Recalling from expression~(\ref{eq:C_n_split}) that $C_n^\spl(\alpha)$ is a spherical set, let $r^2\{C_n^{\spl}(\alpha)\}$ denote the squared radius of $C_n^{\spl}(\alpha)$. Theorem~\ref{thm:split_p0} solves for the value of $p_0$ that minimizes $\E[r^2\{C_n^{\spl}(\alpha)\}]$, using the fact that $\|\bar{Y}_0 - \bar{Y}_1\|^2 \sim (4/n)\chi^2_d$.

\begin{theorem} \label{thm:split_p0}
Let $Y_1, \ldots, Y_n \sim N(\theta^*, I_d)$. The splitting proportion that minimizes $\E[r^2\{C_n^{\spl}(\alpha)\}]$ is 
\begin{align}
p_0^* &= 1 - \frac{\sqrt{4d^2 + 8d\log\left(1/\alpha\right)} - 2d}{4\log\left(1/\alpha\right)}.
\end{align}
\end{theorem}

As $d\to\infty$ for fixed $\alpha$, the optimal split proportion $p_0^*$ converges to 0.5. See Section~\ref{sec:supp_thm} of the supplement for a proof of Theorem~\ref{thm:split_p0} and a derivation of this fact. Alternatively, as $\alpha\to 0$ for fixed $d$, the proportion $p_0^*$ converges to 1, suggesting that one should use nearly all data for likelihood estimation. This is not an issue for reasonable $\alpha$ levels, though. For instance, at $d=1$, one would need to set $\alpha < \exp(-40)$ to produce an optimal split proportion $p_0^*$ that exceeds 0.90. 

Figure~\ref{fig:p0_sq_rad} shows the average squared radius of the split LRT at $p_0^*$ and at surrounding choices of $p_0$. The expected squared radius, given by the red curve, is more sensitive to changes in $p_0$ at higher values of $d$. That is, use of the optimal $p_0^*$ has a greater effect on the split LRT squared radius in higher dimensions. In high dimensions, though, $p_0^*$ is close to 0.5. It is thus a reasonable choice to use $p_0 = 0.5$ in all dimensions. We use $p_0 = 0.5$ for all remaining analyses.

\begin{figure}
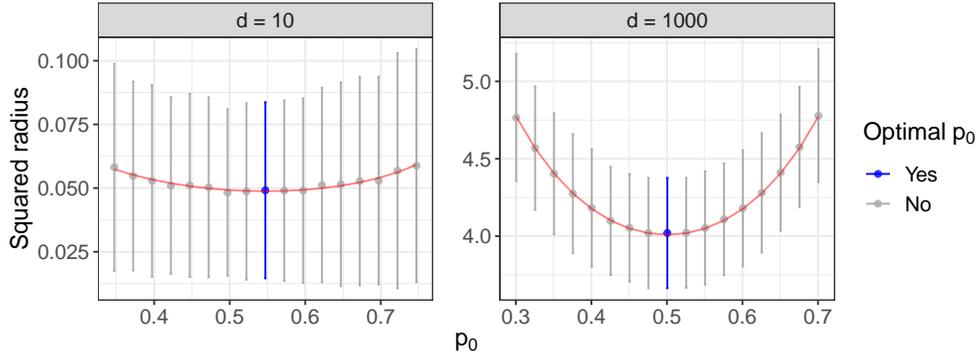

\figuresize{.65}
\figurebox{20pc}{25pc}{}[Figure3.pdf]
\caption{Squared radius of multivariate normal split LRT with varying $p_0$. We simulate $Y_1, \ldots, Y_{1000} \sim N(0, I_d)$ and compute the split LRT region at $\alpha = 0.10$ and varying $p_0$. We repeat this simulation 1000 times. At each $p_0$, the circular point is the mean squared radius and the error bar represents the mean squared radius $\pm$ 1.96 standard deviations. Hence, the error bars represent a typical range of squared radius values for each $d$ and $p_0$.  Blue points/lines correspond to $p_0^*$. The red curve is the expected squared radius. See Theorem~\ref{thm:split_p0} proof in the supplement for a derivation of the expected squared radius at $p_0$.}
\label{fig:p0_sq_rad}
\end{figure}

Recent work by \cite{strieder2022choice} arrives at somewhat different conclusions about the optimal split proportion in a setting that is similar to this work. Here we have solved for the optimal $p_0^*$ in the context of minimizing $\E[r^2\{C_n^\spl(\alpha)\}]$ for confidence set construction. \cite{strieder2022choice} similarly show results on $d$-dimensional multivariate normal data, but they test the null hypothesis that the first $d-k$ values of the $d$-dimensional mean vector equal zero. In settings where $k\approx d$, the split proportion that maximizes the power is similar to our $p_0^*$ and also converges to 0.5 in high dimensions. In settings where $k \ll d$, their optimal split ratio uses $p_0 < 0.5$, and the gap between their split ratio and 0.5 becomes larger in higher dimensions. Hence, if using the standard version of the split LRT, the optimal split ratio may vary depending on whether the goal is confidence set construction or hypothesis testing.

In the cross-fit case, we conjecture that $p_0 = 0.5$ minimizes the expected squared diameter. Simulations in Section~\ref{sec:supp_sim_CF} of the supplement support this claim. Intuitively, since the cross-fit approach uses both $\D_0$ and $\D_1$ once for parameter estimation and once for likelihood computation, we should not gain any efficiency by using unbalanced sets.

\subsection{Split versus cross-fit volume}

In Figure~\ref{fig:fixed_data_regions}, we see that the cross-fit LRT set $C_n^\CF(\alpha)$ is not a subset of the split LRT set $C_n^{\spl}(\alpha)$. Nevertheless, both empirically and theoretically, the volume of $C_n^{\CF}(\alpha)$, denoted $\vol\{C_n^\CF(\alpha)\}$, is less than or equal to $\vol\{C_n^{\spl}(\alpha) \}$. Theorem~\ref{thm:CF_subset} proves that the cross-fit LRT constructs smaller confidence sets than the split LRT. See Section~\ref{sec:supp_thm} of the supplement for a proof.

\begin{theorem} \label{thm:CF_subset}
Suppose $Y_1, \ldots, Y_n$ are independent and identically distributed observations from  $N(\theta^*, I_d)$. Split the sample such that $\D_0$ and $\D_1$ each contain $n/2$ observations. Use $\D_0$ and $\D_1$ to define the split and cross-fit sets. Then $C_n^\CF(\alpha)$ is a subset of a translation of the split LRT set, recentered at $\bar{Y}$. That is, $C_n^\CF(\alpha) \subseteq \left\{\theta \in \Theta: \|\theta - \bar{Y}\|^2 < (4/n)\log(1/\alpha) + \|\bar{Y}_0 - \bar{Y}_1\|^2 \right\}$, and hence $\vol\{C_n^\CF(\alpha)\} \leq \vol\{C_n^{\spl}(\alpha)\}$. Furthermore, if and only if $\bar{Y}_0 = \bar{Y}_1$, $C_n^\CF(\alpha)$ and $C_n^\spl(\alpha)$ have equal volume and are in fact the same set.
\end{theorem}

Out of all universal methods, our simulations have shown that the subsampling approach tends to produce the smallest sets. Constructing a subsampling region can be computationally intensive, though, especially when the limiting subsampling test statistic is intractable. The cross-fit approach may be a reasonable compromise in settings where repeated subsampling is computationally prohibitive.

\subsection{Bounds on the size of universal LRT sets} \label{sec:comparative_size}

Figure~\ref{fig:fixed_data_regions} demonstrated the appearance of the four LRT regions in the $d=2$ case at $\alpha = 0.1$. We observe that the classical LRT and the split LRT produce the smallest and largest confidence regions, respectively. While the split LRT region's radius appears to be approximately twice the classical LRT region's radius, we consider whether the ratio of their squared radii diverges in high dimensions or for very small $\alpha$. We characterize the ratio of squared radii in terms of the expected ratio. The expected squared radius of $C_n^{\spl}(\alpha)$ is 
\begin{align}
\E[r^2\{C_n^{\spl}(\alpha)\}] &= (4/n)\log(1/\alpha) + (4/n)d. \label{eq:sq_rad_split}
\end{align}
Thus, the expected ratio of the split LRT squared radius over the classical LRT radius is
\begin{align}
\frac{\E[r^2\{C_n^{\spl}(\alpha)\}]}{r^2\{C_n^\LRT(\alpha)\}} &= \frac{(4/n)\log(1/\alpha) + (4/n)d}{c_{\alpha,d}/n} = \frac{4\log(1/\alpha) + 4d}{c_{\alpha,d}}. \label{eq:split_usual}
\end{align}
For $d\geq 2$ and $\alpha \leq 0.17$,
\begin{align}
\frac{4\log(1/\alpha) + 4d}{2\log(1/\alpha) + d + 2\sqrt{d\log(1/\alpha)}} \leq \frac{\E[r^2\{C_n^{\spl}(\alpha)\}]}{r^2\{C_n^\LRT(\alpha)\}} \leq \frac{4\log(1/\alpha) + 4d}{2\log(1/\alpha) + d - 5/2}. \label{eq:ratio_sq_rad_1}
\end{align}
For $d = 1$ and $\alpha \leq \exp\left(-\frac{5(1+\sqrt{5})}{4} \right),$ 
\begin{align}
\frac{4\log(1/\alpha) + 4}{2\log(1/\alpha) + 1 + 2\sqrt{\log(1/\alpha)}} \leq \frac{\E[r^2\{C_n^{\spl}(\alpha)\}]}{r^2\{C_n^\LRT(\alpha)\}} \leq \frac{4\log(1/\alpha) + 4}{2\log(1/\alpha) + 9 - 4\sqrt{5 + 2\log(1/\alpha)}}. \label{eq:ratio_sq_rad_2}
\end{align}

See Section~\ref{sec:supp_eq} of the supplement for derivations of (\ref{eq:sq_rad_split}), (\ref{eq:ratio_sq_rad_1}), and (\ref{eq:ratio_sq_rad_2}). The derivation of (\ref{eq:ratio_sq_rad_1}) relies on chi square quantile bounds from Theorem A and Proposition 5.1 of \cite{inglot2010inequalities}. The derivation of the upper bound in (\ref{eq:ratio_sq_rad_2}) involves a bound from \cite{pollard2015} and \cite{feller1968probability}. The restrictions on $\alpha$ and $d$ are necessary for the upper bounds to be valid. The lower bound is valid for any $d\geq 1$ and $\alpha\in (0,1)$. The upper and lower bounds both converge to 4 as $d \to \infty$. In addition, all bounds converge to 2 as $\alpha \to 0$. Figure~\ref{fig:split_usual} shows the true value of $\E[r^2\{C_n^{\spl}(\alpha)\}] \: / \: r^2\{C_n^\LRT(\alpha)\}$ as well as the proved lower and upper bounds on this expectation at $d = 10$ and $d = 100,000$. We observe that the bounds converge to 2 for very small $\alpha$ relative to the dimension, and we observe  that the bounds converge to 4 for high dimensions relative to $\alpha$. Interestingly, we see that the expected value of the ratio is not monotone increasing in $\alpha$.

\begin{figure}
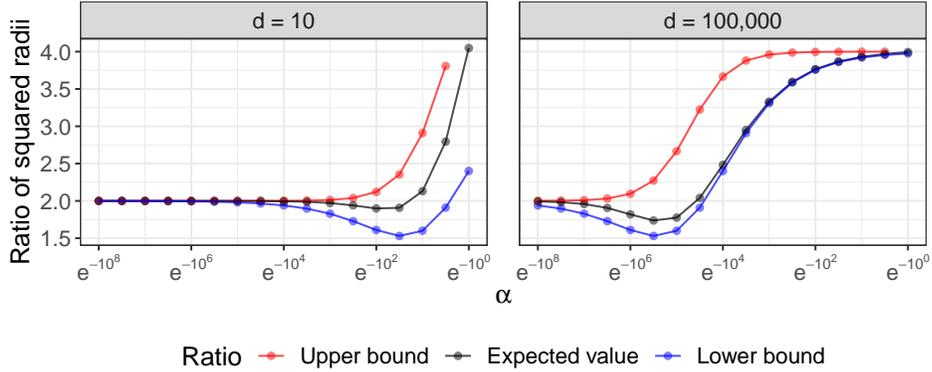

\figuresize{.7}
\figurebox{20pc}{25pc}{}[Figure4.pdf]
\caption{Expectation (black), lower bound (blue), and upper bound (red) of $E\left[r^2\left\{C_n^{\spl}\right(\alpha)\} / r^2\{C_n^\LRT(\alpha)\}\right]$. The expected value equals the expression from (\ref{eq:split_usual}). The lower and upper bounds correspond to the bounds in (\ref{eq:ratio_sq_rad_1}). Data points correspond to values at $\alpha = \exp(-10^x)$ for $x$ from 8 to 0 in increments of $-0.5$.}
\label{fig:split_usual}
\end{figure}

Furthermore, $r\{C_n^{\spl}(\alpha)\} / r\{C_n^\LRT(\alpha)\} \leq 2$ with probability of approximately $1-\alpha$ in high dimensions. Theorem~\ref{thm:bounds} formalizes this result. See Section~\ref{sec:supp_thm} of the supplement for a proof.

\begin{theorem} \label{thm:bounds}
Let $f_d(x)$ be the probability density function of the $\chi^2_d$ distribution, and let $c_{\alpha, d}$ be the upper $\alpha$ quantile of the $\chi^2_d$ distribution. Assume $c_{\alpha, d} + \log(\alpha) > d - 2$. Then
\begin{align*}
\P\left[r\{C_n^{\spl}(\alpha)\} / r\{C_n^\LRT(\alpha)\} \leq 2\right] &\geq 1 - \alpha - \log(1/\alpha) f_d\{c_{\alpha, d} + \log(\alpha)\} \\ 
\text{and} \quad \P\left[r\{C_n^{\spl}(\alpha)\} / r\{C_n^\LRT(\alpha)\} \leq 2\right] &\leq 1 - \alpha - \log(1/\alpha) f_d(c_{\alpha, d}).
\end{align*}
As $d\to\infty$ for fixed $\alpha \leq 0.17$, $\log(1/\alpha) f_d\{c_{\alpha, d} + \log(\alpha)\}$ and $\log(1/\alpha) f_d(c_{\alpha, d})$ both converge to 0.
\end{theorem}

As one sufficient condition for Theorem~\ref{thm:bounds}, if $d \geq 2$ and $\alpha \leq 0.17$, then it holds that $c_{\alpha, d} + \log(\alpha) > d - 2$. Figure~\ref{fig:ratio_leq_4_lowdim} shows that in high dimensions, the bounds from Theorem~\ref{thm:bounds} are close to $1-\alpha$ and, hence, close to each other. Both theoretically and empirically, the ratio of radii $r\{C_n^{\spl}(\alpha)\} / r\{C_n^\LRT(\alpha)\}$ is less than 2 with probably slightly below $1-\alpha$ in higher dimensions. 

\begin{figure}
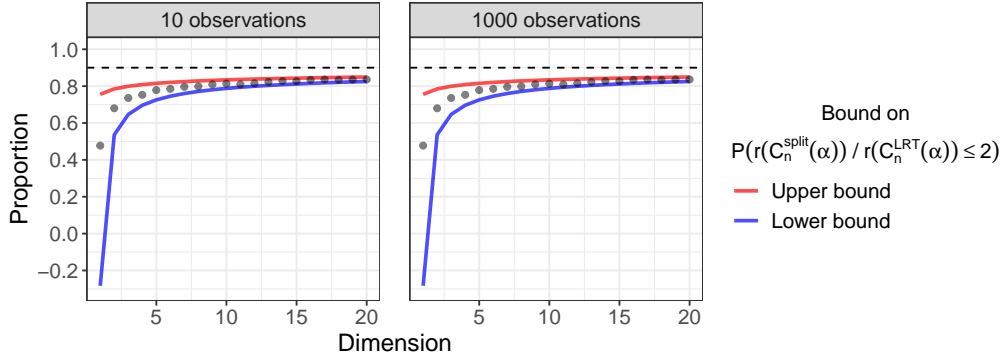

\figuresize{.65}
\figurebox{20pc}{25pc}{}[Figure5.pdf]
\caption{We perform 10,000 simulations in which we simulate a data sample $Y_1, \ldots, Y_{1000} \sim N(0, I_2)$, construct the split and classical LRT confidence sets, and compute the squared radii. The points represent the proportion of these simulations in which $r\{C_n^{\spl}(\alpha)\} / r\{C_n^\LRT(\alpha)\} \leq 2$. The red and blue curves are the lower and upper bounds on $\P[r\{C_n^{\spl}(\alpha)\} \: / \: r\{C_n^\LRT(\alpha)\} \leq 2]$ from Theorem~\ref{thm:bounds} at $\alpha = 0.1$.}
\label{fig:ratio_leq_4_lowdim} 
\end{figure}

From (\ref{eq:subsplit_radius}) and (\ref{eq:sq_rad_split}), we can see that 
\begin{align}
\frac{r^2\{C_n^\subsplit(\alpha)\}}{\E[r^2\{C_n^\spl(\alpha)\}]} \approx \frac{5}{6} \left\{\frac{(d/2)\log(5/2) + \log(1/\alpha)}{d + \log(1/\alpha)} \right\}. \label{eq:subsplit_split}
\end{align}
Combining (\ref{eq:ratio_sq_rad_1}) and (\ref{eq:subsplit_split}), $r^2\{C_n^\subsplit(\alpha)\} / r^2\{C_n^\LRT(\alpha)\}$ is approximately $4(5/12) \log(5/2) \approx 3/2$ as $d\to\infty$ for fixed $\alpha$, and the ratio is approximately $2(5/6) = 5/3$ as $\alpha \to 0$ for fixed $d$. Equivalently, $r\{C_n^\subsplit(\alpha)\} / r\{C_n^\LRT(\alpha)\} \approx 1.24$ as $d\to\infty$ for fixed $\alpha$, and $r\{C_n^\subsplit(\alpha)\} / r\{C_n^\LRT(\alpha)\} \approx 1.29$ as $\alpha\to 0$ for fixed $d$. Recall that the classical LRT cutoff is dimension dependent and uses the exact distribution's quantile, while the universal LRT cutoff is dimension independent. Regardless, in the extreme cases of $d\to\infty$ or $\alpha \to 0$, the ratio of the classical LRT region's radius to the subsampling universal LRT region's radius is less than 2. Although the ratio of radii is bounded by a constant, the ratio of volumes can still become large in high dimensions.

\subsection{Power}

While the universal methods provide conservative confidence regions for $\theta^*$, we establish that the universal tests can still have high power. Suppose we wish to test $H_0: \theta^* = 0$ versus $H_1: \theta^* \neq 0$ at level $\alpha$. We reject $H_0$ if $0 \notin C_n(\alpha),$ where $C_n(\alpha)$ is the confidence set defined by some likelihood ratio test. The power of the test at $\theta^* \neq 0$ is $\P_{\theta^*}\{0 \notin C_n(\alpha)\}.$

First, we consider the classical LRT, stated in (\ref{eq:C_n_usual}). The power of the classical LRT at $\theta^*$ is 
\begin{align}
\power\{C_n^\LRT(\alpha); \theta^*\} &= \P_{\theta^*}\left(\|\bar{Y}\|^2 >  c_{\alpha,d} / n \right) \approx \Phi\left\{\frac{d + n\|\theta^*\|^2 - c_{\alpha, d}}{\sqrt{2(d + 2n\|\theta^*\|^2)}} \right\}. \label{eq:power_usual}
\end{align}
We can find a similar representation for the approximate power of the limiting subsampling LRT as $B \to \infty$:
\begin{align*}
\power\{C_n^{\subsplit}(\alpha); \theta^*\} &\approx \P_{\theta^*}\left[n\|\bar{Y}\|^2 \geq \frac{10}{3} \log\left\{\left(\frac{5}{2}\right)^{d/2} \frac{1}{\alpha} \right\} \right] \\
&\approx \Phi\left(\frac{1}{\sqrt{2(d + 2n\|\theta^*\|^2)}} \left[d + n\|\theta^*\|^2 - \frac{10}{3} \log\left\{\left(\frac{5}{2}\right)^{d/2} \frac{1}{\alpha} \right\} \right] \right). \stepcounter{equation}\tag{\theequation} \label{eq:power_subsplit}
\end{align*}
Since $n\|\bar{Y}\|^2 \sim \chi^2\left(df = d, \lambda = n \|\theta^*\|^2 \right)$, (\ref{eq:power_usual}) and (\ref{eq:power_subsplit}) use the normal approximation to the non-central $\chi^2$ distribution with a large noncentrality parameter $\lambda$ \citep{chun2009normal}. See Section~\ref{sec:supp_eq} of the supplement for derivations of (\ref{eq:power_usual}) and (\ref{eq:power_subsplit}). 

The power of the split LRT is 
\begin{align*}
\power\{C_n^{\spl}(\alpha); \theta^*\} &= \P_{\theta^*}\left\{\|\bar{Y}_0\|^2 \geq (4/n) \log(1/\alpha) + \|\bar{Y}_0 - \bar{Y}_1\|^2 \right\}
\end{align*}
and the power of the cross-fit LRT is
\begin{align*}
\power\{C_n^\CF(\alpha); \theta^*\} &= \P_{\theta^*}\left[ \exp\left(-\frac{n}{4}\|\bar{Y}_0 - \bar{Y}_1\|^2\right) \left\{\exp\left( \frac{n}{4}\|\bar{Y}_0\|^2 \right) + \exp\left(\frac{n}{4}\|\bar{Y}_1\|^2  \right) \right\} \geq \frac{2}{\alpha} \right]. 
\end{align*}

Consider the approximate power of $C_n^\LRT(\alpha)$ and $C_n^\subsplit(\alpha)$ for fixed $\alpha$ if $n\|\theta^*\|^2$ is constant. In this setting, the increase in data as $n\to\infty$, which makes rejecting $H_0$ easier, is offset by moving $\theta^*$ closer to the null, which makes rejecting $H_0$ more challenging.  If $n\|\theta^*\|^2$ is constant, then the approximate power expressions in (\ref{eq:power_usual}) and (\ref{eq:power_subsplit}) are both constant as well. In fact, for both the classical and subsampled split LRTs, $1/\sqrt{n}$ is the exact rate at which to shrink $\|\theta^*\|$ such that the approximate power stays constant as $n$ increases.

As $n\|\theta^*\|^2 \to \infty$ for fixed $\alpha$, the power of the tests approaches 1. Importantly, this shows that although the universal methods are conservative, they will all have high power for sufficiently large $n$ or for $\|\theta^*\|$ sufficiently far from $0$. As $\alpha \to 0$, the power approaches 0.

Figure~\ref{fig:power_vs_norm_2} plots the power of the LRTs against $\|\theta^*\|^2$. Each vector $\theta^*$ has the form $c \vec{1}.$ To plot the classical and subsampling LRT power, this figure uses the standard normal cumulative distribution function approximation to the non-central $\chi^2$ cumulative distribution function. We use simulations to approximate the power of the split and cross-fit LRTs. For a given value of $\theta^*$, we simulate $n=1000$ observations $Y_1, \ldots, Y_n \sim N(\theta^*, I_d)$. We construct split LRT and cross-fit LRT confidence sets from this sample. Then we test whether $\theta = 0$ is in each confidence set. We repeat this procedure 5000 times at each $\theta^*$, and each procedure's  estimated power at $\theta^*$ is the proportion of times that $0 \notin C_n(\alpha)$. 

As we would expect, the power is higher when $\theta^*$ is farther from $0$. In addition, the classical LRT has the highest power, followed in order by the subsampling LRT, the cross-fit LRT, and the split LRT. Interestingly, at $d=1$ the subsampling and cross-fit LRT have nearly identical approximate power. As $d$ increases, the difference between the subsampling and cross-fit LRT power increases.

\begin{figure}
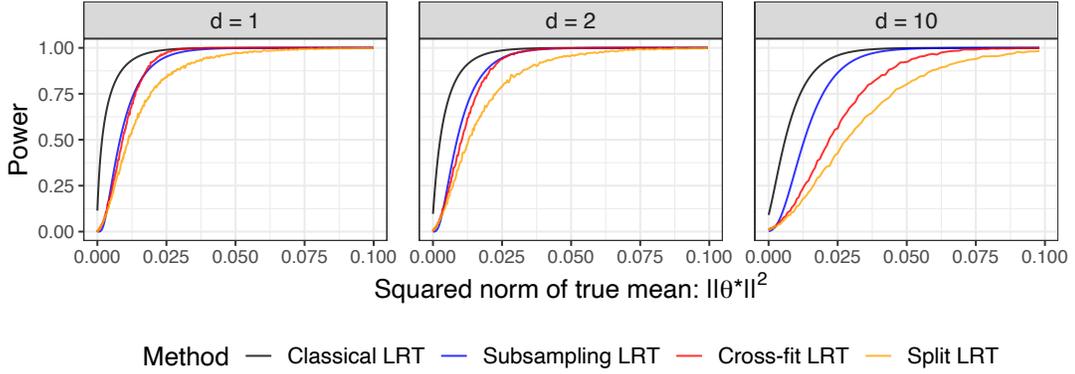

\figuresize{.7}
\figurebox{20pc}{25pc}{}[Figure6.pdf]
\caption{Estimated power of classical LRT (black), limiting subsampling LRT (blue), cross-fit LRT (red), and split LRT (orange). We are testing $H_0: \theta^* = 0$ versus $H_1: \theta^* \neq 0$ across varying true $\|\theta^*\|^2$. We use the standard normal cumulative distribution function approximation for the classical and subsampling LRT power calculations, and we use simulations to estimate the cross-fit and split LRT power.}
\label{fig:power_vs_norm_2}
\end{figure}

\section{Example: hypothesis testing a doughnut null set} \label{sec:example}

Instead of presenting a simulation which further confirms our theoretical findings, we instead present here an example of a nontrivial testing problem that appears to be beyond the current reach of our mathematical analysis. Below, a procedure based on universal inference can have higher power than a more standard intersection approach using the classical, exact confidence set. This motivates the need for further study of the pros and cons of such methods. 

Suppose we observe an independent and identically distributed sample $Y_1, \ldots, Y_n \sim N(\theta^*, I_d)$, and we wish to test
\[H_0: \|\theta^*\| \in [0.5, 1.0] \text{ versus } H_1: \|\theta^*\| \notin [0.5, 1.0].
\]
Then $\Theta_0 = \{\theta\in\R^d : \|\theta\| \in [0.5, 1.0]\}$ and $\Theta_1 = \{\theta\in\R^d : \|\theta\| \notin [0.5, 1.0]\}$. 
The nonconvex structure of $\Theta_0$ makes it unclear how to construct a valid test based on a limiting distribution. Nevertheless, we can use alternative methods, including universal inference tools, to construct valid hypothesis tests for $H_0: \|\theta^*\| \in [0.5, 1.0]$. We compare three approaches to this test. 

\smallskip
\textit{Approach 1: Intersect confidence set with $\Theta_0$.} $C_n^\LRT(\alpha) = \{\theta \in \Theta: \|\theta - \bar{Y}\|^2 \leq c_{\alpha, d} / n \}$ is a level $\alpha$ confidence set for $\theta^*$, where $c_{\alpha, d}$ is the upper $\alpha$ quantile of the $\chi^2_d$ distribution. Suppose we reject $H_0$ if and only if $C_n^\LRT(\alpha) \cap \Theta_0 = \emptyset$. We can see that this test has valid type I error control. Assume $\theta^* \in \Theta_0$. Then
\begin{align*}
\P_{\theta^*}\left\{ C_n^\LRT(\alpha) \cap \Theta_0 = \emptyset \right\} &\leq \P_{\theta^*}\left\{\theta^* \notin C_n^\LRT(\alpha) \cup \theta^* \notin \Theta_0 \right\} \\
&= \P_{\theta^*}\left\{ \theta^* \notin C_n^\LRT(\alpha) \right\} \\
&= \alpha.
\end{align*} 
To implement this test, we need to check whether the intersection $C_n^\LRT(\alpha) \cap \Theta_0$ is empty. First, we set $\hat{\theta}^{\text{proj}}$ to the projection of $\bar{Y}$ onto $\Theta_0$. That is, 
\[\hat{\theta}^{\text{proj}} = \begin{cases}
0.5 \: \bar{Y} / \|\bar{Y}\| &\text{if  } \|\bar{Y}\| < 0.5 \\
\bar{Y} &\text{if  } \|\bar{Y}\| \in [0.5, 1.0] \\
\bar{Y} / \|\bar{Y}\| &\text{if  } \|\bar{Y}\| > 1
\end{cases}.
\]
Now $\hat{\theta}^{\text{proj}}$ minimizes $\|\theta - \bar{Y}\|^2$ out of all $\theta\in\Theta_0$. So $C_n^\LRT(\alpha) \cap \Theta_0 = \emptyset$ if and only if \mbox{$\hat{\theta}^{\text{proj}} \notin C_n^{\LRT}(\alpha)$.}

\smallskip
\textit{Approach 2: Subsampled split LRT.} To implement the subsampled split LRT, we repeatedly split the observations into $\D_{0,b}$ and $\D_{1,b}$. Let $\hat{\theta}_{1,b}$ be any parameter estimated on the data in $\D_{1,b}$. Let $\hat{\theta}_{0,b}^\spl$ be the maximum likelihood estimate under $H_0: \|\theta^*\| \in [0.5, 1.0]$, estimated on the data in $\D_{0,b}$. Table~\ref{tab:subsample_split_ripr} presents the chosen expression for $\hat{\theta}_{1,b}$ and the maximum likelihood estimate of $\hat{\theta}_{0,b}^\spl$. The subsampled split LRT rejects $H_0$ if $B^{-1} \sum_{b=1}^B U_{n,b} \geq 1/\alpha$, where
\begin{align*}
U_{n,b} = \mathcal{L}_{0,b}(\hat{\theta}_{1,b}) \: / \: \mathcal{L}_{0,b}(\hat{\theta}_{0,b}^\spl) = \prod_{Y_i\in \D_{0,b}} \{p_{\hat{\theta}_{1,b}}(Y_i) \: / \: p_{\hat{\theta}_{0,b}^\spl}(Y_i) \}.
\end{align*}

\smallskip
\textit{Approach 3: Subsampled hybrid LRT.} As an alternative to the split LRT, \cite{wasserman2020universal} establish a test based on the reversed information projection (RIPR); also see \cite{grunwald2020safe}. We first define the RIPR, following Definition~4.2 of the PhD thesis by \cite{li1999estimation}. Let $Q$ be a distribution with density $q$, and let $\mathcal{P}_\Theta$ be a convex set of densities, or redefine it as its convex hull. Let $D_\KL(\cdot \: \| \: \cdot)$ be the Kullback-Leibler divergence. The RIPR of $q$ onto $\mathcal{P}_\Theta$ is a (sub-)density $p^*$ such that for arbitrary sequences $p_n$ in $\mathcal{P}_\Theta$, $D_\KL(q \: \| \: p_n) \to \inf_{\theta\in\Theta} D_\KL(q \: \| \: p_\theta)$ implies $\log(p_n) \to \log(p^*)$ in $L^1(Q)$. Lemma 4.1 of \cite{li1999estimation} proves that $p^*$ exists and is unique; further, $p^*$ satisfies $D_\KL(q \: \| \: p^*) = \inf_{\theta\in\Theta} D_\KL(q \: \| \: p_\theta)$, and if $Y \sim q$, then for all $\theta\in\Theta$, $\E_q\{p_\theta(Y) / p^*(Y)\} \leq 1$.

Using similar logic to Theorem~\ref{thm:valid}, \cite{wasserman2020universal} apply this property to construct a  split RIPR LRT. Let $\mathcal{P}_{\Theta_0}$ be the set of all densities in $H_0$ or its convex hull. Suppose $\hat{\theta}_1$ is an estimator constructed on $\D_1$. 
Let $p_0^*$ be the RIPR of $p_{\hat{\theta}_1}$ onto $\mathcal{P}_{\Theta_0}$. If the true $p_{\theta^*} \in \mathcal{P}_{\Theta_0}$, then $\E_{\theta^*}\{p_{\hat{\theta}_1}(Y) / p_0^*(Y)\} = \E_{\hat{\theta}_1}\{p_{\theta^*}(Y) / p_0^*(Y)\} \leq 1$. Then a level $\alpha$ hypothesis test rejects $H_0$ if $R_n \geq 1/\alpha$, where
\begin{equation*}
R_n = \prod_{Y_i \in \D_0} \{p_{\hat{\theta}_1}(Y_i) \: / \: p_0^*(Y_i)\}.
\end{equation*}
This test is valid because if $\theta^* \in \Theta_0$, then
$\P_{\theta^*}(R_n \geq 1/\alpha) \leq \alpha \E_{\theta^*}\{p_{\hat{\theta}_1}(Y) / p_0^*(Y) \} \leq \alpha.$
Furthermore, the RIPR test statistic will always exceed the split LRT statistic when the two tests use the same numerator, since the split LRT denominator maximizes the likelihood under $H_0$ on $\D_0$. Thus, the RIPR test will have higher power than the split LRT. More generally, one can project $p_{\hat{\theta}_1}^{|\D_0|}$ onto $\mathcal{P}_{\Theta_0}^{|\D_0|}$, but we omit this discussion for brevity.

In the doughnut test setting, we let $\mathcal{P}_{\Theta_0}$ be the set of all convex combinations of $N(\theta, I_d)$ densities such that $\|\theta\| \in [0.5, 1]$. To implement the subsampled hybrid LRT for this test, we also repeatedly split the observations into $\D_{0,b}$ and $\D_{1,b}$. Depending on the value of $\|\bar{Y}_{1,b}\|$, we take one of three approaches:

\begin{enumerate}
\item If $\|\bar{Y}_{1,b}\| < 0.5$, use the split LRT on the $b^{th}$ subsample. We define $\hat{\theta}_{1,b}$ and $\hat{\theta}_{0,b}^\spl$ as in Table~\ref{tab:subsample_split_ripr}, and the split LRT statistic is $U_{n,b} = \mathcal{L}_{0,b}(\hat{\theta}_{1,b}) / \mathcal{L}_{0,b}(\hat{\theta}_{0,b}^\spl)$. 
\item If $\|\bar{Y}_{1,b}\| \in [0.5, 1]$, set the $b^{th}$ subsample's test statistic to 1.
\item If $\|\bar{Y}_{1,b}\| > 1$, use the RIPR LRT on the $b^{th}$ subsample.  We define $\hat{\theta}_{1,b}$ and $\hat{\theta}_{0,b}^\RIPR$ as in Table~\ref{tab:subsample_split_ripr}, and the RIPR statistic is $R_{n,b} = \mathcal{L}_{0,b}(\hat{\theta}_{1,b}) / \mathcal{L}_{0,b}(\hat{\theta}_{0,b}^\RIPR)$. 
\end{enumerate}
Theorem~\ref{thm:hybrid_valid} defines a valid test based on this approach, as proved in Section~\ref{sec:supp_thm} of the supplement.

\begin{theorem} \label{thm:hybrid_valid}
In the doughnut null hypothesis test setting, assume the subsampled test statistics $U_{n,b}$ and $R_{n,b}$, $1\leq b\leq B$, as defined above. A valid level $\alpha$ test rejects $H_0$ when $$\frac{1}{B} \sum_{b=1}^B \left\{U_{n,b} \one(\|\bar{Y}_{1,b}\| < 0.5) + \one(\|\bar{Y}_{1,b}\| \in [0.5, 1]) + R_{n,b} \one(\|\bar{Y}_{1,b}\| > 1) \right\} \geq 1/\alpha.$$
\end{theorem}

To justify the hybrid approach, recall that the RIPR test will have higher power than the split LRT when it is possible to implement the RIPR. Based on the construction of $\hat{\theta}_{1,b}$, if $\|\bar{Y}_{1,b}\| > 1$, then $\|\hat{\theta}_{1,b}\| > 1$. In this setting, the proof of Theorem~\ref{thm:hybrid_valid} shows that the density $p_{\theta}$, with $\theta = \hat{\theta}_{1,b} / \|\hat{\theta}_{1,b}\|$, is the RIPR of $\hat{\theta}_{1,b}$ onto $\mathcal{P}_{\Theta_0}$. On the other hand, it is unclear how to implement the RIPR when $\|\bar{Y}_{1,b}\| < 0.5$, in which case $\|\hat{\theta}_{1,b}\| < 0.5$. The hybrid approach allows us to use the RIPR when it is implementable, and it relies on the split LRT to provide a valid test when the RIPR is not implementable.

\begin{table}
\def~{\hphantom{0}}
\renewcommand{\arraystretch}{1.3}
\tbl{Requirements and choices for the numerator and denominator in a single subsample of the split LRT and RIPR LRT statistics}{%
\begin{tabular}{p{1.8cm}p{6.1cm}p{5.2cm}} \hline
Method & Split LRT & RIPR LRT \\ \hline
Restrictions on use & None & $\|\bar{Y}_1\| > 1$. This is a computational restriction. RIPR unknown for $\|\bar{Y}_1\| = \|\hat{\theta}_1\| < 0.5$. \\ \hline
Numerator & $p_{\hat{\theta}_{1}}$, where $\hat{\theta}_{1}$ is any parameter fit on $\D_1$. & $p_{\hat{\theta}_{1}}$, where $\hat{\theta}_{1}$ is any parameter fit on $\D_1$. \\
Fitted value & Choose $\hat{\theta}_{1} = \bar{Y}_1$. & Choose $\hat{\theta}_{1} = \bar{Y}_1$. \\ \hline
Denominator & $p_{\hat{\theta}_0}$, where $\hat{\theta}_0$ is the maximum likelihood estimate under $H_0$, constructed from $\D_0$. & $p_0^*$ is the RIPR of $p_{\hat{\theta}_1}$ onto $\mathcal{P}_{\Theta_0}$. \\ 
Fitted value & No choices.  & No choices. \\
& $\hat{\theta}_{0}^\spl = \begin{cases}
0.5 \left(\bar{Y}_0 / \|\bar{Y}_0\|\right) & : \|\bar{Y}_0\| < 0.5 \\
\bar{Y}_0 & : \|\bar{Y}_0\| \in [0.5, 1] \\
\bar{Y}_0 / \|\bar{Y}_0\| & : \|\bar{Y}_0\| > 1
\end{cases}$ & Since $\|\hat{\theta}_1\| > 1$, $p_0^* = p_\theta$, where $\theta = \hat{\theta}_0^{\RIPR} = \hat{\theta}_1 / \|\hat{\theta}_1\|$. \\ \hline
\end{tabular}}
\label{tab:subsample_split_ripr}
\end{table}

Figure~\ref{fig:power_interval} shows the simulated power of these three tests of $H_0: \|\theta^*\| \in [0.5, 1.0]$ versus $H_1: \|\theta^*\| \notin [0.5, 1.0]$. The intersection method and the subsampled hybrid LRT have the highest power. Interestingly, out of those two methods, the test with higher power varies across dimensions. When $d=2$ or $d=1000$, the simulated power of the subsampled hybrid LRT is less than or equal to the power of the standard intersection approach. At the intermediate dimensions of $d = 10$ and $d=100$, the simulated power of the subsampled hybrid LRT is greater than or equal to the power of the standard intersection approach for $\|\theta^*\| > 1$. The latter two cases show that even in the Gaussian setting, hypothesis tests based on a universal LRT can have higher power than tests based on the exact confidence set. When $\|\theta^*\| < 0.5$, the hybrid test and the split test have approximately the same power. When $\|\theta^*\| > 1$, the hybrid test has higher power than the split test. We see that the intersection method always has higher power than the subsampled split LRT. One might consider whether we could combine the RIPR with the intersection method instead of combining the RIPR with the split LRT. It is unclear how to construct such a valid test, though, since the RIPR approach uses both sample splitting and subsampling while the intersection approach uses neither.

\begin{figure}
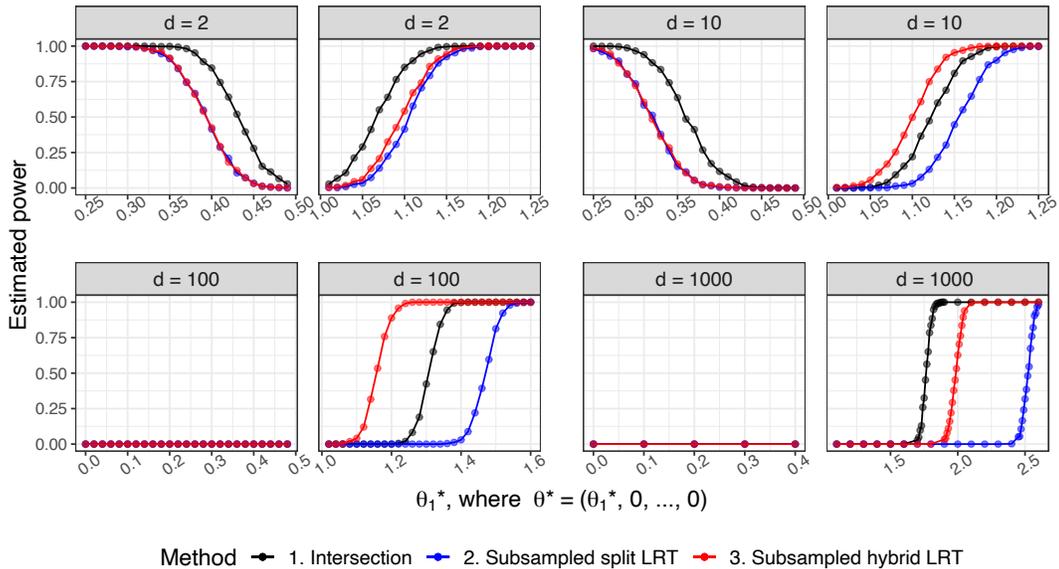

\figuresize{.62}
\figurebox{20pc}{25pc}{}[Figure7.pdf]
\caption{Estimated power of $H_0: \|\theta^*\| \in [0.5, 1.0]$ versus $H_1: \|\theta^*\| \notin [0.5, 1.0]$ using the intersection (black), subsampled split LRT (blue), and subsampled hybrid LRT (red) methods. In these simulations, we set $\theta^* = (\theta_1^*, 0, \ldots, 0)$. The x-axis is the value of $\theta_1^* = \|\theta^*\|$ for each simulation. For each dimension, the left panel satisfies $\|\theta^*\| < 0.5$, and the right panel satisfies $\|\theta^*\| > 1$. We set $\alpha = 0.10$ and $n=1000$, and we perform 1000 simulations at each value of $\|\theta^*\|$. We subsample $B = 100$ times.}
\label{fig:power_interval}
\end{figure}

We can provide a partial theoretical justification for Figure~\ref{fig:power_interval}. For one, it is possible to derive an exact formula for the power of the intersection approach. Using the fact that $n\|\bar{Y}\|^2$ follows a non-central $\chi^2$ distribution, we can write the power of the intersection method in terms of the non-central $\chi^2$ cumulative distribution function. When $d=100$ or $d=1000$, the hybrid method has no power at $\|\theta^*\| = 0$, though we would expect this case to have the highest power out of $\|\theta^*\| < 0.5$. At $d=100$ and $\|\theta^*\| = 0$, the hybrid method satisfies $\|\bar{Y}_{1,b}\| < 0.5$ in most simulations, but the test statistic is too small to reject $H_0$. At $d = 1000$ and $\|\theta^*\| = 0$, $(n/2)\|\bar{Y}_{1,b}\|^2 \sim \chi^2_d$ is approximately $d$ \citep[Lemma 2]{dasgupta2007probabilistic}. Hence $\|\bar{Y}_{1,b}\| \approx \sqrt{2}$, which means the hybrid approach selects the ``incorrect'' case of $\|\bar{Y}_{1,b}\| >1$. This test also has approximately zero power. See Section~\ref{sec:supp_intersect_test} of the supplement for more details. In addition, for any given subsample, the hybrid LRT power is provably greater than or equal to the split LRT power. This holds because the RIPR test statistic is always larger than the split test statistic when both tests use the same numerator \citep{wasserman2020universal}. The theoretical justification behind the relative power of the intersection and subsampled hybrid methods remains an open question, since the power of the latter method is not easily tractable.

\section{Discussion}

The recent development of the universal LRT provides a hypothesis testing framework that is valid in finite samples and does not rely on regularity conditions. We have explored the performance of several universal LRT variants in the simple but fundamental case of testing for the mean $\theta^*$ when data arise from a $N(\theta^*, I_d)$ distribution. We have seen that even in high dimensions or for very small $\alpha$, the ratio of the radius of the limiting subsampling universal LRT confidence set over the radius of an exact confidence set is less than 2. While the universal method tests the likelihood ratio against a dimension-independent cutoff, the universal LRT can still exhibit reasonable performance in high dimensions.  

Future research directions may focus on settings where hypothesis tests were previously intractable or only asymptotically valid. Researchers can apply the universal LRT in any setting where it is possible to write a likelihood ratio or, more generally, upper bound the maximum likelihood under the null hypothesis. This allows for the development of valid tests for the number of components in mixture models and for log-concavity of the underlying density. Additionally, we have shown proof of concept that the universal LRT can be more powerful than existing valid tests. In the Gaussian setting, this phenomenon may apply more generally across other tests of non-convex null parameter spaces. \cite{wasserman2020universal} also describe how the universal LRT can be used to test independence versus conditional independence in a Gaussian setting. Recent work by \cite{guo2020} also provides a valid test in that setting, but the relative power of these two approaches is currently unknown.

\section*{Acknowledgement}
RD is currently employed at Novartis Pharmaceuticals Corporation. This work was primarily conducted while RD was at Carnegie Mellon University. RD's research was supported by the National Science Foundation Graduate Research Fellowship Program under Grant Nos.\ DGE 1252522 and DGE 1745016. AR's research is supported by the Adobe Faculty Research Award, an ARL Large Grant, and the National Science Foundation under Grant Nos.\ DMS 2053804, DMS 1916320, and DMS (CAREER)    1945266. Any opinions, findings, and conclusions or recommendations expressed in this material are those of the authors and do not necessarily reflect the views of the National Science Foundation. This work used the Extreme Science and Engineering Discovery Environment (XSEDE) \citep{xsede}, which is supported by National Science Foundation grant number ACI-1548562. Specifically, it used the Bridges system \citep{xsedebridges}, which is supported by NSF award number ACI-1445606, at the Pittsburgh Supercomputing Center (PSC). This work made extensive use of the \texttt{R} statistical software \citep{Rcore}, as well as the \texttt{cowplot} \citep{cowplot}, \texttt{data.table} \citep{datatable}, \texttt{ggConvexHull} \citep{ggconvexhull}, \texttt{gtable} \citep{gtable}, \texttt{latex2exp} \citep{latex2exp}, \texttt{MASS} \citep{mass}, \texttt{progress} \citep{progress}, \texttt{Rcpp} \citep{Rcpp1, Rcpp2, Rcpp3}, \texttt{splancs} \citep{splancs}, \texttt{tidyverse} \citep{tidyverse}, and \texttt{tripack} \citep{tripack} packages.

\bibliographystyle{biometrika}
\bibliography{paper}

\renewcommand{\thesection}{S\arabic{section}} 
\renewcommand{\thefigure}{S\arabic{figure}}
\renewcommand{\theequation}{S\arabic{equation}}
\setcounter{section}{0}
\setcounter{figure}{0}
\setcounter{theorem}{0}
\setcounter{equation}{0}

\newpage
\section*{\textbf{Appendix}}

\section{Proofs of Theorems} \label{sec:supp_thm}

\begin{theorem} 
$C_n^{\spl}(\alpha)$ is a valid $100(1-\alpha)\%$ confidence set for $\theta^*$. As a consequence, and equivalently, when testing an arbitrary composite null $H_0: \theta^*\in\Theta_0$ versus $H_1: \theta^*\in\Theta\:\backslash\:\Theta_0$, rejecting $H_0$ when $\Theta_0 \cap C_n^{\spl}(\alpha) = \emptyset$ provides a valid level $\alpha$ hypothesis test. The latter rule reduces to rejecting if $T_n(\hat \theta_0) \geq 1/\alpha$, where $\hat \theta_0 = \arg\max_{\theta \in \Theta_0} \mathcal{L}_0(\theta)$ is the maximum likelihood estimate under $H_0$.
\end{theorem} 

\begin{proof}
This result is due to \cite{wasserman2020universal}. To prove this fact, we show that $E_{\theta^*}\left[T_n(\theta^*) \mid \D_1\right] \leq 1$. First, we use only the data in $\D_1$ to fit a parameter $\hat{\theta}_1$. Let $\mathcal{M}(\theta) = support(P_\theta)$. Let $\mathcal{M}(\theta)^{|\D_0|}$ be the Cartesian product of $|\D_0|$ sets $\mathcal{M}(\theta)$, i.e., the support of $|\D_0|$ iid observations from $P_\theta$. We see
\begin{align*}
\E_{\theta^*}&\left[T_n(\theta^*) \mid \D_1\right] = \E_{\theta^*}\left[\frac{\mathcal{L}_0(\hat{\theta}_1)}{\mathcal{L}_0(\theta^*)} \: \Bigg| \: \D_1 \right] = \int_{\mathcal{M}(\theta^*)^{|\D_0|}} \frac{\prod_{y_i \in \D_0} p_{\hat{\theta}_1}(y_i)}{\prod_{y_i \in \D_0} p_{\theta^*}(y_i)} \prod_{y_i \in \D_0} p_{\theta^*}(y_i) dy_i \\
&= \int_{\mathcal{M}(\theta^*)^{|\D_0|}} \prod_{y_i \in \D_0} p_{\hat{\theta}_1}(y_i) dy_i \leq \int_{\mathcal{M}(\hat{\theta}_1)^{|\D_0|}} \prod_{y_i \in \D_0} p_{\hat{\theta}_1}(y_i) dy_i \\
&\overset{\text{iid}}{=} \prod_{y_i \in \D_0} \left[\int_{\mathcal{M}(\hat{\theta}_1)} p_{\hat{\theta}_1}(y_i) dy_i \right] = 1.
\end{align*}
Applying Markov's inequality and the above fact,
\begin{equation*}
\P_{\theta^*}\left(\theta^* \notin C_n^{\spl}(\alpha)\right) = \P_{\theta^*}\left(T_n(\theta^*) \geq 1/\alpha \right) \leq \alpha \E_{\theta^*}[T_n(\theta^*)] = \alpha \E_{\theta^*} \left[ \E_{\theta^*}\left[ T_n(\theta^*) \mid \D_1\right] \right] \leq \alpha.
\end{equation*}
This shows that $\theta^* \in C_n^{\spl}(\alpha)$ with probability at least $1-\alpha$. Alternatively, suppose we want to test $H_0: \theta^*\in\Theta_0$ versus $H_1: \theta^*\in\Theta\:\backslash\:\Theta_0$. Suppose we reject $H_0$ when $\Theta_0 \cap C_n^{\spl}(\alpha) = \emptyset$. Under $H_0$,
\begin{equation*}
\P_{\theta^*}\left\{\Theta_0 \cap C_n^{\spl}(\alpha) = \emptyset \right\} \leq \P_{\theta^*}\left\{\theta^* \notin \Theta_0 \cap C_n^{\spl}(\alpha)\right\} = \P_{\theta^*}\left\{\theta^* \notin C_n^{\spl}(\alpha)\right\} \leq \alpha. 
\end{equation*}
Hence, rejecting $H_0$ when $\Theta_0 \cap C_n^{\spl}(\alpha) = \emptyset$ provides a valid level $\alpha$ hypothesis test. 
\end{proof}

Before proving Theorem~\ref{thm:limiting_ratio}, we establish Lemma~\ref{lemma:normal_limit} and Lemma~\ref{lemma:expectation}. We draw heavily on finite population central limit theorem results from \cite{hajek1960limiting} and \cite{li2017general}. Lemma~\ref{lemma:normal_limit} combines key results from these two papers and adapts them to our setting.

\begin{lemma} \label{lemma:normal_limit}
Let $(\D_n)_{n\in 2\mathbb{N}}$ be a sequence of datasets, where $\D_n = \{Y_{n1}, \ldots, Y_{nn}\}$ and each $Y_{ni}$ is an independent observation from $N(\theta^*, I_d)$. Let $\D_{0,n}$ be a sample of $n/2$ observations from $\D_n$. Define $\bar{Y}_n = \frac{1}{n} \sum_{i=1}^n Y_{ni}$ and $\bar{Y}_{0,n} = \frac{2}{n} \sum_{Y_{ni} \in \D_{0,n}} Y_{ni}$. As $n\to\infty$, $\sqrt{n}(\bar{Y}_{0,n} - \bar{Y}_n)$ converges in distribution to $N(0, I_d)$ with probability 1.
\end{lemma} 

\begin{proof}
We show a highlight of the proof of Lemma~\ref{lemma:normal_limit}, in five steps. 

\textbf{Step 1 \citep{hajek1960limiting}}: Show that simple random sampling and Poisson sampling approaches produce the same limiting distributions.

In the notation of \cite{hajek1960limiting}, suppose we have an infinite sequence of simple random sample experiments indexed by $\nu$. Experiment $\nu$ draws a simple random sample of size $n_\nu$ from a population of size $N_\nu$ given by $\{Y_{\nu 1}, \ldots, Y_{\nu N_\nu}\}$. We assume that $n_\nu \to \infty$ and $N_\nu - n_\nu \to \infty$. In the simple random sampling set-up, a subset $s_k$ of indices $\{1, \ldots, N_\nu\}$ is chosen with probability
\[P(s_k) = \begin{cases} 
      \binom{N_\nu}{n_\nu}^{-1} & \text{if } |s_k| = n_\nu  \\
      0 & \text{otherwise}
   \end{cases}.
\]
In contrast, in a Poisson sampling approach with mean sample size $n_\nu$, a subset $s_k$ is chosen with probability 
$$P(s_k) = \left(\frac{n_\nu}{N_\nu}\right)^k \left(1 - \frac{n_\nu}{N_\nu}\right)^{N_\nu - k}.$$

We say that each experiment produces a simple random sample (SRS) $s_n$ and a Poisson sample $s_k$ such that $s_n \subseteq s_k$ or $s_k \subseteq s_n$. To construct these samples, we take two steps:
\begin{enumerate}
\item[(i)] Draw $k \sim \text{Binom}(N_\nu, n_\nu / N_\nu)$.
\item[(ii)] If $k = n$, choose SRS $s_n$, and set $s_k = s_n$. \\
If $k > n$, choose SRS $s_k$, and then let $s_n$ be an SRS of size $n$ from $s_k$. \\
If $k < n$, choose SRS $s_n$, and then let $s_k$ be an SRS of size $k$ from $s_n$.
\end{enumerate}
Using the two samples, we define two random variables:
$$\eta_\nu = \sum_{i\in s_n} (Y_{\nu i} - \bar{Y}_\nu) \qquad \text{and} \qquad \eta_\nu^* = \sum_{i\in s_k} (Y_{\nu i} - \bar{Y}_\nu).$$
We can show that the variance of $\eta_\nu^*$ is $$D\eta_\nu^* = \V(\eta_\nu^*) = \frac{n_\nu}{N_\nu} \left(1 - \frac{n_\nu}{N_\nu}\right) \sum_{i=1}^{N_\nu} (Y_{\nu i} - \bar{Y}_\nu)^2~.$$
Under the assumption that $n_\nu \to \infty$ and $N - n_\nu \to \infty$, we can then show that 
\begin{equation} \label{eq:same_limit_dist}
\lim_{\nu\to\infty} \frac{\E[(\eta_\nu - \eta_\nu^*)^2]}{D\eta_\nu^*} = 0.
\end{equation}
Remark 2.1 of \cite{hajek1960limiting} states that
(\ref{eq:same_limit_dist}) implies that the limiting distributions of $\eta_\nu / \sqrt{D\eta_\nu^*}$ and $\eta^*_\nu / \sqrt{D\eta_\nu^*}$ are the same if they exist, and they exist under the same conditions. To see this, we use Chebyshev's inequality. For $\epsilon > 0$,
$$\P\left(\left|\frac{\eta_\nu}{\sqrt{D\eta_\nu^*}} - \frac{\eta_\nu^*}{\sqrt{D\eta_\nu^*}} \right| > \epsilon \right) \leq \frac{1}{\epsilon^2} \V\left(\frac{\eta_\nu - \eta^*_\nu}{\sqrt{D\eta_\nu^*}}\right) = \frac{1}{\epsilon^2} \frac{\E[(\eta_\nu - \eta_\nu^*)^2]}{D\eta_\nu^*} \overset{\nu\to\infty}{\to} 0.$$
This means that $\left|\eta_\nu / \sqrt{D\eta_\nu^*} - \eta_\nu^* / \sqrt{D\eta_\nu^*} \right| \overset{p}{\to} 0.$ Under this condition, for any distribution $W$, $\eta_\nu / \sqrt{D\eta_\nu^*} \rightsquigarrow W$ if and only if $\eta_\nu^* / \sqrt{D\eta_\nu^*} \rightsquigarrow W$.

Since $\eta_\nu^*$ is a sum of independent random variables, it will be easier to work with $\eta^*_\nu / D\eta_\nu^*$ than to work with $\eta_\nu / D\eta_\nu^*$. 

\textbf{Step 2 \citep{hajek1960limiting}}: Find conditions such that $\eta_\nu / \sqrt{D\eta_\nu^*} \rightsquigarrow N(0, 1)$. (We can think of $\eta_\nu$ as $(n/2)(\bar{Y}_{0,n} - \bar{Y}_n)$ and $D\eta_\nu^*$ as $\V(\sum_{i=1}^n B_i(Y_{ni} - \bar{Y}_n))$ for $B_i \overset{\text{iid}}{\sim} \text{Bernoulli}(1/2)$.)

Theorem 3.1 in \cite{hajek1960limiting} is the key result for asymptotic normality. We present an intermediate result from the proof of Theorem 3.1.

{\em
Let $\xi_\nu = \sum_{i\in s_{n,\nu}} Y_{\nu, i}$. (So $\eta_\nu = \xi_\nu - n_\nu \bar{Y}_\nu$.) Let $D\xi_\nu$ be the variance of $\xi_\nu$. Let $S_{\nu\tau}$ be the subset of $S_\nu = \{1, \ldots, N_\nu\}$ on which the inequality $$|Y_{\nu i} - \bar{Y}_\nu| > \tau \sqrt{D\xi_\nu}$$ holds. Suppose that $n_\nu \to \infty$ and $N_\nu - n_\nu \to \infty$. If 
\begin{equation} \label{eq:asymp_normal_cond}
\lim_{\nu\to\infty} \frac{\sum_{i\in S_{\nu\tau}} (Y_{\nu i} - \bar{Y}_\nu)^2}{\sum_{i\in S_{\nu}} (Y_{\nu i} - \bar{Y}_\nu)^2} = 0 \quad \text{for any $\tau > 0$},
\end{equation}
then $\eta_\nu / \sqrt{D\eta_\nu^*} \rightsquigarrow N(0, 1)$.
}

We will show that $\eta_\nu^* / \sqrt{D\eta_\nu^*} \rightsquigarrow N(0, 1)$, and then we can appeal to Step 1's result. $\eta_\nu^*$ is the centered sum of the Poisson sampling terms. We can write $\eta_\nu^*$ as $$\eta_\nu^* = \sum_{i=1}^{N_\nu} \zeta_{\nu i}, \: \text{where } \: \zeta_{\nu i} =
     \begin{cases}
       Y_{\nu i} - \bar{Y}_\nu & \text{with probabilty } n_\nu / N_\nu \\
       0 & \text{with probabilty } 1 - n_\nu / N_\nu
     \end{cases}.$$
In this setting, Lindeberg's condition for $\eta_\nu^* / \sqrt{D\eta_\nu^*} \rightsquigarrow N(0, 1)$ is for all $\tau > 0$,
$$\lim_{\nu \to \infty} \frac{1}{D\eta_\nu^*} \sum_{i = 1}^{N_\nu} \E\left[(\zeta_{\nu i} - \E[\zeta_{\nu i}])^2 \cdot \one\left(|\zeta_{\nu i} - \E[\zeta_{\nu i}| > \tau\sqrt{D\eta_\nu^*} \right) \right] = 0.$$ We can show that (\ref{eq:asymp_normal_cond}) implies that the Lindeberg condition is satisfied. Since Step 1 implies that the limiting distribution of $\eta_\nu / \sqrt{D\eta_\nu^*}$ must be the same as the limiting distribution of $\eta_\nu^* / \sqrt{D\eta_\nu^*}$, we conclude that $\eta_\nu / \sqrt{D\eta_\nu^*} \rightsquigarrow N(0, 1).$

\textbf{Step 3:} If $d=1$, show that $\eta_\nu / \sqrt{D\eta_\nu^*} \rightsquigarrow N(0, 1)$ implies $\sqrt{n}(\bar{Y}_{0,n} - \bar{Y}_n) \rightsquigarrow N(0, 1)$.

This is mostly a matter of adapting Step 2's result to our setting. When $n_\nu / N_\nu = 1/2$, $\eta_\nu$ is the same random variable as $(n/2)(\bar{Y}_{0,n} - \bar{Y}_n)$. Using the formula for $D\eta_\nu^*$,
$$\frac{\sqrt{n}(\bar{Y}_{0,n} - \bar{Y}_n)}{\sqrt{\frac{1}{n} \sum_{i=1}^n (Y_{ni} - \bar{Y}_n)^2}} = \frac{(n/2)(\bar{Y}_{0,n} - \bar{Y}_n)}{\sqrt{\frac{1}{4} \sum_{i=1}^n (Y_{ni} - \bar{Y}_n)^2}} \overset{d}{=} \frac{\eta_\nu}{\sqrt{D\eta_\nu^*}} \rightsquigarrow N(0,1).$$
In addition, $\sqrt{\frac{1}{n} \sum_{i=1}^n (Y_{ni} - \bar{Y}_n)^2} / \sqrt{\V(Y_{ni})} \overset{p}{\to} 1.$ By Slutsky's Theorem, $\sqrt{n}(\bar{Y}_{0,n} - \bar{Y}_n) \rightsquigarrow N(0, 1)$.

\textbf{Step 4 \citep{li2017general}}: If $Y_{n1}, \ldots, Y_{nn} \sim N(\theta^*, 1)$, show that the condition of Step 2 is satisfied with probability 1.

These results come from page 2 of the appendix of \cite{li2017general}. The authors show that if the $Y_{ni}$s are iid draws from a superpopulation with $2 + \epsilon$ ($\epsilon > 0$) absolute moments and nonzero variance, then $(1/n) \max_{1\leq i \leq n} (Y_{ni} - \bar{Y}_n)^2 \equiv  m_n / n \to 0$ with probability 1. Furthermore, they show that $m_n / n \to 0$ implies their condition (A2), which is a rewriting of \cite{hajek1960limiting}'s condition (\ref{eq:asymp_normal_cond}).

Since $N(\theta^*, 1)$ satisfies the superpopulation conditions, condition (\ref{eq:asymp_normal_cond}) is satisfied with probability 1. Then following Steps 2 and 3, $\sqrt{n}(\bar{Y}_{0,n} - \bar{Y}_n) \rightsquigarrow N(0,1)$. 

\textbf{Step 5 \citep{hajek1960limiting}}: Extend results to $d > 1$.

In $d$ dimensions, suppose $Y_{n1}, \ldots, Y_{nn} \sim N(\theta^*, I_d)$. Remark 3.2 of \cite{hajek1960limiting} notes that we can user the Cram\'er-Wold device to extend the results to the multivariate case. Let $Z = (Z^{(1)}, \ldots, Z^{(d)})$ represent the $N(0, I_d)$ distribution. Then for each component, $Z^{(j)} \sim N(0, 1)$. By the Cram\'er-Wold device, we can say that $\sqrt{n}(\bar{Y}_{0,n} - \bar{Y}_n) \rightsquigarrow Z$ if and only if for any $\lambda \in \R^d$, $\sum_{j=1}^d \lambda^{(j)} \sqrt{n} (\bar{Y}_{0,n}^{(j)} - \bar{Y}_n^{(j)}) \rightsquigarrow \sum_{j=1}^d \lambda^{(j)} Z^{(j)}.$

For any dimension $j$, we can think of $Y_{n1}^{(j)}, \ldots, Y_{nn}^{(j)}$ as draws from a $N(\theta^{*(j)}, 1)$ superpopulation. So the superpopulation conditions from Step 4 are satisfied, which means $\sqrt{n}(\bar{Y}_{0,n}^{(j)} - \bar{Y}_n^{(j)}) \rightsquigarrow Z^{(j)}$. We conclude that $\sqrt{n}(\bar{Y}_{0,n} - \bar{Y}_n) \rightsquigarrow N(0, I_d)$.
\end{proof}

\begin{lemma} \label{lemma:expectation}
Assume $(\D_n)_{n\in 2\mathbb{N}}$ is a sequence of data sets such that $\D_n = \{Y_{n1}, Y_{n2}, \ldots, Y_{nn}\}$ with observations $Y_{nj}\overset{\text{iid}}{\sim} N(\theta^*, I_d)$. Let $\D_{0,n}$ be a sample of $n/2$ observations from $\D_n$. Define $\bar{Y}_n = (1/n)\sum_{i=1}^n Y_{ni}$ and $\bar{Y}_{0,n} = (2/n) \sum_{Y_{ni} \in \D_{0,n}} Y_{ni}$. Let $c>0$, and let $(\theta_n)$ be a sequence that satisfies $\|\bar{Y}_n - \theta_n\| \leq c/\sqrt{n}$ for all $n$. Define $X_n \equiv \sqrt{n}(\bar{Y}_{0,n} - \bar{Y}_n)$. Let $Z$ denote a $N(0, I_d)$ random variable. Then
{\footnotesize
\begin{align*}
\E\left[\exp\left(-\frac{3}{4} X_n^T X_n + \frac{\sqrt{n}}{2} X_n^T\left(\bar{Y}_n - \theta_n \right) \right) \mid \D_n \right] - \E\left[\exp\left(-\frac{3}{4} Z^T Z + \frac{\sqrt{n}}{2} Z^T\left(\bar{Y}_n - \theta_n \right) \right) \mid \D_n \right] = o_P(1).
\end{align*}
}
\end{lemma}

\begin{proof}
Since $(\theta_n)$ is chosen such that $\|\bar{Y}_n - \theta_n\| \leq c/\sqrt{n}$, we can re-write $\theta_n = \bar{Y}_n + (c/\sqrt{n}) v_n$, where $v_n \in \R^d$ satisfies $\|v_n\| \leq 1$ for all $n$.

Define a function $f$ by $$f(x_n, v_n) \equiv \exp\left(-\frac{3}{4} x_n^T x_n - \frac{c}{2} x_n^T v_n \right).$$ $f$ is clearly a continuous function. We can also show that $f$ is bounded. Define $$g(x_n, v_n) \equiv -\frac{3}{4} x_n^T x_n - \frac{c}{2} x_n^T v_n$$ so that $f(x_n, v_n) = \exp(g(x_n, v_n))$. We can see that $$\frac{\partial}{\partial x_n} g(x_n, v_n) =  -\frac{3}{2} x_n - \frac{c}{2} v_n \overset{\text{set}}{=} \vec{0}$$ is solved by $x_n = -(c/3) v_n$. Since $g(x_n, v_n)$ is concave in $x_n$, $g(x_n, v_n)$ is maximized at $x_n = -(c/3)v_n$ for any $v_n$. Since $f(x_n, v_n) = \exp(g(x_n, v_n))$, $f(x_n, v_n)$ is also maximized at this value of $x_n$ for any $v_n$. Under the assumption that $\|v_n\| \leq 1$, we see
\begin{align*}
f(x_n, v_n) &\leq \exp\left(-\frac{3}{4}\left(-\frac{c}{3}\right)^2 v_n^T v_n - \frac{c}{2} \left(-\frac{c}{3}\right) v_n^T v_n \right) \\
&= \exp\left( -\frac{c^2}{12} \|v_n\|^2 + \frac{c^2}{6} \|v_n\|^2 \right) \\
&\leq \exp\left(\frac{c^2}{12} \right).
\end{align*}
Thus, $f(x_n, v_n)$ is a continuous and bounded function.

The claim of Lemma~\ref{lemma:expectation} is equivalent to $\E[f(X_n, v_n) \mid \D_n] - \E[f(Z, v_n) \mid \D_n] = o_P(1)$. The Portmanteau Theorem provides several equivalent definitions of convergence in distribution, including that $X_n \rightsquigarrow Z$ if and only if $\E[h(X_n)] \to \E[h(Z)]$ for every continuous, bounded function $h$.  We prove the result on $f(X_n, v_n)$ by modifying the \cite{van2000asymptotic}, Chapter 2, proof of this Portmanteau Theorem result.

Let $\gamma > 0$. Fix $\epsilon > 0$ such that 
\begin{equation}
\epsilon < \gamma \: / \: (3 + 3\exp(c^2/12))~. \label{eq:epsilon_gamma}
\end{equation}
Choose a large enough compact rectangle $I$ such that 
\begin{equation}
\P(Z\notin I) < \epsilon~. \label{eq:Z_epsilon}
\end{equation}

Let $\mathcal{B}_1(0)$ be the $d$-dimensional ball of radius 1 centered at 0. By construction, each $v_n \in \mathcal{B}_1(0)$. Since $f$ is continuous and $I \times \mathcal{B}_1(0)$ is compact, $f(x_n, v_n)$ is uniformly continuous on $I \times \mathcal{B}_1(0)$. We can thus partition $I \times \mathcal{B}_1(0)$ into $J$ compact regions $I_j \times V_j$ where $I \times \mathcal{B}_1(0) = \cup_{j=1}^J (I_j \times V_j)$ such that for any $j$ and for any $(x_{n1}, v_{n1}), (x_{n2}, v_{n2}) \in I_j \times V_j$, $|f(x_{n1}, v_{n1}) - f(x_{n2}, v_{n2})| < \epsilon$. (For instance the $I_j$ regions may be rectangles and the $V_j$ regions may be rectangles truncated at the boundaries of $\mathcal{B}_1(0)$. These rectangular regions may be appropriately sized such that within a region $I_j \times V_j$, $d((x_{n1}, v_{n1}), (x_{n2}, v_{n2}))$ is small enough that $|f(x_{n1}, v_{n1}) - f(x_{n2}, v_{n2})| < \epsilon$.)

Select a point $(x_j', v_j')$ from each $I_j\times V_j$. Define $$f_\epsilon(x, v) = \sum_{j=1}^J f(x_j', v_j') \one((x, v) \in I_j\times V_j)~.$$ For a given sample $\D_n$, we note that there are $\binom{n}{n/2}$ possible values of $X_n$, since there are $\binom{n}{n/2}$ possible values of $\bar{Y}_{0,n}$. We denote the sum over all possible values of $X_n$ as $\sum_{X_n}$~.

Note that
\begin{align*}
&\left|\E[f(X_n, v_n) \mid \D_n] - \E[f_\epsilon(X_n, v_n) \mid \D_n] \right| \\
&= \left|\binom{n}{n/2}^{-1} \sum_{X_n} f(X_n, v_n) -  \binom{n}{n/2}^{-1} \sum_{X_n} f_\epsilon(X_n, v_n)\right| \\
&= \Bigg|\binom{n}{n/2}^{-1} \sum_{X_n} \big[(f(X_n, v_n) - f_\epsilon(X_n, v_n)) \one(X_n \in I) +  (f(X_n, v_n) - f_\epsilon(X_n, v_n)) \one(X_n \notin I) \big] \Bigg| \\
&\leq \binom{n}{n/2}^{-1} \sum_{X_n} |f(X_n, v_n) - f_\epsilon(X_n, v_n)| \one(X_n \in I) + \\
&\qquad \binom{n}{n/2}^{-1} \sum_{X_n} |f(X_n, v_n) - f_\epsilon(X_n, v_n)| \one(X_n \notin I) \\
&= \binom{n}{n/2}^{-1} \sum_{X_n} |f(X_n, v_n) - f_\epsilon(X_n, v_n)| \one(X_n \in I, v_n \in \mathcal{B}_1(0)) + \\
&\qquad \binom{n}{n/2}^{-1} \sum_{X_n} |f(X_n, v_n) - f_\epsilon(X_n, v_n)| \one(X_n \notin I) \\
&< \binom{n}{n/2}^{-1} \sum_{X_n} \epsilon + \binom{n}{n/2}^{-1} \sum_{X_n} |f(X_n, v_n)| \one(X_n \notin I) \\
&\leq \epsilon + \exp\left(c^2 / 12 \right) \P(X_n \notin I \mid \D_n)~. \stepcounter{equation}\tag{\theequation}\label{eq:conv_prob_pt1}
\end{align*}

Similarly, we show that
\begin{align*}
&\Big|\E[f(Z, v_n) \mid \D_n] - \E[f_\epsilon(Z, v_n) \mid \D_n] \Big| \\
&= \Big| \E\big[(f(Z, v_n) - f_\epsilon(Z, v_n)) \one(Z \in I) + (f(Z, v_n) - f_\epsilon(Z, v_n)) \one(Z \notin I) \mid \D_n \big] \Big| \\
&\leq \E\left[\Big|f(Z, v_n) - f_\epsilon(Z, v_n)\Big| \one(Z \in I) \mid \D_n \right] + \E\left[\Big|f(Z, v_n) - f_\epsilon(Z, v_n)\Big| \one(Z \notin I) \mid \D_n \right] \\
&= \E\left[\Big|f(Z, v_n) - f_\epsilon(Z, v_n)\Big| \one(Z \in I, v_n \in \mathcal{B}_1(0)) \mid \D_n \right] + \E\left[\Big|f(Z, v_n) - f_\epsilon(Z, v_n)\Big| \one(Z \notin I) \mid \D_n \right] \\
&< \epsilon + \exp(c^2 / 12) \P(Z\notin I \mid \D_n) \\
&= \epsilon + \exp(c^2 / 12) \P(Z\notin I) \\
&< \epsilon + \epsilon\exp(c^2 / 12)~. \stepcounter{equation}\tag{\theequation}\label{eq:conv_prob_pt2}
\end{align*}

In addition, we see that
\begin{align*}
&\big|\E\left[f_\epsilon(X_n, v_n) \mid \D_n \right] - \E\left[f_\epsilon(Z, v_n) \mid \D_n \right] \big| \\
&= \left|\binom{n}{n/2}^{-1} \sum_{X_n} f_\epsilon(X_n, v_n) - \E[f_\epsilon(Z, v_n)] \right| \\
&= \left|\binom{n}{n/2}^{-1} \sum_{X_n} \sum_{j=1}^J f(x_j', v_j') \one((X_n, v_n) \in I_j \times V_j) - \sum_{j=1}^J f(x_j', v_j') \P(Z \in I_j) \one(v_n \in V_j) \right| \\
&\leq \sum_{j=1}^J \left| \binom{n}{n/2}^{-1} \sum_{X_n} f(x_j', v_j') \one(X_n \in I_j) \one(v_n \in V_j) - f(x_j', v_j') \P(Z \in I_j) \one(v_n \in V_j) \right| \\
&\leq \sum_{j=1}^J \left| \binom{n}{n/2}^{-1} \sum_{X_n} f(x_j', v_j') \one(X_n \in I_j) - f(x_j', v_j') \P(Z \in I_j) \right| \\
&= \sum_{j=1}^J \left|f(x_j', v_j') \left[\binom{n}{n/2}^{-1} \sum_{X_n} \one(X_n \in I_j) - \P(Z\in I_j) \right] \right| \\
&\leq \sum_{j=1}^J \left|\P(X_n \in I_j \mid \D_n) - \P(Z \in I_j) \right| \times \left| f(x_j', v_j') \right|~. \stepcounter{equation}\tag{\theequation}\label{eq:conv_prob_pt3}
\end{align*}

For the sequence of datasets $(\D_n)_{n\in 2\mathbb{N}}$, Lemma~\ref{lemma:normal_limit} establishes that $X_n \rightsquigarrow N(0, I_d)$ with probability 1. This tells us that with probability 1 over the randomness in sequences $(\D_n)_{n\in 2\mathbb{N}}$, $\lim_{n\to\infty} \P(X_n \in I \mid \D_n) = \P(Z\in I)$. Since almost sure convergence implies convergence in probability, for any $\delta > 0$, 
\begin{align} 
\lim_{n\to\infty} \P\big(|\P(X_n \in I \mid \D_n) - \P(Z\in I)| > \delta\big) &= 0 \label{eq:conv_in_prob_I} \\
\text{and } \lim_{n\to\infty} \P\big(|\P(X_n \in I_j \mid \D_n) - \P(Z\in I_j)| > \delta\big) &= 0 \text{ for $1\leq j\leq J$}~. \label{eq:conv_in_prob_Ij}
\end{align}
The outer probability is over the randomness in the sequences $(\D_n)_{n\in 2\mathbb{N}}$.

Now we see
\begin{align*}
&\lim_{n\to\infty} \P\big( \big| \E[f(X_n, v_n) \mid \D_n] - \E[f(Z, v_n) \mid \D_n] \big| > \gamma) \\
&\leq \lim_{n\to\infty} \P\big( \big| \E[f(X_n, v_n) \mid \D_n] - \E[f_\epsilon(X_n, v_n) \mid \D_n] \big| + \\
&\hspace*{5em} \big| \E[f_\epsilon(X_n, v_n) \mid \D_n] - \E[f_\epsilon(Z, v_n) \mid \D_n] \big| + \\
&\hspace*{5em} \big| \E[f_\epsilon(Z, v_n) \mid \D_n] - \E[f(Z, v_n) \mid \D_n] \big| > \gamma \big) \\
&\leq \lim_{n\to\infty} \P \big(\big| \E[f(X_n, v_n) \mid \D_n] - \E[f_\epsilon(X_n, v_n) \mid \D_n] \big| > \gamma/3 \big) + \\
&\hspace*{2em} \lim_{n\to\infty} \P \big(\big| \E[f_\epsilon(X_n, v_n) \mid \D_n] - \E[f_\epsilon(Z, v_n) \mid \D_n] \big| > \gamma/3 \big) + \\
&\hspace*{2em} \lim_{n\to\infty} \P \big(\big| \E[f_\epsilon(Z, v_n) \mid \D_n] - \E[f(Z, v_n) \mid \D_n] \big| > \gamma/3 \big) \\
&\leq \lim_{n\to\infty} \P \big(\epsilon + \exp(c^2/12) \P(X_n \notin I \mid \D_n) > \gamma/3\big) + \lim_{n\to\infty} \P \big( \epsilon + \epsilon \exp(c^2/12) > \gamma/3) + \\
&\hspace*{2em} \lim_{n\to\infty} \P \left(\sum_{j=1}^J \big| \P(X_n \in I_j \mid \D_n) - \P(Z\in I_j) \big| \times |f(x_j', v_j')| > \gamma/3  \right) \text{ by (\ref{eq:conv_prob_pt1}), (\ref{eq:conv_prob_pt2}), and (\ref{eq:conv_prob_pt3})} \\
&= \lim_{n\to\infty} \P \big(\epsilon + \exp(c^2/12) \P(X_n \notin I \mid \D_n) > \gamma/3\big) + \\
&\hspace*{2em} \lim_{n\to\infty} \P \left(\sum_{j=1}^J \big| \P(X_n \in I_j \mid \D_n) - \P(Z\in I_j) \big| \times |f(x_j', v_j')| > \gamma/3  \right) \text{ by (\ref{eq:epsilon_gamma})} \\
&\leq \lim_{n\to\infty} \P\big(\epsilon + \exp(c^2/12)\left(\P(X_n \notin I\mid \D_n) - \P(Z\notin I)\right) > \gamma/3 - \exp(c^2/12) \P(Z\notin I) \big) + \\
&\hspace*{2em} \lim_{n\to\infty} \sum_{j=1}^J \P\left( \big| \P(X_n\in I_j \mid \D_n) - \P(Z\in I_j)\big|  > (\gamma/3) |f(x_j', v_j')|^{-1} \right) \\
&\leq \lim_{n\to\infty} \P\left(\epsilon + \exp(c^2/12) (\P(X_n \notin I\mid \D_n) - \P(Z\notin I)) > \gamma/3 - \epsilon \exp(c^2/12) \right)\text{ by (\ref{eq:Z_epsilon}) and (\ref{eq:conv_in_prob_Ij})} \\
&= \lim_{n\to\infty} \P\left(\P(X_n\notin I\mid \D_n) - \P(Z\notin I) > \frac{\gamma - 3\epsilon - 3\epsilon \exp(c^2/12)}{3\exp(c^2/12)} \right) \\
&= 0 \text{ by (\ref{eq:epsilon_gamma}) and (\ref{eq:conv_in_prob_I})}.
\end{align*}

We have shown that for arbitrary $\gamma > 0$,
\begin{align*}
\lim_{n\to\infty} \P\big( \big| \E[f(X_n, v_n) \mid \D_n] - \E[f(Z, v_n) \mid \D_n] \big| > \gamma) = 0.
\end{align*}
We conclude that $\E[f(X_n, v_n) \mid \D_n] - \E[f(Z, v_n) \mid \D_n] = o_P(1)$. 
\end{proof}

\begin{theorem}
Assume we have a sequence of datasets $(\D_{n})_{n\in 2\mathbb{N}}$, where $\D_n = \{Y_{n1}, \ldots, Y_{nn}\}$ and each $Y_{ni}$ is an independent observation from $N(\theta^*, I_d)$. Let $\D_{0,n}$ be a sample of $n/2$ observations from $\D_n$, and let $\D_{1,n} = \D_n \backslash \D_{0,n}$. Define $\bar{Y}_n = (1/n)\sum_{i=1}^n Y_{ni}$, $\bar{Y}_{0,n} = (2/n) \sum_{Y_{ni} \in \D_{0,n}} Y_{ni}$, and $\bar{Y}_{1,n} = (2/n) \sum_{Y_{ni} \in \D_{1,n}} Y_{ni}$. Let $c > 0$, and let $(\theta_n)$ be a sequence that satisfies $\|\bar{Y}_n - \theta_n\| \leq c/\sqrt{n}$ for all $n$. Then
\begin{align*}
\E\{T_n(\theta_n) \mid \D_n\} \: / \: \left\{\exp\left(\frac{3n}{10} \|\bar{Y}_n - \theta_n\|^2 \right) \left(\frac{2}{5}\right)^{d/2} \right\} &= 1 + o_P(1).
\end{align*}
\end{theorem} 

\begin{proof}
Define $X_n \equiv \sqrt{n}(\bar{Y}_{0,n} - \bar{Y}_n)$ and let $Z \sim N(0,I_d)$. In addition, define $\mu_n \equiv (\sqrt{n} / 5) (\bar{Y}_n - \theta_n)$ and $\Omega \equiv (2/5)I_d$. Then

{\small
\begin{align*}
\E&[T_n(\theta_n) \mid \D_n] \: / \: \left\{\exp\left(\frac{3n}{10} \|\bar{Y}_n - \theta_n\|^2 \right) \left(\frac{2}{5}\right)^{d/2} \right\}  \\
&= \E\left[\exp\left(-\frac{n}{4}\|\bar{Y}_{0,n} - \bar{Y}_{1,n}\|^2 + \frac{n}{4}\|\bar{Y}_{0,n} -\theta_n\|^2 \right) \mid \D_n \right] \: / \: \left\{\exp\left(\frac{3n}{10} \|\bar{Y}_n - \theta_n\|^2 \right) \left(\frac{2}{5}\right)^{d/2} \right\}  \\
&= \E\left[\exp\left(-\frac{n}{4}\|2\bar{Y}_{0,n} - 2\bar{Y}_n\|^2 + \frac{n}{4}\|\bar{Y}_{0,n} -\theta_n\|^2 \right) \mid \D_n \right] \exp\left(-\frac{3n}{10} \|\bar{Y}_n - \theta_n\|^2 \right) \left(\frac{2}{5}\right)^{-d/2} \\
&= \E\left[\exp\left(-n \|\bar{Y}_{0,n} - \bar{Y}_n\|^2 + \frac{n}{4}\|\bar{Y}_{0,n} - \bar{Y}_n + \bar{Y}_n -\theta_n\|^2 \right) \mid \D_n \right] \exp\left(-\frac{3n}{10} \|\bar{Y}_n - \theta_n\|^2 \right) \left(\frac{2}{5}\right)^{-d/2}  \\
&= \E\left[\exp\left( -\frac{3n}{4} \|\bar{Y}_{0,n} - \bar{Y}_n \|^2 + \frac{n}{2}(\bar{Y}_{0,n} - \bar{Y}_n)^T (\bar{Y}_n - \theta_n) + \frac{n}{4} \|\bar{Y}_n - \theta_n\|^2 \right) \mid \D_n \right] \times \\
&\qquad \exp\left(-\frac{3n}{10} \|\bar{Y}_n - \theta_n\|^2 \right) \left(\frac{2}{5}\right)^{-d/2}  \\
&= \E\left[\exp\left(-\frac{3}{4} X_n^T X_n + \frac{\sqrt{n}}{2} X_n^T\left(\bar{Y}_n - \theta_n \right) \right) \mid \D_n \right] \exp\left(-\frac{n}{20} \|\bar{Y}_n - \theta_n\|^2 \right) \left(\frac{2}{5}\right)^{-d/2}  \\
&= \E\left[\exp\left(-\frac{3}{4} X_n^T X_n + \frac{\sqrt{n}}{2} X_n^T\left(\bar{Y}_n - \theta_n \right) \right) \mid \D_n \right] / \E\left[\exp\left(-\frac{3}{4} Z^T Z + \frac{\sqrt{n}}{2} Z^T\left(\bar{Y}_n - \theta_n \right) \right) \mid \D_n \right] \stepcounter{equation}\tag{\theequation}\label{eq:exp_z} \\
&= 1 + o_P(1). \stepcounter{equation}\tag{\theequation}\label{eq:divide_o_1}
\end{align*}
}

Step (\ref{eq:exp_z}) holds because
{\small
\begin{align*}
\E&\left[\exp\left(-\frac{3}{4} Z^T Z + \frac{\sqrt{n}}{2} Z^T\left(\bar{Y}_n - \theta_n \right) \right) \mid \D_n \right] \\
&= \int_{\R^d} \Bigg[ \frac{1}{(2\pi)^{d/2} |I_d|^{1/2}} \exp\left(-\frac{1}{2}z^T z \right) \exp\left(-\frac{3}{4} z^T z + \frac{\sqrt{n}}{2} z^T\left(\bar{Y}_n - \theta_n\right)\right)\Bigg] dz \\
&= \int_{\R^d} \Bigg[ \frac{1}{(2\pi)^{d/2}} \exp\left(-\frac{5}{4}z^T z + \frac{\sqrt{n}}{2} z^T \left(\bar{Y}_n - \theta_n\right) \right)\Bigg] dz  \\
&= |\Omega|^{1/2} \int_{\R^d} \Bigg[ \frac{1}{(2\pi)^{d/2} |\Omega|^{1/2}} \exp\left(-\frac{1}{2} (z - \mu_n)^T \Omega^{-1} (z - \mu_n) + \frac{n}{20} \|\bar{Y}_n - \theta_n\|^2  \right)\Bigg] dz \stepcounter{equation}\tag{\theequation}\label{eq:exp_plugin} \\
&= \exp\left(\frac{n}{20} \|\bar{Y}_n - \theta_n\|^2 \right) |\Omega|^{1/2}\\
&= \exp\left(\frac{n}{20} \|\bar{Y}_n - \theta_n\|^2 \right) \left( \frac{2}{5} \right)^{d/2}. 
\end{align*}
}

Step (\ref{eq:exp_plugin}) uses the following equality:
{\small
\begin{align*}
-\frac{5}{4} &z^T z + \frac{\sqrt{n}}{2} z^T (\bar{Y}_n - \theta_n) \\
&= -\frac{5}{4} \left[z^T z - \frac{2\sqrt{n}}{5} z^T (\bar{Y}_n - \theta_n) + \frac{n}{25} (\bar{Y}_n - \theta_n)^T (\bar{Y}_n - \theta_n) - \frac{n}{25} (\bar{Y}_n - \theta_n)^T (\bar{Y}_n - \theta_n) \right] \\
&= -\frac{5}{4} \left(z - \frac{\sqrt{n}}{5} (\bar{Y}_n - \theta_n) \right)^T \left(z - \frac{\sqrt{n}}{5} (\bar{Y}_n - \theta_n) \right) + \frac{n}{20} \|\bar{Y}_n - \theta_n\|^2 \\
&= -\frac{1}{2} \left(z - \frac{\sqrt{n}}{5} (\bar{Y}_n - \theta_n) \right)^T \left(\frac{5}{2} I_d \right) \left(z - \frac{\sqrt{n}}{5} (\bar{Y}_n - \theta_n) \right) + \frac{n}{20} \|\bar{Y}_n - \theta_n\|^2 \\
&= -\frac{1}{2} (z-\mu_n)^T \Omega^{-1} (z-\mu_n) + \frac{n}{20} \|\bar{Y}_n - \theta_n\|^2~.
\end{align*}
}

To justify step (\ref{eq:divide_o_1}), note that $\E\left[\exp\left(-\frac{3}{4} Z^T Z + \frac{\sqrt{n}}{2} Z^T\left(\bar{Y}_n - \theta_n \right) \right) \mid \D_n \right]$, which equals $\exp\left(\frac{n}{20} \|\bar{Y}_n - \theta_n\|^2 \right) \left(\frac{2}{5}\right)^{d/2}$, is bounded between $(2/5)^{d/2}$ and $\exp(c^2/20) (2/5)^{d/2}$ under the assumption that $\|\bar{Y}_n - \theta_n\| \leq c/\sqrt{n}$. By Lemma~\ref{lemma:expectation}, 
{\footnotesize
\begin{equation*}
\E\left[\exp\left(-\frac{3}{4} X_n^T X_n + \frac{\sqrt{n}}{2} X_n^T\left(\bar{Y}_n - \theta_n \right) \right) \mid \D_n \right] - \E\left[\exp\left(-\frac{3}{4} Z^T Z + \frac{\sqrt{n}}{2} Z^T\left(\bar{Y}_n - \theta_n \right) \right) \mid \D_n \right] = o_P(1).
\end{equation*}
}
Combining these two facts, we conclude that
{\footnotesize
\begin{equation*}
\E\left[\exp\left(-\frac{3}{4} X_n^T X_n + \frac{\sqrt{n}}{2} X_n^T\left(\bar{Y}_n - \theta_n \right) \right) \mid \D_n \right] / \E\left[\exp\left(-\frac{3}{4} Z^T Z + \frac{\sqrt{n}}{2} Z^T\left(\bar{Y}_n - \theta_n \right) \right) \mid \D_n \right] = 1 + o_P(1).
\end{equation*}
}
\end{proof}

\begin{theorem}
Let $Y_1, \ldots, Y_n \sim N(\theta^*, I_d)$. The splitting proportion that minimizes $\E[r^2\{C_n^{\spl}(\alpha)\}]$ is 
\begin{align*}
p_0^* &= 1 - \frac{\sqrt{4d^2 + 8d\log\left(\frac{1}{\alpha}\right)} - 2d}{4\log\left(\frac{1}{\alpha}\right)}~.
\end{align*}
\end{theorem}

\begin{proof}
Recall that $p_0$ represents the proportion of observations that we place in $\D_0$. 

We know that
\begin{align*}
\bar{Y}_0 &\sim N\left(\theta^*, \: \V = \frac{1}{np_0} I_d\right) \\
\bar{Y}_1 &\sim N\left(\theta^*, \: \V = \frac{1}{n(1-p_0)} I_d\right).
\end{align*}
Since all observations in $\D_0$ and $\D_1$ are mutually independent, this implies
\begin{align*}
\bar{Y}_0 - \bar{Y}_1 &\sim N\left(0, \V = \left(\frac{1}{np_0} + \frac{1}{n(1-p_0)}\right)I_d \right) \stepcounter{equation}\tag{\theequation}
\end{align*}
and, hence,
\begin{align*}
\left(\frac{1}{np_0} + \frac{1}{n(1-p_0)}\right)^{-1/2} \left(\bar{Y}_0 - \bar{Y}_1\right) &\sim N\left(0, I_d \right).
\end{align*}
We now see
\begin{align*}
\|\bar{Y}_0 - \bar{Y}_1\|^2 &= \left(\frac{1}{np_0} + \frac{1}{n(1-p_0)}\right) \left\| \left(\frac{1}{np_0} + \frac{1}{n(1-p_0)}\right)^{-1/2} (\bar{Y}_0 - \bar{Y}_1) \right\|^2 \\
&\overset{d}{=} \left(\frac{1}{np_0} + \frac{1}{n(1-p_0)}\right) \chi^2_d~. \stepcounter{equation}\tag{\theequation}\label{eq:p0_chisq}
\end{align*}
When $p_0 = \frac{1}{2}$, this expression is $\frac{4}{n}\chi^2_d$, in agreement with the derivation of equation~\ref{eq:sq_rad_split}.

Setting $\hat{\theta}_1 = \bar{Y}_1$, at $\theta\in \R^d$ we construct the test statistic:
\begin{align*}
T_n(\theta) &= \frac{\prod_{Y_{i}\in \D_0} \exp\left(-\frac{1}{2}(Y_{i} - \hat{\theta}_1)^T (Y_{i} - \hat{\theta}_1) \right)}{\prod_{Y_{i}\in \D_0} \exp\left(-\frac{1}{2}(Y_{i} - \theta)^T (Y_{i} - \theta) \right)} \\
&= \exp\left(\sum_{Y_{i}\in \D_0} \left(-\frac{1}{2}(\bar{Y}_0 - \bar{Y}_1)^T (\bar{Y}_0 - \bar{Y}_1) + \frac{1}{2} (\bar{Y}_0 -\theta)^T (\bar{Y}_0-\theta) \right) \right) \\
&= \exp\left( -\frac{np_0}{2} \| \bar{Y}_0 - \bar{Y}_1 \|^2 + \frac{np_0}{2} \| \bar{Y}_0 -\theta \|^2 \right)~. 
\end{align*}
Using a split proportion of $p_0$, the split LRT confidence set is now
\begin{align*}
C_n^{\spl} &= \left\{\theta\in\Theta: \exp\left( -\frac{np_0}{2} \| \bar{Y}_0 - \bar{Y}_1 \|^2 + \frac{np_0}{2} \| \bar{Y}_0 -\theta \|^2 \right)  \leq \frac{1}{\alpha}  \right\} \\
&= \left\{\theta\in\Theta: -\frac{np_0}{2} \| \bar{Y}_0 - \bar{Y}_1 \|^2 + \frac{np_0}{2} \| \bar{Y}_0 -\theta \|^2  \leq \log \left(\frac{1}{\alpha}\right)  \right\} \\
&= \left\{\theta\in\Theta: \frac{np_0}{2} \| \bar{Y}_0 -\theta \|^2  \leq \log \left(\frac{1}{\alpha}\right)  + \frac{np_0}{2} \| \bar{Y}_0 - \bar{Y}_1 \|^2 \right\} \\
&= \left\{\theta\in\Theta: \| \bar{Y}_0 -\theta \|^2  \leq \frac{2}{np_0} \log \left(\frac{1}{\alpha}\right)  + \| \bar{Y}_0 - \bar{Y}_1 \|^2 \right\}~. \\
\end{align*}

The squared radius is thus $r^2(C_n^{\spl}(\alpha)) = \frac{2}{np_0} \log \left(\frac{1}{\alpha}\right)  + \| \bar{Y}_0 - \bar{Y}_1 \|^2$. By (\ref{eq:p0_chisq}), the expected squared radius at a given value of $p_0$ is
\begin{align*}
\footnotesize s(p_0) = \frac{2}{np_0}\log\left(\frac{1}{\alpha}\right) + \left(\frac{1}{np_0} + \frac{1}{n(1-p_0)}\right)d~.
\end{align*}

We can now minimize this function:
\begin{align*}
0 &\overset{\text{set}}{=} \frac{\partial}{\partial p_0} s(p_0) = \frac{-2}{np_0^2}\log\left(\frac{1}{\alpha}\right) - \frac{d}{np_0^2} + \frac{d}{n(1-p_0)^2} \\
&\Updownarrow \\
0 &= -2(1-p_0)^2 \log\left(\frac{1}{\alpha}\right) - d(1-p_0)^2 + dp_0^2 \\
&= -2(1-2p_0+p_0^2) \log\left(\frac{1}{\alpha}\right) - d(1-2p_0 + p_0^2) + dp_0^2 \\
&= -2\log\left(\frac{1}{\alpha}\right) + 4p_0\log\left(\frac{1}{\alpha}\right) - 2p_0^2\log\left(\frac{1}{\alpha}\right) - d + 2dp_0 - dp_0^2 + dp_0^2 \\
&= p_0^2 \left(-2\log\left(\frac{1}{\alpha}\right)\right) + p_0\left(4\log\left(\frac{1}{\alpha}\right) + 2d \right) + \left(-2\log\left(\frac{1}{\alpha}\right) - d\right).
\end{align*}
This is now a quadratic expression in $p_0$. Thus, this formula is solved by
\begin{align*}
p_0 &= \frac{-4\log\left(\frac{1}{\alpha}\right) - 2d \pm \sqrt{\left(4\log\left(\frac{1}{\alpha}\right) + 2d\right)^2 - 4\left(-2\log\left(\frac{1}{\alpha}\right)\right)\left(-2\log\left(\frac{1}{\alpha}\right) - d \right)}}{2\left(-2\log\left(\frac{1}{\alpha}\right)\right)} \\
&= \frac{4\log\left(\frac{1}{\alpha}\right) + 2d \pm \sqrt{4d^2 + 8d\log\left(\frac{1}{\alpha}\right)}}{4\log\left(\frac{1}{\alpha}\right)}~.
\end{align*}
We now consider the $\pm$ choice. In the $+$ direction, we have
$$p_0 = \frac{4\log\left(\frac{1}{\alpha}\right) + 2d + \sqrt{4d^2 + 8d\log\left(\frac{1}{\alpha}\right)}}{4\log\left(\frac{1}{\alpha}\right)} = 1 + \frac{2d + \sqrt{4d^2 + 8d\log\left(\frac{1}{\alpha}\right)}}{4\log\left(\frac{1}{\alpha}\right)} > 1~.$$
However, in the $-$ direction, we can show that $p_0 \in \left(\frac{1}{2}, 1\right)$. We note that 
\begin{align*}
2d &< \sqrt{4d^2 + 8d\log\left(\frac{1}{\alpha}\right)} < \sqrt{4d^2 + 8d\log\left(\frac{1}{\alpha}\right) + 4\left(\log\left(\frac{1}{\alpha}\right)\right)^2} \\
&= \sqrt{\left(2d+2\log\left(\frac{1}{\alpha}\right)\right)^2} = 2d+2\log\left(\frac{1}{\alpha}\right).
\end{align*}
So
\begin{align*}
p_0 &= 1 + \frac{2d - \sqrt{4d^2 + 8d\log\left(\frac{1}{\alpha}\right)}}{4\log\left(\frac{1}{\alpha}\right)} < 1 + \frac{2d - 2d}{4\log\left(\frac{1}{\alpha}\right)} = 1 \\
&\text{and} \\
p_0 &= 1 + \frac{2d - \sqrt{4d^2 + 8d\log\left(\frac{1}{\alpha}\right)}}{4\log\left(\frac{1}{\alpha}\right)} > 1 + \frac{2d - 2d - 2\log\left(\frac{1}{\alpha}\right)}{4\log\left(\frac{1}{\alpha}\right)} = 1 - \frac{1}{2} = \frac{1}{2}~.
\end{align*}
This means that $$ p_0^* = 1 - \frac{\sqrt{4d^2 + 8d\log\left(\frac{1}{\alpha}\right)} - 2d}{4\log\left(\frac{1}{\alpha}\right)}$$ optimizes $s(p_0)$, and $p_0^* \in \left(\frac{1}{2}, 1\right)$. Furthermore, this optimum must be a minimum, since for any $p_0 \in (0,1)$, 
$$\frac{\partial^2}{\partial p_0^2}\: s(p_0) = \frac{4}{np_0^3}\log\left(\frac{1}{\alpha}\right) + \frac{2d}{np_0^3} + \frac{2d}{n(1-p_0)^3} > 0~.$$
We can use L'H\^{o}pital's Rule to show that $p_0^* \to \frac{1}{2}$ as $d\to\infty$:
\begin{align*}
\lim_{d\to\infty} p_0^* &= 1 -  \lim_{d\to\infty}  \frac{\sqrt{4d^2 + 8d\log\left(\frac{1}{\alpha}\right)} - 2d}{4\log\left(\frac{1}{\alpha}\right)} \\
&= 1 - \lim_{d\to\infty}  \frac{\sqrt{4 + (8/d)\log\left(1/\alpha\right)} - 2}{(4/d)\log(1/\alpha)} \\
&= 1 - \lim_{d\to\infty} \frac{\frac{1}{2} \left(4 + (8/d) \log(1/\alpha)\right)^{-1/2} (-8/d^2) \log(1/\alpha)}{(-4/d^2) \log(1/\alpha)} \\
&= 1 - \lim_{d\to\infty}  \left(4 + (8/d) \log(1/\alpha)\right)^{-1/2} \\
&= \frac{1}{2}~.
\end{align*}
We conclude that as $d\to\infty$ for fixed $\alpha$, the optimal choice of $p_0^* \to 0.5$.
\end{proof}

\begin{theorem}
Suppose $Y_1, \ldots, Y_n$ are iid observations from  $N(\theta^*, I_d)$. Split the sample such that $\D_0$ and $\D_1$ each contain $n/2$ observations. Use $\D_0$ and $\D_1$ to define the split and cross-fit sets. Then $C_n^\CF(\alpha)$ is a subset of a translation of the split LRT set, recentered at $\bar{Y}$. That is, $C_n^\CF(\alpha) \subseteq \left\{\theta \in \Theta: \|\theta - \bar{Y}\|^2 < (4/n)\log(1/\alpha) + \|\bar{Y}_0 - \bar{Y}_1\|^2 \right\}$, and hence $\vol\{C_n^\CF(\alpha)\} \leq \vol\{C_n^{\spl}(\alpha)\}$. Furthermore, if and only if $\bar{Y}_0 = \bar{Y}_1$, $C_n^\CF(\alpha)$ and $C_n^\spl(\alpha)$ have equal volume and are in fact the same set.
\end{theorem}

\begin{proof}
Let $\theta\in C_n^{\CF}(\alpha)$. Then
\begin{align*}
\exp&\left(-\frac{n}{4}\|\bar{Y}_0 - \bar{Y}_1\|^2 + \frac{n}{4} \|\bar{Y}-\theta\|^2 \right) \\
&= \exp\left(-\frac{n}{4}\|\bar{Y}_0 - \bar{Y}_1\|^2 + \frac{n}{4} \left\|\frac{1}{2}(\bar{Y}_0-\theta) + \frac{1}{2}(\bar{Y}_1-\theta)\right\|^2 \right) \\
&\leq \exp\left(-\frac{n}{4}\|\bar{Y}_0 - \bar{Y}_1\|^2 + \frac{n}{8}\|\bar{Y}_0 - \theta\|^2 + \frac{n}{8} \|\bar{Y}_1 - \theta\|^2 \right) \stepcounter{equation}\tag{\theequation}\label{eq:cfp0_ineq1} \\ 
&= \exp\left(-\frac{n}{8}\|\bar{Y}_0 - \bar{Y}_1\|^2 + \frac{n}{8}\|\bar{Y}_0 - \theta\|^2 -\frac{n}{8}\|\bar{Y}_0 - \bar{Y}_1\|^2 + \frac{n}{8} \|\bar{Y}_1 - \theta\|^2 \right) \\
&\leq \frac{1}{2}\Bigg[ \exp\left(-\frac{n}{4}\|\bar{Y}_0 - \bar{Y}_1\|^2 + \frac{n}{4}\|\bar{Y}_0 - \theta\|^2 \right) + \\
&\hspace*{3em} \exp\left(-\frac{n}{4} \|\bar{Y}_0 - \bar{Y}_1\|^2 + \frac{n}{4} \|\bar{Y}_1 - \theta\|^2 \right) \Bigg] \stepcounter{equation}\tag{\theequation}\label{eq:cfp0_ineq2} \\
&< \frac{1}{\alpha}~.
\end{align*}
Line (\ref{eq:cfp0_ineq1}) holds because $\|\cdot\|^2$ is convex. Line (\ref{eq:cfp0_ineq2}) holds because $\exp(\cdot)$ is convex. Thus, \\ $C_n^{\CF}(\alpha) \subseteq \left\{\theta\in\Theta : \|\bar{Y} - \theta\|^2 < (4/n) \log(1/\alpha) + \|\bar{Y}_0 - \bar{Y}_1\|^2 \right\}$, which has the same volume as $C_n^{\spl}(\alpha) = \left\{\theta\in\Theta: \|\bar{Y}_0 - \theta\|^2 < (4/n) \log\left(1/\alpha\right) + \|\bar{Y}_0 - \bar{Y}_1\|^2 \right\}$. Hence, it also holds that $\vol\left(C_n^{\CF}(\alpha)\right) \leq \vol\left(C_n^{\spl}(\alpha)\right)$.

Now suppose $\bar{Y}_0 = \bar{Y}_1$. Since $\|\cdot\|^2$ and $\exp(\cdot)$ are strictly convex, equality holds in (\ref{eq:cfp0_ineq1}) and (\ref{eq:cfp0_ineq2}) only in this case. This means that if and only if $\bar{Y}_0 = \bar{Y}_1$,
\begin{align*}
C_n^{\spl}(\alpha) &= \left\{ \theta \in \Theta: \exp\left(-\frac{n}{4}\|\bar{Y}_0 - \bar{Y}_1\|^2 + \frac{n}{4} \|\bar{Y}_0 - \theta\|^2 \right) < \frac{1}{\alpha} \right\} \\
&= \left\{ \theta \in \Theta: \exp\left(-\frac{n}{4}\|\bar{Y}_0 - \bar{Y}_1\|^2 + \frac{n}{4} \|\bar{Y}-\theta\|^2 \right) < \frac{1}{\alpha} \right\} \\
&= \left\{ \theta \in \Theta: \frac{1}{2} \exp\left(-\frac{n}{4}\|\bar{Y}_0 - \bar{Y}_1\|^2\right) \left\{ \exp\left(\frac{n}{4}\|\bar{Y}_0 - \theta\|^2\right) + \exp\left(\frac{n}{4}\|\bar{Y}_1 - \theta\|^2\right) \right\} < \frac{1}{\alpha} \right\} \\
&= C_n^{\CF}(\alpha)~.
\end{align*}
Thus,  $\vol\left(C_n^{\CF}(\alpha)\right) = \vol\left(C_n^{\spl}(\alpha)\right)$.
\end{proof}

\begin{theorem}
Let $f_d(x)$ be the probability density function of the $\chi^2_d$ distribution, and let $c_{\alpha, d}$ be the upper $\alpha$ quantile of the $\chi^2_d$ distribution. Assume $c_{\alpha, d} + \log(\alpha) > d - 2$. Then
\begin{align*}
\P\left[r\{C_n^{\spl}(\alpha)\} / r\{C_n^\LRT(\alpha)\} \leq 2\right] &\geq 1 - \alpha - \log(1/\alpha) f_d\{c_{\alpha, d} + \log(\alpha)\} \\ 
\text{and} \quad \P\left[r\{C_n^{\spl}(\alpha)\} / r\{C_n^\LRT(\alpha)\} \leq 2\right] &\leq 1 - \alpha - \log(1/\alpha) f_d(c_{\alpha, d}).
\end{align*}
As $d\to\infty$ for fixed $\alpha \leq 0.17$, $\log(1/\alpha) f_d\{c_{\alpha, d} + \log(\alpha)\}$ and $\log(1/\alpha) f_d(c_{\alpha, d})$ both converge to 0.
\end{theorem}

\begin{proof}
We divide this proof into a proof of the bounds on $\P\left[r\{C_n^{\spl}(\alpha)\} / r\{C_n^\LRT(\alpha)\} \leq 2\right]$, a proof of the fact that $d\geq 2$ and $\alpha \leq 0.17$ implies $c_{\alpha, d} + \log(\alpha) > d - 2$, and a proof of the behavior as $d\to\infty$ for $\alpha \leq 0.17$.
\newline \newline
\textit{Proof of bounds in Theorem~\ref{thm:bounds}.}
\newline \newline
We use the fact that $r^2(C_n^{\spl}(\alpha)) = \frac{4}{n}\log(1/\alpha) + \|\bar{Y}_0 - \bar{Y}_1\|^2$. As established in the proof of Theorem~\ref{thm:split_p0} and the derivation of equation~\ref{eq:sq_rad_split}, we know that $\|\bar{Y}_0 - \bar{Y}_1\|^2 \overset{d}{=} (4/n)\chi^2_d$. Let $X\sim \chi^2_d$. Note that $\log(\alpha) < 0$. Then
\begin{align*}
\P\left(r(C_n^{\spl}(\alpha)) \: / \: r(C_n^\LRT(\alpha)) \leq 2\right) &= \P\left(r^2(C_n^{\spl}(\alpha)) \: / \: r^2(C_n^\LRT(\alpha)) \leq 4\right) \\
&= \P\left(r^2(C_n^{\spl}(\alpha)) \leq \frac{4}{n} c_{\alpha, d} \right) \\
&= \P\left(\frac{4}{n}\log(1/\alpha) + \frac{4}{n}X \leq  \frac{4}{n} c_{\alpha, d} \right) \\
&= \P\left(\log(1/\alpha) + X \leq c_{\alpha, d} \right) \\
&= \P(X \leq c_{\alpha, d} + \log(\alpha)) \\
&= \P(X \leq c_{\alpha, d}) - \P(c_{\alpha, d} + \log(\alpha) \leq X \leq c_{\alpha, d}) \\
&= 1 - \alpha - \P(c_{\alpha, d} + \log(\alpha) \leq X \leq c_{\alpha, d})~.
\end{align*}
Now we need to bound $\P(c_{\alpha, d} + \log(\alpha) \leq X \leq c_{\alpha, d})$. Under the assumed conditions, we show that the $\chi^2_d$ pdf is decreasing on $[c_{\alpha, d} + \log(\alpha), c_{\alpha, d}]$. Let $f_d(x)$ be the $\chi^2_d$ pdf. Since $$\frac{\partial}{\partial x} f_d(x) = \frac{1}{2^{d/2} \: \Gamma(d/2)} \left[\left(\frac{d}{2}-1\right) x^{d/2 - 2} e^{-x/2} + x^{d/2 - 1} \left(-\frac{1}{2}e^{-x/2}\right) \right],$$ $f_d(\cdot)$ is decreasing at $x$ if and only if $$\frac{1}{2^{d/2} \: \Gamma(d/2)} \left[\left(\frac{d}{2}-1\right) x^{d/2 - 2} e^{-x/2} + x^{d/2 - 1} \left(-\frac{1}{2}e^{-x/2}\right) \right] < 0~.$$ Re-writing, this implies $$x^{d/2 - 1} \left(\frac{1}{2}e^{-x/2}\right) > \left(\frac{d}{2}-1\right) x^{d/2 - 2} e^{-x/2}~,$$ which holds if and only if $x > d-2$.

By our initial assumption, $c_{\alpha, d} + \log(\alpha) > d - 2.$ Thus, $f_d(x)$ is decreasing on \mbox{$[c_{\alpha, d} + \log(\alpha), c_{\alpha, d}]$.} Since the interval has length $\log(1/\alpha)$, $$\log(1/\alpha) f_d(c_{\alpha, d})  \:\: \leq  \:\: \P(c_{\alpha, d} + \log(\alpha) \leq X \leq c_{\alpha, d})  \:\: \leq  \:\: \log(1/\alpha) f_d(c_{\alpha,d} + \log(\alpha))~.$$ The bounds on $\P\left(r(C_n^{\spl}(\alpha)) \: / \: r(C_n^\LRT(\alpha)) \leq 2\right)$ follow immediately. 
\newline \newline
\textit{Proof that Theorem~\ref{thm:bounds} condition is satisfied for $d \geq 2$ and $\alpha \leq 0.17$.}
\newline \newline
In the text, we note that if $d \geq 2$ and $\alpha \leq 0.17$, then $c_{\alpha, d} + \log(\alpha) > d - 2$. To see this, we use a fact from \cite{inglot2010inequalities}: For $d\geq 2$ and $\alpha \leq 0.17$,  it holds that $c_{\alpha, d} \geq d + 2\log(1/\alpha) - 5/2$. This implies 
$$c_{\alpha, d} + \log(\alpha) \geq d + \log(1/\alpha) - 5/2 \geq d + \log(1/0.17) - 5/2 > d - 2~,$$
which concludes the argument. 
\newline \newline
\textit{Proof of behavior as $d\to\infty$ for $\alpha \leq 0.17$ in Theorem~\ref{thm:bounds}.}
\newline \newline
Assume $d\geq 3$ and $\alpha \leq 0.17$. Above we showed that if $f_d(x)$ is the $\chi^2_d$ pdf, then $f_d(x)$ is decreasing in $x$ for $x > d - 2$. We know that 
$$d - 2 < d + \log(1/.17) - 5/2 \leq d + \log(1/\alpha) - 5/2~.$$
Also, because $c_{\alpha, d} \geq d + 2\log(1/\alpha) - 5/2$ for $d\geq 2$ and $\alpha \leq 0.17$, 
$$c_{\alpha, d} \geq c_{\alpha, d} + \log(\alpha) \geq d + 2\log(1/\alpha) - \log(1/\alpha) - 5/2 = d + \log(1/\alpha) - 5/2~.$$
Hence $f_d(x)$ is decreasing for $x \geq d + \log(1/\alpha) - 5/2$. If we prove that $$\lim_{d\to\infty} f_d(d + \log(1/\alpha) - 5/2) = 0~,$$ we can conclude that
\begin{align*}
\lim_{d\to\infty} \log(1/\alpha) f_d\{c_{\alpha, d} + \log(\alpha)\} &= 0 \quad\text{and} \\ 
\lim_{d\to\infty} \log(1/\alpha) f_d(c_{\alpha, d}) &= 0~.
\end{align*}
Since the $\chi^2_d$ pdf uses the Gamma function, we will make use of Stirling's formula. This formula states that $\Gamma(m) \sim \sqrt{2\pi (m-1)} (m-1)^{m-1} / \exp(m-1)$, where $\sim$  means that as $m\to\infty$, the ratio of the two sides converges to 1. We see that
\begin{align*}
f_d&(d + \log(1/\alpha) - 5/2) \\
&= \frac{1}{2^{d/2} \Gamma(d/2)} \left\{d + \log\left(\frac{1}{\alpha}\right) - \frac{5}{2}\right\}^{d/2-1}\exp\left[-\frac{1}{2} \left\{d + \log\left(\frac{1}{\alpha}\right) - \frac{5}{2} \right\} \right] \\
&\sim 2^{-d/2} \left\{2\pi \left(\frac{d}{2} - 1\right)\right\}^{-1/2} \left(\frac{d}{2} - 1\right)^{-(d/2 - 1)} \exp\left(\frac{d}{2} - 1\right) \times \\ 
&\qquad \left\{d + \log\left(\frac{1}{\alpha}\right) - \frac{5}{2}\right\}^{d/2-1}\exp\left[-\frac{1}{2} \left\{d + \log\left(\frac{1}{\alpha}\right) - \frac{5}{2} \right\} \right] \\
&= 2^{-d/2} \left\{2\pi \left(\frac{d}{2} - 1\right)\right\}^{-1/2} \left(\frac{d}{2} - 1\right)^{-(d/2 - 1)} \times \\ 
&\qquad \left\{d + \log\left(\frac{1}{\alpha}\right) - \frac{5}{2}\right\}^{d/2-1}\exp\left\{-1 - \frac{1}{2}\log\left(\frac{1}{\alpha}\right) + \frac{5}{4} \right\} \\
&= 2^{-(d/2 - 1)} \left(\frac{1}{2} \right) \left\{2\pi \left(\frac{d}{2} - 1\right)\right\}^{-1/2} \left(\frac{d}{2} - 1\right)^{-(d/2 - 1)} \times \\ 
&\qquad \left\{d + \log\left(\frac{1}{\alpha}\right) - \frac{5}{2}\right\}^{d/2-1}\exp\left\{\frac{1}{4} - \frac{1}{2}\log\left(\frac{1}{\alpha}\right) \right\} \\
&= \frac{1}{2} \left(\pi (d - 2)\right)^{-1/2} \left(d - 2\right)^{-(d/2 - 1)} \left\{d + \log\left(\frac{1}{\alpha}\right) - \frac{5}{2}\right\}^{d/2-1}\exp\left\{\frac{1}{4} - \frac{1}{2}\log\left(\frac{1}{\alpha}\right) \right\} \\
&= \frac{1}{2} \left(\pi (d - 2)\right)^{-1/2} \left(\frac{d - 2 + \log(1/\alpha) - 1/2}{d-2} \right)^{d/2 - 1} \exp\left\{1/4 - (1/2)\log(1/\alpha) \right\} \\
&= \frac{1}{2} \left(\pi (d - 2)\right)^{-1/2} \left(1 + \frac{\log(1/\alpha) - 1/2}{d-2} \right)^{d/2 - 1} \exp\left\{1/4 - (1/2)\log(1/\alpha) \right\}.
\end{align*}
Since $\lim_{m\to\infty} (1 + x/m)^m = \exp(x)$, we see that
\begin{align*}
\lim_{d\to\infty} \left(1 + \frac{\log(1/\alpha) - 1/2}{d-2} \right)^{d/2 - 1} &= \left( \lim_{d\to\infty} \left\{1 + \frac{\log(1/\alpha) - 1/2}{d - 2} \right\}^{d-2} \right)^{1/2} \\
&= \left( \exp\left\{\log(1/\alpha) - 1/2\right\} \right)^{1/2} \\
&= \left\{(1/\alpha) \exp(-1/2) \right\}^{1/2} \\
&= (1/\alpha)^{1/2} \exp(-1/4)~.
\end{align*}
This implies that 
\begin{align*}
f_d(d + \log(1/\alpha) - 5/2) &\sim  \frac{1}{2} \left(\pi (d - 2)\right)^{-1/2} (1/\alpha)^{1/2}  \exp(-1/4) \exp\left\{1/4 - (1/2)\log(1/\alpha) \right\} \\
&= \frac{1}{2} \left(\pi (d - 2)\right)^{-1/2} (1/\alpha)^{1/2}  \exp\left\{- (1/2)\log(1/\alpha) \right\}  \\
&= \frac{1}{2\sqrt{\pi(d-2)}}~.
\end{align*}
We conclude that $\lim_{d\to\infty} f_d(d + \log(1/\alpha) - 5/2) = 0$.
\end{proof}

Before proving Theorem~\ref{thm:hybrid_valid}, we establish Lemma~\ref{lemma:ripr} and Lemma~\ref{lemma:Rn_leq_1}.

\begin{lemma} \label{lemma:ripr}
Assume the doughnut null test setting. Let $\mathcal{P}_{\Theta_0}$ be the set of all convex combinations of $N(\theta, I_d)$ densities such that $\|\theta\| \in [0.5, 1]$. When $\|\bar{Y}_1\| > 1$ and $\hat{\theta}_1 = \bar{Y}_1$, the RIPR of $p_{\hat{\theta}_1}$ onto $\mathcal{P}_{\Theta_0}$ is $p_{\hat{\theta}_1 / \|\hat{\theta}_1\|}$.
\end{lemma}

\begin{proof}
Suppose $\|\bar{Y}_1\| > 1$. Defining $\hat{\theta}_1 = \bar{Y}_1$ as in Table \ref{tab:subsample_split_ripr}, $\|\hat{\theta}_1\| > 1$. The RIPR of $\hat{\theta}_1$ onto the convex set $\mathcal{P}_{\Theta_0}$ minimizes $D_\KL(p_{\hat{\theta}_1} \| p_0)$ out of all densities $p_0 \in \mathcal{P}_{\Theta_0}$. Suppose $p_0 \in \mathcal{P}_{\Theta_0}$. Then we can write $p_0$ as a mixture of $N(\theta_k, I_d)$ densities. We write $p_0 = \sum_{k=1}^K w_k p_{\theta_k}$, where $K\in\mathbb{N}$, $\sum_{k=1}^K w_k = 1$, and for each $k=1,\ldots, K$, $0 \leq w_k \leq 1$ and $\|\theta_k\| \in [0.5, 1]$. Note that $p_{\hat{\theta}_1 / \|\hat{\theta}_1\|} \in \mathcal{P}_{\Theta_0}$. To prove that $D_\KL(p_{\hat{\theta}_1} \: \| \: p_{\hat{\theta}_1 / \|\hat{\theta}_1\|}) = \inf_{p_0 \in \mathcal{P}_{\Theta_0}} D_\KL(p_{\hat{\theta}_1} \: \| \: p_0)$, we show that $D_\KL(p_{\hat{\theta}_1} \: \| \: p_{\hat{\theta}_1 / \|\hat{\theta}_1\|}) \leq D_\KL(p_{\hat{\theta}_1} \: \| \: \sum_{k=1}^K w_k p_{\theta_k})$:
\begin{align*}
D_\KL&\left(p_{\hat{\theta}_1} \: \Big| \Big| \: \sum_{k=1}^K w_k p_{\theta_k}\right) - D_\KL\left(p_{\hat{\theta}_1} \: \| \: p_{\hat{\theta}_1 / \|\hat{\theta}_1\|} \right) \\
&= \int_{\R^d} p_{\hat{\theta}_1}(y) \log\left(\frac{p_{\hat{\theta}_1}(y)}{\sum_{k=1}^K w_k p_{\theta_k}(y)}\right) dy - \int_{\R^d} p_{\hat{\theta}_1}(y) \log\left(\frac{p_{\hat{\theta}_1}(y)}{p_{\hat{\theta}_1 / \|\hat{\theta}_1\|}(y)}\right) dy \\
&= \int_{\R^d} p_{\hat{\theta}_1}(y) \log\left(\frac{p_{\hat{\theta}_1 / \|\hat{\theta}_1\|}(y)}{\sum_{k=1}^K w_k p_{\theta_k}(y)} \right) dy = - \int_{\R^d} p_{\hat{\theta}_1}(y) \log\left(\frac{\sum_{k=1}^K w_k p_{\theta_k}(y)}{p_{\hat{\theta}_1 / \|\hat{\theta}_1\|}(y)} \right) dy \\
&= -\E_{\hat{\theta}_1} \left[\log\left\{\frac{\sum_{k=1}^K w_k p_{\theta_k}(Y)}{p_{\hat{\theta}_1 / \|\hat{\theta}_1\|}(Y)} \right\} \right] \geq -\log \E_{\hat{\theta}_1}\left\{\frac{\sum_{k=1}^K w_k p_{\theta_k}(Y)}{p_{\hat{\theta}_1 / \|\hat{\theta}_1\|}(Y)} \right\} \stepcounter{equation}\tag{\theequation}\label{eq:ripr_proof_1} \\
&= -\log \left[\sum_{k=1}^K w_k \E_{\hat{\theta}_1}\left\{\frac{p_{\theta_k}(Y)}{p_{\hat{\theta}_1 / \|\hat{\theta}_1\|}(Y)} \right\} \right] \geq -\log \left\{\sum_{k=1}^K w_k (1) \right\} = 0~. \stepcounter{equation}\tag{\theequation}\label{eq:ripr_proof_2} 
\end{align*}
The inequality in (\ref{eq:ripr_proof_1}) holds by Jensen's inequality. The inequality in (\ref{eq:ripr_proof_2}) holds by the following derivation:
{\small
\begin{align*}
\E_{\hat{\theta}_1}&\left\{\frac{p_{\theta_k}(Y)}{p_{\hat{\theta}_1 / \|\hat{\theta}_1\|}(Y)} \right\} \\
&= \int_{\R^d} \frac{1}{(2\pi)^{d/2}} \exp\left(-\frac{1}{2} \|y - \hat{\theta}_1\|^2\right) \frac{\exp\left(-\frac{1}{2} \|y - \theta_k\|^2 \right)}{\exp\left(-\frac{1}{2} \|y - \hat{\theta}_1 / \|\hat{\theta}_1\| \|^2 \right)} dy \\
&= \int_{\R^d} \frac{1}{(2\pi)^{d/2}} \exp\left(-\frac{1}{2} \|y - \hat{\theta}_1\|^2 - \frac{1}{2}\|y - \hat{\theta}_1 + \hat{\theta}_1 - \theta_k\|^2 + \frac{1}{2}\|y - \hat{\theta}_1 + \hat{\theta}_1 - \hat{\theta}_1 / \|\hat{\theta}_1\| \|^2 \right) dy \\
&\\ 
&= \int_{\R^d} \frac{1}{(2\pi)^{d/2}} \exp\bigg(-\frac{1}{2} \|y - \hat{\theta}_1\|^2 - (y - \hat{\theta}_1)^T (\hat{\theta}_1 - \theta_k) - \frac{1}{2} \|\hat{\theta}_1 - \theta_k\|^2 + (y - \hat{\theta}_1)^T (\hat{\theta}_1 - \hat{\theta}_1/\|\hat{\theta}_1\|) + \\
&\hspace*{10em} \frac{1}{2}\|\hat{\theta}_1 - \hat{\theta}_1/\|\hat{\theta}_1\|\|^2 \bigg) dy \\
&= \exp\left( \frac{1}{2}\|\hat{\theta}_1 - \hat{\theta}_1/\|\hat{\theta}_1\|\|^2 - \frac{1}{2} \|\hat{\theta}_1 - \theta_k\|^2 \right) \int_{\R^d} \frac{1}{(2\pi)^{d/2}} \exp\left(-\frac{1}{2} \|y - \hat{\theta}_1\|^2 + (y - \hat{\theta}_1)^T (\theta_k - \hat{\theta}_1 / \|\hat{\theta}_1\|) \right) dy~.
\end{align*}
}
Reinterpreting the integral as an expectation, we see
{\small
\begin{align*}
\E_{\hat{\theta}_1}&\left\{\frac{p_{\theta_k}(Y)}{p_{\hat{\theta}_1 / \|\hat{\theta}_1\|}(Y)} \right\} \\
&= \exp\left( \frac{1}{2}\|\hat{\theta}_1 - \hat{\theta}_1/\|\hat{\theta}_1\|\|^2 - \frac{1}{2} \|\hat{\theta}_1 - \theta_k\|^2 \right) \E_{\hat{\theta}_1}\left[\exp\left\{(Y - \hat{\theta}_1)^T (\theta_k - \hat{\theta}_1 / \|\hat{\theta}_1\|) \right\} \right] \\
&= \exp\left( \frac{1}{2}\|\hat{\theta}_1 - \hat{\theta}_1/\|\hat{\theta}_1\|\|^2 - \frac{1}{2} \|\hat{\theta}_1 - \theta_k\|^2 - \hat{\theta}_1^T (\theta_k - \hat{\theta}_1/\|\hat{\theta}_1\|) \right) \E_{\hat{\theta}_1}\left[\exp\left\{(\theta_k - \hat{\theta}_1 / \|\hat{\theta}_1\|)^T Y \right\} \right] \\
&= \exp\left( \frac{1}{2}\|\hat{\theta}_1 - \hat{\theta}_1/\|\hat{\theta}_1\|\|^2 - \frac{1}{2} \|\hat{\theta}_1 - \theta_k\|^2 - \hat{\theta}_1^T (\theta_k - \hat{\theta}_1/\|\hat{\theta}_1\|) \right) \exp\left\{\hat{\theta}_1^T (\theta_k - \hat{\theta}_1/\|\hat{\theta}_1\|) + \frac{1}{2} \|\theta_k - \hat{\theta}_1/\|\hat{\theta}_1\| \|^2 \right\} \\
&= \exp\left( \frac{1}{2}\|\hat{\theta}_1 - \hat{\theta}_1/\|\hat{\theta}_1\|\|^2 - \frac{1}{2} \|\hat{\theta}_1 - \theta_k\|^2 + \frac{1}{2} \|\theta_k - \hat{\theta}_1/\|\hat{\theta}_1\| \|^2 \right) \\
&= \exp\bigg(\frac{1}{2} \|\hat{\theta}_1\|^2  - \hat{\theta}_1^T \hat{\theta}_1 / \|\hat{\theta}_1\| + \frac{1}{2}\hat{\theta}_1^T \hat{\theta}_1 / \|\hat{\theta}_1\|^2 - \frac{1}{2} \|\hat{\theta}_1\|^2 + \hat{\theta}_1^T \theta_k - \frac{1}{2}\|\theta_k\|^2 + \\ 
&\hspace*{4em} \frac{1}{2}\|\theta_k\|^2 - \theta_k^T \hat{\theta}_1 / \|\hat{\theta}_1\| + \frac{1}{2} \hat{\theta}_1^T \hat{\theta}_1 / \|\hat{\theta}_1\|^2 \bigg) \\
&= \exp\left(\hat{\theta}_1^T \hat{\theta}_1 / \|\hat{\theta}_1\|^2 - \hat{\theta}_1^T \hat{\theta}_1 / \|\hat{\theta}_1\| - \theta_k^T \hat{\theta}_1 / \|\hat{\theta}_1\| + \hat{\theta}_1^T \theta_k  \right) \\
&= \exp\left\{(\hat{\theta}_1 / \|\hat{\theta}_1\| - \hat{\theta}_1)^T (\hat{\theta}_1 / \|\hat{\theta}_1\| - \theta_k) \right\} \\
&\leq \exp(0) \stepcounter{equation}\tag{\theequation}\label{eq:ripr_proof_3} \\
&= 1~.
\end{align*}
}
To justify (\ref{eq:ripr_proof_3}), note that 
\begin{equation*}
(\hat{\theta}_1 / \|\hat{\theta}_1\| - \hat{\theta}_1)^T (\hat{\theta}_1 / \|\hat{\theta}_1\| - \theta_k) = \Big\|\hat{\theta}_1 / \|\hat{\theta}_1\| - \hat{\theta}_1\Big\| \Big\|\hat{\theta}_1 / \|\hat{\theta}_1\| - \theta_k\Big\| \text{cos}(\gamma)~, 
\end{equation*}
where $\gamma$ is the angle between $\hat{\theta}_1 / \|\hat{\theta}_1\| - \hat{\theta}_1$ and $\hat{\theta}_1 / \|\hat{\theta}_1\| - \theta_k$. Recall that the outer border of $\Theta_0$ is a sphere, $\|\theta_k\| \in [0.5, 1]$, $\|\hat{\theta}_1\| > 1$, and $\hat{\theta}_1 / \|\hat{\theta}_1\|$ is on the outer border of $\Theta_0$. Thus, $\gamma$ will always be between $90^\circ$ and $270^\circ$. (See Figure~\ref{fig:circles_annotated}.) This implies that $(\hat{\theta}_1 / \|\hat{\theta}_1\| - \hat{\theta}_1)^T (\hat{\theta}_1 / \|\hat{\theta}_1\| - \theta_k) \leq 0$.
\end{proof}

\begin{figure}
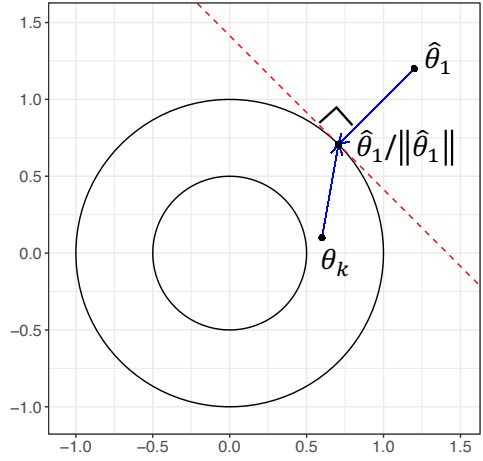

\figuresize{.35}
\figurebox{20pc}{25pc}{}[SuppFigure1.pdf]
\caption{Lemma~\ref{lemma:ripr} companion diagram. The angle between $\hat{\theta}_1 / \|\hat{\theta}_1\| - \hat{\theta}_1$ and $\hat{\theta}_1 / \|\hat{\theta}_1\| - \theta_k$ must be between $90^\circ$ and $270^\circ$.}
\label{fig:circles_annotated}
\end{figure}

\begin{lemma} \label{lemma:Rn_leq_1}
Assume the doughnut null test setting. Let $R_n = \prod_{Y_i \in \D_0} \{p_{\hat{\theta}_1}(Y_i) / p_{\hat{\theta}_1 / \|\hat{\theta}_1\|}(Y_i) \}$. If $\theta^*\in\Theta_0$, then $\E_{\theta^*}\{R_n \one(\|\bar{Y}_1\| > 1) \mid \D_1 \} \leq \one(\|\bar{Y}_1\| > 1).$
\end{lemma}

\begin{proof}
If $\D_1$ satisfies $\|\bar{Y}_1\| \leq 1$, then
\begin{equation*}
\E_{\theta^*}\{R_n \one(\|\bar{Y}_1\| > 1) \mid \D_1 \} = 0 = \one(\|\bar{Y}_1\| > 1)~.
\end{equation*}

Now suppose $\D_1$ satisfies $\|\bar{Y}_1\| > 1$. Then $\|\hat{\theta}_1\| > 1$, and $p_{\hat{\theta}_1 / \|\hat{\theta}_1\|}$ is the RIPR of $p_{\hat{\theta}_1}$ onto the convex set of densities $\mathcal{P}_{\Theta_0}$, as proved in Lemma~\ref{lemma:ripr}. Since $\theta^*\in\Theta_0$, $\hat{\theta}_1\in\Theta_1$, and $p_{\hat{\theta}_1 / \|\hat{\theta}_1\|}$ is the RIPR of $p_{\hat{\theta}_1}$ onto $\mathcal{P}_{\Theta_0}$, we know $\E_{\theta^*}\{p_{\hat{\theta}_1}(Y) / p_{\hat{\theta}_1 / \|\hat{\theta}_1\|}(Y)\} \leq 1$, as explained under \textit{Approach 3: Subsampled hybrid LRT} in the main text. So
\begin{align*}
\E_{\theta^*}\{R_n \one(\|\bar{Y}_1\| > 1) \mid \D_1 \} &= \E_{\theta^*}\left[\prod_{Y_i \in \D_0} \{p_{\hat{\theta}_1}(Y_i) / p_{\hat{\theta}_1 / \|\hat{\theta}_1\|}(Y_i) \}  \right] \\
&\overset{\text{iid}}{=} \prod_{i=1}^{n/2} \E_{\theta^*}\left\{ p_{\hat{\theta}_1}(Y_i) / p_{\hat{\theta}_1 / \|\hat{\theta}_1\|}(Y_i) \right\} \\
&\leq 1 \\
&= \one(\|\bar{Y}_1\| > 1)~.
\end{align*}
\end{proof}

\begin{theorem}
In the doughnut null hypothesis test setting, assume the subsampled test statistics $U_{n,b} = \L_{0,b}(\hat{\theta}_{1,b}) \: / \: \L_{0,b}(\hat{\theta}_{0,b}^\spl)$ and $R_{n,b} = \L_{0,b}(\hat{\theta}_{1,b}) \: / \: \L_{0,b}(\hat{\theta}_{0,b}^\RIPR)$, $1\leq b\leq B$. The test that rejects $H_0$ when $$\frac{1}{B} \sum_{b=1}^B \left\{U_{n,b} \one(\|\bar{Y}_{1,b}\| < 0.5) + \one(\|\bar{Y}_{1,b}\| \in [0.5, 1]) + R_{n,b} \one(\|\bar{Y}_{1,b}\| > 1) \right\} \geq 1/\alpha$$ is a valid level $\alpha$ test.
\end{theorem}

\begin{proof}
Assume $\theta^* \in \Theta_0$. The probability of falsely rejecting $H_0$ is
{\small
\begin{align*}
\P_{\theta^*}&\left[\frac{1}{B} \sum_{b=1}^B \left\{U_{n,b} \one(\|\bar{Y}_{1,b}\| < 0.5) + \one(\|\bar{Y}_{1,b}\| \in [0.5, 1]) + R_{n,b} \one(\|\bar{Y}_{1,b}\| > 1) \right\} \geq 1/\alpha \right] \\
&\leq \alpha \E_{\theta^*}\left[\frac{1}{B} \sum_{b=1}^B \left\{U_{n,b} \one(\|\bar{Y}_{1,b}\| < 0.5) + \one(\|\bar{Y}_{1,b}\| \in [0.5, 1]) + R_{n,b} \one(\|\bar{Y}_{1,b}\| > 1) \right\} \right] \\
&\leq \alpha \E_{\theta^*}\left[\frac{1}{B} \sum_{b=1}^B \left\{T_{n,b}(\theta^*) \one(\|\bar{Y}_{1,b}\| < 0.5) + \one(\|\bar{Y}_{1,b}\| \in [0.5, 1]) + R_{n,b} \one(\|\bar{Y}_{1,b}\| > 1) \right\} \right] \stepcounter{equation}\tag{\theequation}\label{eq:hybrid_valid_1} \\
&= \alpha \E_{\theta^*} \left\{T_{n}(\theta^*) \one(\|\bar{Y}_1\| < 0.5) + \one(\|\bar{Y}_{1}\| \in [0.5, 1]) + R_n \one(\|\bar{Y}_{1}\| > 1) \right\} \\
&= \alpha \E_{\theta^*}\left[ \E_{\theta^*}\left\{T_{n}(\theta^*) \one(\|\bar{Y}_{1}\| < 0.5) \mid \D_1 \right\} \right] + \alpha\P_{\theta^*}(\|\bar{Y}_1\| \in [0.5, 1]) + \alpha \E_{\theta^*}\left[\E_{\theta^*} \left\{R_{n} \one(\|\bar{Y}_{1}\| > 1) \mid \D_1 \right\} \right]   \\
&\leq \alpha \E_{\theta^*}\left[\one(\|\bar{Y}_{1}\| < 0.5) \E_{\theta^*}\left\{T_{n}(\theta^*)  \mid \D_1 \right\} \right] + \alpha\P_{\theta^*}(\|\bar{Y}_1\| \in [0.5, 1]) + \alpha \E_{\theta^*}\{\one(\|\bar{Y}_1\| > 1)\} \stepcounter{equation}\tag{\theequation}\label{eq:hybrid_valid_2} \\
&\leq \alpha \E_{\theta^*}\{\one(\|\bar{Y}_1\| <  0.5)\} + \alpha\P_{\theta^*}(\|\bar{Y}_1\| \in [0.5, 1]) + \alpha \P_{\theta^*}\{\one(\|\bar{Y}_{1}\| > 1) \} \stepcounter{equation}\tag{\theequation}\label{eq:hybrid_valid_3} \\
&= \alpha \left\{\P_{\theta^*}(\|\bar{Y}_1\| < 0.5) + \P_{\theta^*}(\|\bar{Y}_1\| \in [0.5, 1]) + \P_{\theta^*}(\|\bar{Y}_1\| > 1) \right\} \\
&= \alpha~.
\end{align*}
}
(\ref{eq:hybrid_valid_1}) holds because $\hat{\theta}_{0,b}^\spl = \argmax{\theta \in \Theta_0} \L_{0,b}(\theta)$. Since $\theta^* \in \Theta_0$, 
\begin{equation*}
U_{n,b} = \L_{0,b}(\hat{\theta}_{1}) / \L_{0,b}(\hat{\theta}_{0,b}^\spl) \leq \L_{0,b}(\hat{\theta}_{1}) / \L_{0,b}(\theta^*) = T_{n,b}(\theta^*)~.
\end{equation*}
(\ref{eq:hybrid_valid_2}) holds by Lemma~\ref{lemma:Rn_leq_1}. (\ref{eq:hybrid_valid_3}) holds because $\E_{\theta^*}\{T_n(\theta^*) \mid \D_1\} \leq 1$, as established by Theorem~\ref{thm:valid}.
\end{proof}

\section{Derivations of Equations} \label{sec:supp_eq}

\begin{proof}[Derivation of Equation~\ref{eq:C_n_usual}]
The usual likelihood ratio confidence set for $\theta^*\in\R^d$ is given by $$C_n^\LRT(\alpha) = \left\{\theta\in\Theta: 2\log \frac{\L(\bar{Y})}{\L(\theta)} \leq c_{\alpha,d} \right\},$$ where $c_{\alpha, d}$ is the upper $\alpha$ quantile of the $\chi^2_d$ distribution. $\bar{Y}$ is the sample mean of the $Y_i$ observations, and it is also the MLE estimate for $\theta^*$. We re-write this confidence set such that the squared radius of the set is apparent. Note that

\begin{align*}
2\log \frac{\L(\bar{Y})}{\L(\theta)} &= 2\log\left(\frac{\Pi_{i=1}^n \exp\left(-\frac{1}{2} (Y_i - \bar{Y})^T (Y_i - \bar{Y})\right)}{\Pi_{i=1}^n \exp\left(-\frac{1}{2} (Y_i - \theta)^T (Y_i - \theta)\right)} \right) \\
&= 2\log\left( \exp\left(-\frac{1}{2}\sum_{i=1}^n (Y_i - \bar{Y})^T (Y_i - \bar{Y}) + \frac{1}{2}\sum_{i=1}^n (Y_i-\theta)^T (Y_i-\theta) \right) \right) \\
&= - \sum_{i=1}^n (Y_i - \bar{Y})^T (Y_i - \bar{Y}) + \sum_{i=1}^n (Y_i-\theta)^T (Y_i-\theta) \\
&= \sum_{i=1}^n \left(-(Y_i - \bar{Y})^T (Y_i - \bar{Y}) + (Y_i - \bar{Y} + \bar{Y} - \theta)^T (Y_i - \bar{Y} + \bar{Y} - \theta) \right) \\
&= \sum_{i=1}^n \Bigg(-(Y_i - \bar{Y})^T (Y_i - \bar{Y}) + (Y_i-\bar{Y})^T(Y_i - \bar{Y}) + \\
&\hspace*{5em} 2(Y_i-\bar{Y})^T (\bar{Y}-\theta) + (\bar{Y}-\theta)^T (\bar{Y} - \theta) \Bigg) \\
&= n \|\bar{Y} - \theta\|^2~.
\end{align*}
The final step holds because the first two terms cancel and the summation over the third term equals 0. Therefore, $$C_n^\LRT(\alpha) = \left\{\theta\in \Theta: \|\theta - \bar{Y}\|^2 \leq c_{\alpha, d} / n\right\}.$$
This matches the set from equation~\ref{eq:C_n_usual}.
\end{proof}

\begin{proof}[Derivation of Equation~\ref{eq:C_n_split}]
Let $\hat{\theta}_1 = \bar{Y}_1$ be the sample mean of the $n/2$ observations in $\D_1$. Where $$T_n(\theta) = \frac{\L_0(\hat{\theta}_1)}{\L_0(\theta)}~,$$ the universal confidence set using the split likelihood ratio statistic is $$C_n^{\spl}(\alpha) = \left\{\theta\in\Theta : T_n(\theta) < \frac{1}{\alpha} \right\}.$$ We also re-write this confidence set such that the squared radius of the set is apparent. Note that

{\small
\begin{align*}
T_n(\theta) &= \frac{\Pi_{Y_i\in \D_0} \exp\left(-\frac{1}{2} (Y_{i}-\hat{\theta}_1)^T (Y_{i} - \hat{\theta}_1) \right)}{\Pi_{Y_i\in \D_0} \exp\left(-\frac{1}{2} (Y_{i} - \theta)^T (Y_{i}-\theta) \right) } \\
&= \exp\left(\sum_{Y_i\in \D_0} \left(-\frac{1}{2}(Y_{i}-\bar{Y}_1)^T(Y_{i}-\bar{Y}_1) + \frac{1}{2}(Y_{i}-\theta)^T (Y_{i}-\theta) \right) \right) \\
&= \exp\Bigg(\sum_{Y_i\in \D_0} \Bigg(-\frac{1}{2}(Y_{i} - \bar{Y}_0 + \bar{Y}_0 - \bar{Y}_1)^T(Y_{i}- \bar{Y}_0 + \bar{Y}_0 - \bar{Y}_1) + \\
&\hspace*{7em} \frac{1}{2}(Y_{i}- \bar{Y}_0 + \bar{Y}_0-\theta)^T (Y_{i} - \bar{Y}_0 + \bar{Y}_0 -\theta) \Bigg) \Bigg) \\
&= \exp\Bigg( \sum_{Y_i\in \D_0} \Bigg( -\frac{1}{2}\left[(Y_{i}-\bar{Y}_0)^T(Y_{i}-\bar{Y}_0) + 2(Y_{i}-\bar{Y}_0)^T(\bar{Y}_0-\bar{Y}_1) + (\bar{Y}_0-\bar{Y}_1)^T(\bar{Y}_0-\bar{Y}_1) \right] + \\
&\hspace*{8em} \frac{1}{2}\left[(Y_{i}-\bar{Y}_0)^T(Y_{i}-\bar{Y}_0) + 2(Y_{i}-\bar{Y}_0)^T(\bar{Y}_0-\theta) + (\bar{Y}_0-\theta)^T(\bar{Y}_0-\theta) \right] \Bigg) \Bigg) \stepcounter{equation}\tag{\theequation} \label{eq:split_deriv_1} \\ 
&= \exp\left(\sum_{Y_i\in \D_0} \left(-\frac{1}{2}(\bar{Y}_0-\bar{Y}_1)^T(\bar{Y}_0-\bar{Y}_1) + \frac{1}{2}(\bar{Y}_0-\theta)^T(\bar{Y}_0-\theta) \right) \right) \\
&= \exp\left(-\frac{n}{4} \|\bar{Y}_0 - \bar{Y}_1\|^2 + \frac{n}{4} \|\bar{Y}_0 - \theta\|^2 \right). \stepcounter{equation}\tag{\theequation} \label{eq:split_deriv_2} 
\end{align*}
} 
The first and fourth terms of (\ref{eq:split_deriv_1}) cancel, and the cross-product terms equal 0 upon taking the summation. (\ref{eq:split_deriv_2}) holds because $\D_0$ contains $n/2$ elements. Therefore,
\begin{align*}
C_n^{\spl}(\alpha) &= \left\{\theta\in\Theta: T_n(\theta) < \frac{1}{\alpha} \right\} \\
&= \left\{\theta\in\Theta: \exp\left(-\frac{n}{4} \|\bar{Y}_0 - \bar{Y}_1\|^2 + \frac{n}{4} \|\bar{Y}_0 - \theta\|^2 \right) < \frac{1}{\alpha} \right\} \\
&= \left\{\theta\in\Theta: -\frac{n}{4} \|\bar{Y}_0 - \bar{Y}_1\|^2 + \frac{n}{4} \|\bar{Y}_0 - \theta\|^2 < \log\left(\frac{1}{\alpha}\right) \right\} \\
&= \left\{\theta\in\Theta:  \frac{n}{4} \|\bar{Y}_0 - \theta\|^2 < \log\left(\frac{1}{\alpha}\right) + \frac{n}{4} \|\bar{Y}_0 - \bar{Y}_1\|^2 \right\} \\
&= \left\{\theta\in\Theta:  \|\bar{Y}_0 - \theta\|^2 < \frac{4}{n} \log\left(\frac{1}{\alpha}\right) + \|\bar{Y}_0 - \bar{Y}_1\|^2 \right\},
\end{align*}
which concludes our derivation of equation~\ref{eq:C_n_split}.
\end{proof}

\begin{proof}[Derivation of Equation~\ref{eq:sq_rad_split}]
From the definition of $C_n^\spl(\alpha)$, we see that $r^2(C_n^\spl(\alpha)) = (4/n)\log(1/\alpha) + \|\bar{Y}_0 - \bar{Y}_1\|^2$. Let $Y_{0i}$ and $Y_{1i}$, $i = 1, \ldots, n/2$, represent the observations in $\D_0$ and $\D_1$. Note that 
\begin{align*}
\|\bar{Y}_0 - \bar{Y}_1\|^2 &= \left\| \frac{2}{n} \sum_{i=1}^{n/2} (Y_{0i} - Y_{1i}) \right\|^2 = \frac{4}{n} \left\| \frac{1}{\sqrt{n}} \sum_{i=1}^{n/2} (Y_{0i} - Y_{1i}) \right\|^2 \overset{d}{=} \frac{4}{n} \chi^2_d~.
\end{align*}
To see why the last step holds, note that $Y_1, \ldots, Y_n \overset{\text{iid}}{\sim} N(\theta^*, I_d)$. So for any $i$, ${Y_{0i} - Y_{1i} \overset{\text{iid}}{\sim} N(0, 2I_d)}$. Then $\sum_{i=1}^{n/2} (Y_{0i} - Y_{1i}) \overset{\text{iid}}{\sim} N\left(0, \frac{n}{2} (2I_d)\right)$, and $\frac{1}{\sqrt{n}}\sum_{i=1}^{n/2} (Y_{0i} - Y_{1i}) \overset{\text{iid}}{\sim} N(0, I_d)$. This implies that ${r^2(C_n^{\spl}(\alpha)) \overset{d}{=} (4/n)\log(1/\alpha) + (4/n)\chi_d^2}$. Therefore, $\E[r^2(C_n^{\spl}(\alpha))] = (4/n)\log\left(1/\alpha\right) + (4/n)d.$
\end{proof}

\begin{proof}[Derivation of Equation~\ref{eq:ratio_sq_rad_1}]
From equation \ref{eq:split_usual}, we know that 
\begin{align*}
\frac{\E\left[r^2(C_n^{\spl}(\alpha))\right]}{r^2(C_n^\LRT(\alpha))} &= \frac{4\log(1/\alpha) + 4d}{c_{\alpha,d}}~.
\end{align*}
For $d\geq 1$ and $\alpha \in (0,1)$, \cite{inglot2010inequalities} shows the upper bound 
$$c_{\alpha, d} \leq d + 2\log\left(\frac{1}{\alpha} \right) + 2\sqrt{d\log\left(\frac{1}{\alpha} \right)}~.$$
Also, for $d\geq 2$ and $\alpha \leq 0.17$, \cite{inglot2010inequalities} shows the lower bound
$$c_{\alpha, d} \geq d + 2\log\left(\frac{1}{\alpha} \right) - \frac{5}{2}~$$ Combining these facts, we see that for $d\geq 2$ and $\alpha \leq 0.17$,
\begin{align*}
\frac{4\log(1/\alpha) + 4d}{2\log(1/\alpha) + d + 2\sqrt{d\log(1/\alpha)}} \leq \frac{\E\left[r^2(C_n^{\spl}(\alpha))\right]}{r^2(C_n^\LRT(\alpha))} \leq \frac{4\log(1/\alpha) + 4d}{2\log(1/\alpha) + d - 5/2}~.
\end{align*}
This concludes the derivation of equation~\ref{eq:ratio_sq_rad_1}.
\end{proof}

\begin{proof}[Derivation of Equation~\ref{eq:ratio_sq_rad_2}]
From equation~\ref{eq:split_usual}, we know that 
\begin{align*}
\frac{\E\left[r^2(C_n^{\spl}(\alpha))\right]}{r^2(C_n^\LRT(\alpha))} &= \frac{4\log(1/\alpha) + 4d}{c_{\alpha,d}}~.
\end{align*}
The lower bound of equation~\ref{eq:ratio_sq_rad_2} is the same as the lower bound from equation~\ref{eq:ratio_sq_rad_1}. We consider the upper bound. Suppose $d=1$ and $\alpha \leq \exp\left(-\frac{5(1+\sqrt{5})}{4} \right)$. Let $t = -2 + \sqrt{5 + 2\log(1/\alpha)}$. We will show that $c_{\alpha, 1} \geq t^2$ in several steps:
\\[12pt]
\noindent \textit{Step 1:} Show that $t^2 + 4t - 2 < 2\log(1/\alpha)$.
\begin{align*}
t^2 + 4t - 2 &= \left(-2 + \sqrt{5 + 2\log(1/\alpha)} \right)^2 + 4(-2 + \sqrt{5 + 2\log(1/\alpha)}) - 2 \\
&= 4 - 4 \sqrt{5 + 2\log(1/\alpha)} + 5 + 2\log(1/\alpha) - 8 + 4\sqrt{5 + 2\log(1/\alpha)} - 2 \\
&= 2\log(1/\alpha) - 1 \\
&< 2\log(1/\alpha)~.
\end{align*}
\\[12pt]
\noindent \textit{Step 2:} Show that $\log(1/\alpha) > t^2/2 + 2\log(t) + \log(\sqrt{2\pi})$.
Starting with the result from Step 1,
\begin{align*}
\log(1/\alpha) &> \frac{t^2}{2} + 2t - 1 \\
&\geq \frac{t^2}{2} + 2(\log(t) + 1) - 1 \qquad \text{since $t\geq \log(t) + 1$ for $t>0$} \\
&= \frac{t^2}{2} + 2\log(t) + 1 \\
&> \frac{t^2}{2} + 2\log(t) + \log(\sqrt{2\pi})~.
\end{align*}
\\[12pt]
\noindent \textit{Step 3:} Show that $t^2 - 1 \geq t$. We start by showing that $t \geq \frac{1}{2}(1+\sqrt{5})$ follows from our definitions of $t$ and $\alpha$. Since $$\alpha \leq \exp\left(-\frac{5(1+\sqrt{5})}{4}\right),$$ it holds that $$\frac{1}{\alpha} \geq \exp\left(\frac{5(1+\sqrt{5})}{4}\right).$$
Then $8\log(1/\alpha) \geq 10(1+\sqrt{5})$, which implies $4(5+2\log(1/\alpha)) \geq 25 + 10\sqrt{5} + 5$.
Taking the square root of both sides, $$2\sqrt{5 + 2\log(1/\alpha)} \geq 5 + \sqrt{5}.$$ Since we set $t = -2 + \sqrt{5 + 2\log(1/\alpha)}$, this implies $$t \geq \frac{1}{2}(1+\sqrt{5})~.$$

The roots of the convex function $t^2 - t - 1$ are at $t = (1 \pm \sqrt{5})/2$. At $t \geq (1/2)(1 + \sqrt{5})$, we know $t^2 - 1 \geq t$.
\\[12pt]
\noindent \textit{Step 4:} Show that $t^2 \leq c_{\alpha, 1}$. Starting with the results of steps 2 and 3,
\begin{align*}
\log(t^2 - 1) - t^2/2 - \log(\sqrt{2\pi}) &> \log(t^2 - 1) + 2\log(t) + \log(\alpha) \\
&\geq 3\log(t) + \log(\alpha)~.
\end{align*}
Exponentiating, 
\begin{align*}
\left(t^2 - 1\right) \exp\left(-t^2 / 2\right) \left(\frac{1}{\sqrt{2\pi}} \right) &\geq t^3 \alpha~.
\end{align*}
So 
\begin{align*}
\left(\frac{1}{t} - \frac{1}{t^3}\right) \exp\left(-t^2 / 2\right) \left(\frac{1}{\sqrt{2\pi}} \right) &\geq \alpha~.
\end{align*}
If $Z\sim N(0,1)$ and $X = Z^2 \sim \chi^2_1$, then using an inequality on $\P(Z\geq t)$ from Section 2.1 of \cite{pollard2015} and Section 7.1 of \cite{feller1968probability}, 
\begin{align*}
\P(X\geq t^2) = 2\P(Z \geq t) > \P(Z \geq t) \geq \left(\frac{1}{t} - \frac{1}{t^3}\right) \exp\left(-t^2 / 2\right) \left(\frac{1}{\sqrt{2\pi}} \right) \geq \alpha~.
\end{align*}
This implies that $c_{\alpha, 1} \geq t^2 = 2\log(1/\alpha) + 9 - 4\sqrt{5 + 2\log(1/\alpha)}$. Then for $d = 1$ and $\alpha \leq \exp\left(-\frac{5(1+\sqrt{5})}{4} \right),$ 
\begin{align*}
\frac{4\log(1/\alpha) + 4d}{2\log(1/\alpha) + d + 2\sqrt{d\log(1/\alpha)}} \leq \frac{\E[r^2\{C_n^{\spl}(\alpha)\}]}{r^2\{C_n^\LRT(\alpha)\}} \leq \frac{4\log(1/\alpha) + 4d}{2\log(1/\alpha) + 9 - 4\sqrt{5 + 2\log(1/\alpha)}}~.
\end{align*}
Since we are working with $d=1$, we conclude that
\begin{align*}
\frac{4\log(1/\alpha) + 4}{2\log(1/\alpha) + 1 + 2\sqrt{\log(1/\alpha)}} \leq \frac{\E[r^2\{C_n^{\spl}(\alpha)\}]}{r^2\{C_n^\LRT(\alpha)\}} \leq \frac{4\log(1/\alpha) + 4}{2\log(1/\alpha) + 9 - 4\sqrt{5 + 2\log(1/\alpha)}}~,
\end{align*}
as claimed.
\end{proof}

\begin{proof}[Derivation of Equation~\ref{eq:power_usual}]
The classical LRT set is $$C_n^\LRT(\alpha) = \left\{\theta\in\Theta: \|\bar{Y} - \theta\|^2 \leq c_{\alpha,d} \: / \: n\right\},$$ where $c_{\alpha, d}$ is the upper $\alpha$ quantile of the $\chi^2_d$ distribution. Suppose we are testing $H_0: \theta^* = 0$ versus $H_1: \theta^* \neq 0$. The power of the classical LRT at the true $\theta^*$ is thus $$\power(C_n^\LRT(\alpha); \theta^*) = \P_{\theta^*}\left(\|\bar{Y}\|^2 >  c_{\alpha,d} / n \right).$$ We can express the power function of the classical LRT in terms of the CDF of a noncentral $\chi^2$ distribution. Let us denote $\theta^* = (\theta^*_1, \ldots, \theta^*_d)$. We see that 
\begin{align*}
n\|\bar{Y}\|^2 &=  \left\| \frac{1}{\sqrt{n}} \sum_{i=1}^n Y_i \right\|^2 =\sum_{j=1}^d \left(\frac{1}{\sqrt{n}} \sum_{i=1}^n Y_{ij} \right)^2.
\end{align*}
For each dimension $j$, $n^{-1/2} \sum_{i=1}^n Y_{ij}\sim N(\theta^*_j\sqrt{n}, 1)$. So this follows a non-central $\chi^2$ distribution given by
\begin{align*}
n\|\bar{Y}\|^2 &\overset{d}{=} \chi^2\left(df = d, \lambda = \sum_{j=1}^d n(\theta^*_j)^2 \right) \overset{d}{=} \chi^2\left(df = d, \lambda = n \|\theta^*\|^2 \right).
\end{align*}

Let $\Phi(\cdot)$ represent the standard normal CDF. Suppose $X\sim \chi^2(df = d, \lambda = n\|\theta^*\|^2)$. As $d\to\infty$ or as $\lambda \to\infty$, it holds that $$\frac{X - (d + n\|\theta^*\|^2)}{\sqrt{2(d + 2n\|\theta^*\|^2)}} \approx N(0,1).$$ See \cite{chun2009normal}. Using the Normal approximation to the non-central chi-squared CDF, the power of the classical LRT is
\begin{align*}
\power(C_n^\LRT(\alpha); \theta^*) &= \P_{\theta^*}\left(\|\bar{Y}\|^2 >  c_{\alpha,d} / n \right) \\
&= \P_{\theta^*}\left(n\|\bar{Y}\|^2 >  c_{\alpha,d} \right) \\
&= \P_{\theta^*}\left(\frac{n\|\bar{Y}\|^2 - d - n\|\theta^*\|^2}{\sqrt{2(d + 2n\|\theta^*\|^2)}} > \frac{c_{\alpha, d} - d - n\|\theta^*\|^2}{\sqrt{2(d + 2n\|\theta^*\|^2)}} \right) \\
&\approx 1 - \Phi\left(\frac{c_{\alpha, d} - d - n\|\theta^*\|^2}{\sqrt{2(d + 2n\|\theta^*\|^2)}} \right) \\
&= \Phi\left(\frac{d + n\|\theta^*\|^2 - c_{\alpha, d}}{\sqrt{2(d + 2n\|\theta^*\|^2)}} \right).
\end{align*}
This matches the expression from equation~\ref{eq:power_usual}.
\end{proof}

\begin{proof}[Derivation of Equation~\ref{eq:power_subsplit}]
Using methods from the derivation of equation \ref{eq:power_usual}, we can find a representation for the approximate power of the limiting subsampling LRT set as $B \to \infty$. From equation \ref{eq:subsplit_radius}, 
\begin{align*}
C_n^{\subsplit}(\alpha) \approx \left\{\theta\in\Theta: \|\bar{Y}-\theta\|^2 < \frac{10}{3n} \log\left(\left(\frac{5}{2}\right)^{d/2} \frac{1}{\alpha} \right) \right\}
\end{align*}
So the power of the limit of subsampling LRT for a test of $H_0: \theta^* = 0$ versus $H_1: \theta^* \neq 0$ is
\begin{align*}
\power(C_n^{\subsplit}(\alpha); \theta^*) &\approx \P_{\theta^*}\left(n\|\bar{Y}\|^2 \geq \frac{10}{3} \log\left(\left(\frac{5}{2}\right)^{d/2} \frac{1}{\alpha} \right) \right) \\
&= \P_{\theta^*}\left(\frac{n\|\bar{Y}\|^2 - d - n\|\theta^*\|^2}{\sqrt{2(d + 2n\|\theta^*\|^2)}} \geq \frac{(10/3)\log\left((5/2)^{d/2} (1/\alpha)\right) - d - n\|\theta^*\|^2}{\sqrt{2(d + 2n\|\theta^*\|^2)}} \right) \\
&\approx \Phi\left(\frac{1}{\sqrt{2(d + 2n\|\theta^*\|^2)}} \left[d + n\|\theta^*\|^2 - \frac{10}{3} \log\left\{\left(\frac{5}{2}\right)^{d/2} \frac{1}{\alpha} \right\} \right] \right).
\end{align*}
This matches the expression from equation~\ref{eq:power_subsplit}.
\end{proof}

\section{Convexity of confidence sets} \label{sec:convex}

We show that $C_n^\LRT(\alpha)$, $C_n^\spl(\alpha)$, $C_n^\CF(\alpha)$, and $C_n^\subsplit(\alpha)$ are convex sets.

\subsection{$C_n^\LRT(\alpha)$ is a convex set.}

Suppose $\theta_1 \in C_n^{\LRT}(\alpha)$ and $\theta_2 \in C_n^{\LRT}(\alpha)$. Then $\|\theta_1 - \bar{Y}\|^2 \leq c_{\alpha,d}/n$ and $\|\theta_2 - \bar{Y}\|^2 \leq c_{\alpha,d}/n$. Fix $t\in (0,1)$, and let $\theta_3 = t\theta_1 + (1-t)\theta_2$. Since $\|\cdot\|^2$ is convex, we see
\begin{align*}
\|\theta_3 - \bar{Y}\|^2 &= \|t\theta_1 - t\bar{Y} + (1-t)\theta_2 - (1-t)\bar{Y}\|^2 \\
&\leq t \|\theta_1 - \bar{Y}\|^2 + (1-t)  \|\theta_2 - \bar{Y}\|^2 \\
&\leq t \left(c_{\alpha,d}/n\right) + (1-t) \left(c_{\alpha,d}/n\right) \\
&= c_{\alpha,d}/n~.
\end{align*}
We conclude that $\theta_3 \in C_n^{\LRT}(\alpha)$, so $C_n^{\LRT}(\alpha)$ is convex.

\subsection{$C_n^\spl(\alpha)$ and $C_n^\subsplit(\alpha)$ are convex sets.}

Since $C_n^\spl(\alpha)$ is the same as $C_n^\subsplit(\alpha)$ with $B = 1$, it suffices to show that $C_n^\subsplit(\alpha)$ is a convex set. Recall that
\begin{equation*}
C_n^{\subsplit}(\alpha) = \left\{\theta \in \Theta: \frac{1}{B} \sum_{b=1}^{B} \exp\left(-\frac{n}{4}\| \bar{Y}_{0,b} - \bar{Y}_{1,b}\|^2 + \frac{n}{4}\| \bar{Y}_{0,b} - \theta\|^2 \right) < \frac{1}{\alpha} \right\}.
\end{equation*}
Suppose $\theta_1 \in C_n^{\subsplit}(\alpha)$ and $\theta_2 \in C_n^{\subsplit}(\alpha)$. Fix $t\in (0,1)$, and let $\theta_3 = t\theta_1 + (1-t)\theta_2$. We see
{\small
\begin{align*}
\frac{1}{B} \sum_{b=1}^{B} \exp &\left(-\frac{n}{4}\| \bar{Y}_{0,b} - \bar{Y}_{1,b}\|^2 + \frac{n}{4}\| \bar{Y}_{0,b} - \theta_3\|^2 \right) \\
&= \frac{1}{B} \sum_{b=1}^{B} \exp \left(-\frac{n}{4}\| \bar{Y}_{0,b} - \bar{Y}_{1,b}\|^2 + \frac{n}{4}\| \bar{Y}_{0,b} - t \theta_1 - (1-t)\theta_2 \|^2 \right) \\
&= \frac{1}{B} \sum_{b=1}^{B} \exp \left(-\frac{n}{4}\| \bar{Y}_{0,b} - \bar{Y}_{1,b}\|^2 + \frac{n}{4}\| t (\bar{Y}_{0,b} - \theta_1) + (1-t)(\bar{Y}_{0,b} - \theta_2) \|^2 \right) \\
&\leq \frac{1}{B} \sum_{b=1}^{B} \exp \left(-\frac{n}{4}\| \bar{Y}_{0,b} - \bar{Y}_{1,b}\|^2 + \frac{n}{4} t \| \bar{Y}_{0,b} - \theta_1 \|^2 + \frac{n}{4} (1-t) \|\bar{Y}_{0,b} - \theta_2 \|^2 \right)  \stepcounter{equation}\tag{\theequation}\label{eq:convex1} \\
&= \frac{1}{B} \sum_{b=1}^{B} \left[ \exp \left(-\frac{n}{4} t \| \bar{Y}_{0,b} - \bar{Y}_{1,b}\|^2 \right) \exp\left( \frac{n}{4} t \| \bar{Y}_{0,b} - \theta_1 \|^2\right)\right] \times \\
&\hspace{4.5em} \left[ \exp \left(-\frac{n}{4} (1-t) \| \bar{Y}_{0,b} - \bar{Y}_{1,b}\|^2 \right)  \exp \left( \frac{n}{4} (1-t) \|\bar{Y}_{0,b} - \theta_2 \|^2 \right) \right] \\
&\leq \frac{1}{B} \left[\sum_{b=1}^B \exp \left(-\frac{n}{4} \| \bar{Y}_{0,b} - \bar{Y}_{1,b}\|^2 \right) \exp\left( \frac{n}{4} \| \bar{Y}_{0,b} - \theta_1 \|^2\right) \right]^t \times \\
&\hspace{4.5em} \left[\sum_{b=1}^B \exp \left(-\frac{n}{4} \| \bar{Y}_{0,b} - \bar{Y}_{1,b}\|^2 \right)  \exp \left( \frac{n}{4}  \|\bar{Y}_{0,b} - \theta_2 \|^2 \right) \right]^{1-t} \stepcounter{equation}\tag{\theequation}\label{eq:convex2}  \\
&\leq \frac{1}{B} \left(\frac{B}{\alpha}\right)^t \left(\frac{B}{\alpha}\right)^{1-t} \stepcounter{equation}\tag{\theequation}\label{eq:convex3}  \\
&= \frac{1}{\alpha}~.
\end{align*}
}
Inequality (\ref{eq:convex1}) holds because $\|\cdot\|^2$ is convex. Inequality (\ref{eq:convex2}) is due to H{\"o}lder's inequality with $p=\frac{1}{t}$ and $q=\frac{1}{1-t}$. (Hence,  $\frac{1}{p} + \frac{1}{q} = 1$.) Inequality (\ref{eq:convex3}) is true based on the initial assumption that $\theta_1 \in C_n^{\subsplit}(\alpha)$ and $\theta_2 \in C_n^{\subsplit}(\alpha)$. We conclude that $\theta_3 \in C_n^{\subsplit}(\alpha)$, so $C_n^{\subsplit}(\alpha)$ is convex.

\subsection{$C_n^\CF(\alpha)$ is a convex set.}

Recall that
\begin{equation*}
C_n^\CF(\alpha) = \Bigg\{\theta\in\Theta : \exp\left(\frac{n}{4}\|\bar{Y}_0 - \theta\|^2\right) + \exp\left(\frac{n}{4}\|\bar{Y}_1 - \theta\|^2 \right) < \frac{2}{\alpha} \exp\left(\frac{n}{4}\|\bar{Y}_0 - \bar{Y}_1\|^2\right) \Bigg\}~.
\end{equation*}
Suppose $\theta_1 \in C_n^{\CF}(\alpha)$ and $\theta_2 \in C_n^{\CF}(\alpha)$. Fix $t\in (0,1)$, and let $\theta_3 = t\theta_1 + (1-t)\theta_2$. Since \mbox{$\|\cdot\|^2$} and $\exp(\cdot)$ are convex, we see
\begin{align*}
\exp&\left(\frac{n}{4}\|\bar{Y}_0 - \theta_3\|^2\right) + \exp\left(\frac{n}{4}\|\bar{Y}_1 - \theta_3\|^2 \right) \\
&= \exp\left(\frac{n}{4}\|t(\bar{Y}_0 - \theta_1) + (1-t)(\bar{Y}_0 - \theta_2)\|^2\right) + \exp\left(\frac{n}{4}\|t(\bar{Y}_1 - \theta_1) + (1-t)(\bar{Y}_1 - \theta_2)\|^2 \right) \\
&\leq \exp\left(\frac{n}{4} t \|\bar{Y}_0 - \theta_1\|^2 + \frac{n}{4}(1-t) \|\bar{Y}_0 - \theta_2\|^2\right) + \exp\left(\frac{n}{4} t \|\bar{Y}_1 - \theta_1\|^2 + \frac{n}{4}(1-t) \|\bar{Y}_1 - \theta_2\|^2\right) \\
&\leq t \exp\left(\frac{n}{4} \|\bar{Y}_0 - \theta_1\|^2\right) + (1-t) \exp\left(\frac{n}{4} \|\bar{Y}_0 - \theta_2\|^2\right) + \\
&\qquad t \exp\left(\frac{n}{4} \|\bar{Y}_1 - \theta_1\|^2\right) + (1-t) \exp\left( \frac{n}{4} \|\bar{Y}_1 - \theta_2\|^2\right) \\
&= t \left\{\exp\left(\frac{n}{4} \|\bar{Y}_0 - \theta_1\|^2\right) + \exp\left(\frac{n}{4} \|\bar{Y}_1 - \theta_1\|^2\right) \right\} + \\
&\qquad (1-t) \left\{\exp\left(\frac{n}{4} \|\bar{Y}_0 - \theta_2\|^2\right) + \exp\left( \frac{n}{4} \|\bar{Y}_1 - \theta_2\|^2\right) \right\} \\
&< t \left\{\frac{2}{\alpha} \exp\left(\frac{n}{4}\|\bar{Y}_0 - \bar{Y}_1\|^2\right) \right\} + (1-t) \left\{\frac{2}{\alpha} \exp\left(\frac{n}{4}\|\bar{Y}_0 - \bar{Y}_1\|^2\right) \right\} \\
&= \frac{2}{\alpha} \exp\left(\frac{n}{4}\|\bar{Y}_0 - \bar{Y}_1\|^2\right).
\end{align*}
We conclude that $\theta_3 \in C_n^{\CF}(\alpha)$, so $C_n^{\CF}(\alpha)$ is convex.

\section{Simulated cross-fit sets with varying $p_0$} \label{sec:supp_sim_CF}

Under general $p_0$, the cross-fit set is defined as 
\begin{align*}
C_n^{\CF}(\alpha) &= \Bigg\{\theta\in\Theta : \frac{1}{2}\Bigg[ \exp\left(-\frac{np_0}{2}\|\bar{Y}_0 - \bar{Y}_1\|^2 + \frac{np_0}{2}\|\bar{Y}_0 - \theta\|^2 \right) + \\
&\qquad \qquad \exp\left(-\frac{n(1-p_0)}{2} \|\bar{Y}_0 - \bar{Y}_1\|^2 + \frac{n(1-p_0)}{2} \|\bar{Y}_1 - \theta\|^2 \right) \Bigg] < \frac{1}{\alpha} \Bigg\}~. \stepcounter{equation}\tag{\theequation} \label{eq:crossfit_p0} 
\end{align*}
Note that at any $\theta$, the test statistic that defines $C_n^{\CF}(\alpha)$ has a distribution that is symmetric around $p_0 = 0.5$. Hence, the test statistic has the same distribution at $p_0$ and $1-p_0$ for any $p_0 \in (0, 0.5]$. Figure~\ref{fig:crossfit_p0} presents examples of cross-fit sets at $p_0 \in \{0.1, 0.3, 0.5\}$ on a single sample of 1000 observations simulated from a $N(0, I_2)$ distribution. In this example, we see that the region with $p_0 = 0.5$ has the smallest diameter and area.

\begin{figure}
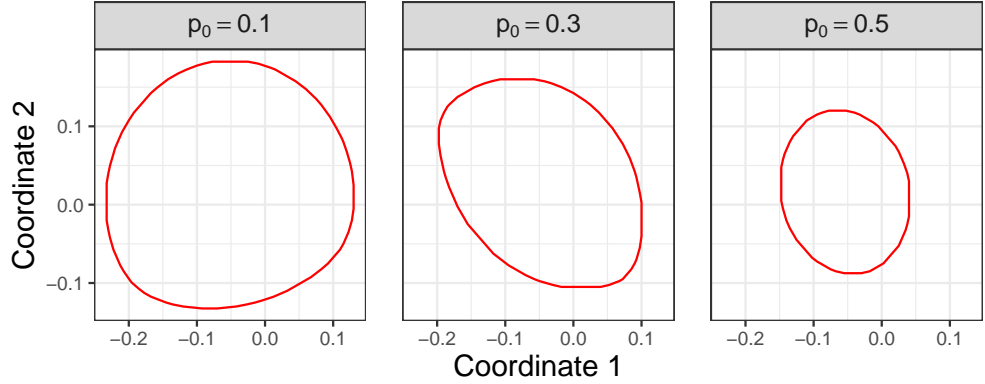

\figuresize{.8}
\figurebox{20pc}{25pc}{}[SuppFigure2.pdf]
\caption{Simulated cross-fit regions at varying $p_0$, using a single data sample.}
\label{fig:crossfit_p0}
\end{figure}

Based on Figure~\ref{fig:crossfit_p0} and the symmetry around $p_0 = 0.5$ in (\ref{eq:crossfit_p0}), we conjecture that $p_0 = 0.5$ will minimize the expected size of the cross-fit sets. We also conduct more extensive simulations at $d = 2$ to study the area of the cross-fit sets. To produce Figure~\ref{fig:CF_area_by_p0}, we simulate 100 datasets of 1000 $N(0, I_2)$ observations. We construct cross-fit sets for $p_0 \in \{0.05, 0.10, \ldots, 0.95\}$ on each dataset. We compute the area of each set by evaluating the cross-fit test statistic over a two-dimensional grid of $\theta$ values, checking which points are inside the set, constructing the convex hull of these points, and computing the area of the convex hull. (From Section~\ref{sec:convex}, recall that the cross-fit set is itself convex.) The average area is minimized at $p_0 = 0.5$.

\begin{figure}
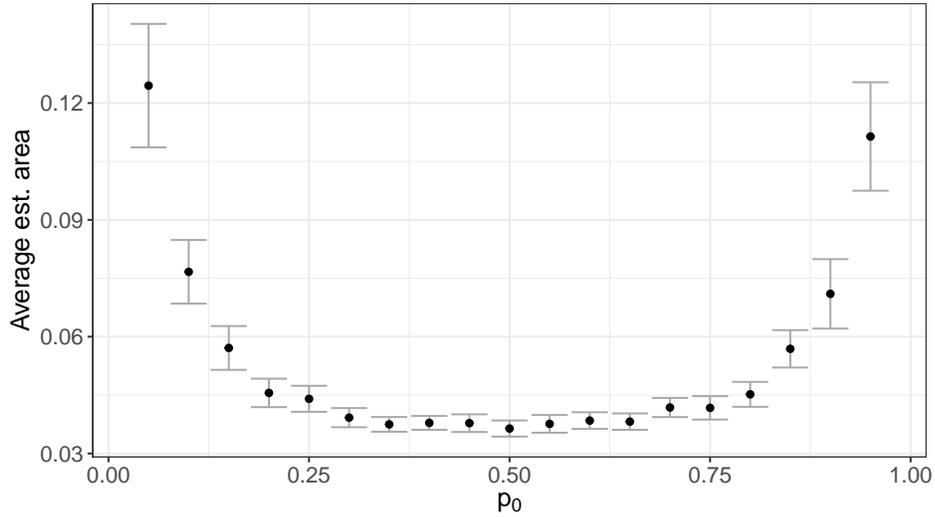

\figuresize{.7}
\figurebox{20pc}{25pc}{}[SuppFigure3.pdf]
\caption{Simulated cross-fit area at $\alpha = 0.1$, $d = 2$, and varying $p_0$. We simulate 100 datasets of 1000 $N(0, I_2)$ observations and construct cross-fit sets for $p_0 \in \{0.05, 0.10, \ldots, 0.95\}$ on each dataset. Where $\hat{\mu}_{p_0}$ is the average simulated area at $p_0$ and $\hat{\sigma}_{p_0}$ is the standard deviation of the simulated area of $p_0$ over 100 simulations, the error bars represent the confidence intervals $[\hat{\mu}_{p_0} - 1.96 \hat{\sigma}_{p_0} / \sqrt{100}, \hat{\mu}_{p_0} + 1.96 \hat{\sigma}_{p_0} / \sqrt{100}]$. Choosing $p_0 = 0.5$ minimizes the average simulated area.}
\label{fig:CF_area_by_p0}
\end{figure}

Figures~\ref{fig:CF_area_by_sq_dist} and \ref{fig:CF_area_by_sq_dist_and_p0} provide another perspective on these simulations by plotting the estimated area against $\|\bar{Y}_0 - \bar{Y}_1\|^2$. Figure~\ref{fig:CF_area_by_sq_dist} aggregates the results across all $p_0$, while Figure~\ref{fig:CF_area_by_sq_dist_and_p0} plots the results separately for each $p_0$. In both settings, we see that smaller values of $p_0$ are associated with smaller estimated areas. By (\ref{eq:p0_chisq}), we know that
$$\E\left[\|\bar{Y}_0 - \bar{Y}_1\|^2\right] = \left(\frac{1}{np_0} + \frac{1}{n(1-p_0)} \right)d~.$$
Then 
$$\frac{\partial}{\partial p_0} \E\left[\|\bar{Y}_0 - \bar{Y}_1\|^2\right] = \frac{-d}{np_0^2} + \frac{d}{n(1-p_0)^2}~,$$
which equals 0 at $p_0 = 0.5$. In addition, for all $p_0 \in (0, 1)$,
$$\frac{\partial^2}{\partial p_0^2} \E\left[\|\bar{Y}_0 - \bar{Y}_1\|^2\right] = \frac{2d}{np_0^3} + \frac{2d}{n(1-p_0)^3} > 0~.$$
Hence $\E\left[\|\bar{Y}_0 - \bar{Y}_1\|^2\right]$ is minimized at $p_0 = 0.5$. We have observed that smaller $\|\bar{Y}_0 - \bar{Y}_1\|^2$ is associated with smaller cross-fit set area at $d=2$, and we have shown that $\E\left[\|\bar{Y}_0 - \bar{Y}_1\|^2\right]$ is minimized at $p_0 = 0.5$. Together, these facts provide additional evidence in favor of the optimality of $p_0 = 0.5$ for the cross-fit case. 

\begin{figure}
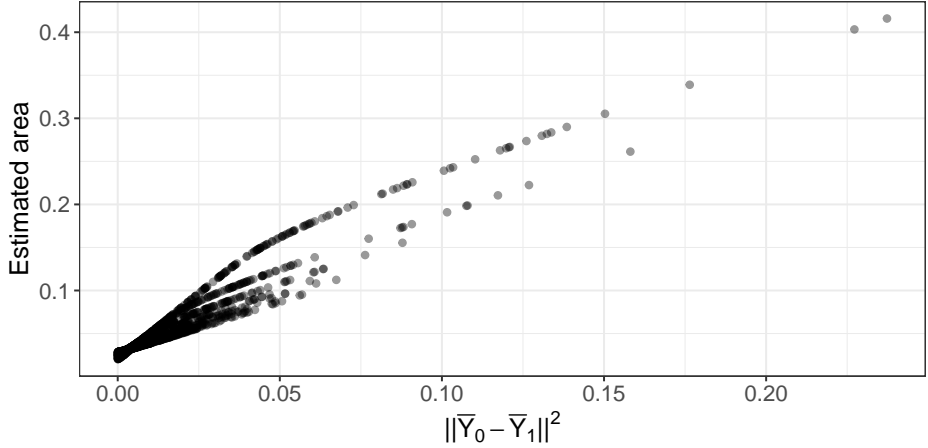

\figuresize{.7}
\figurebox{20pc}{25pc}{}[SuppFigure4.pdf]
\caption{Using the same simulations as Figure~\ref{fig:CF_area_by_p0}, we plot the simulated cross-fit area against \mbox{$\|\bar{Y}_0 - \bar{Y}_1\|^2$.} Hence, this includes simulations across all $p_0 \in \{0.05, 0.10, \ldots, 0.95\}$. Lower values of $\|\bar{Y}_0 - \bar{Y}_1\|^2$ are associated with smaller estimated areas.}
\label{fig:CF_area_by_sq_dist}
\end{figure}

\begin{figure}
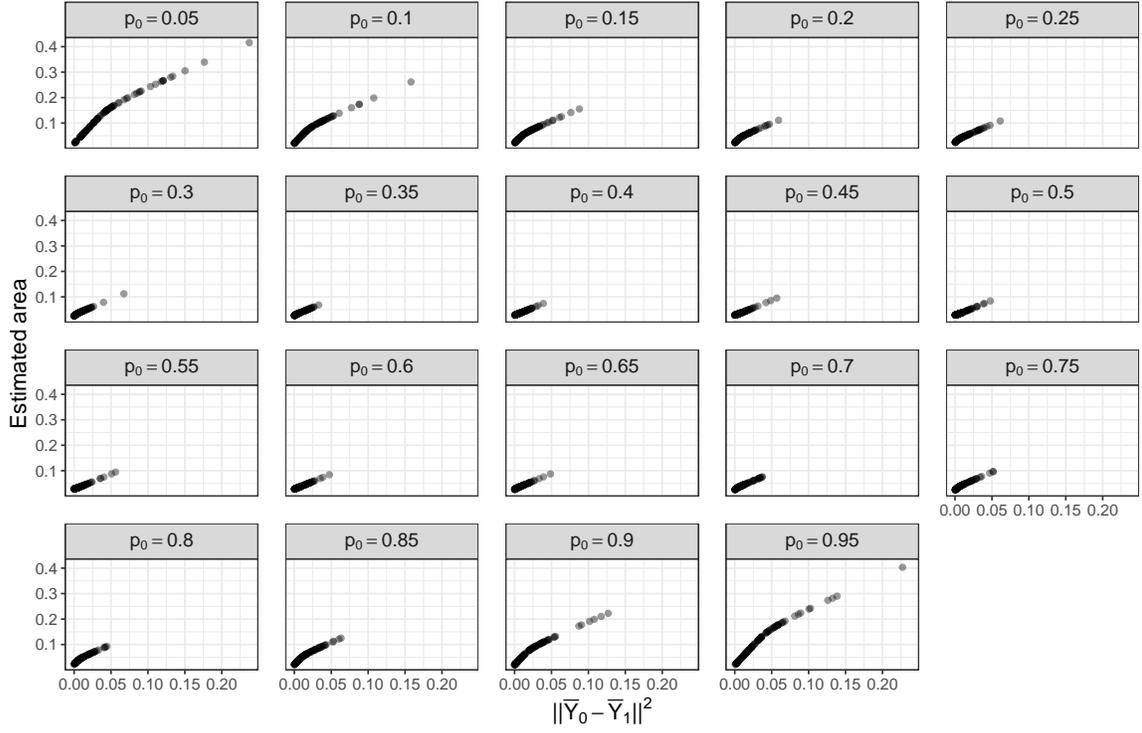

\figuresize{.6}
\figurebox{20pc}{25pc}{}[SuppFigure5.pdf]
\caption{Using the same simulations as Figures~\ref{fig:CF_area_by_p0} and \ref{fig:CF_area_by_sq_dist}, we plot the simulated cross-fit area against $\|\bar{Y}_0 - \bar{Y}_1\|^2$ separately for each $p_0$. Within each choice of $p_0$, simulations with smaller $\|\bar{Y}_0 - \bar{Y}_1\|^2$ also have smaller estimated areas.}
\label{fig:CF_area_by_sq_dist_and_p0}
\end{figure}

\section{Power of Tests of $H_0: \|\theta^*\| \in [0.5, 1]$} \label{sec:supp_intersect_test}

\subsection{Exact Formula for Power of Intersection Test}

In section~\ref{sec:example}, we present hypothesis tests for $H_0: \|\theta^*\| \in [0.5, 1]$ versus $H_1: \|\theta^*\| \notin [0.5, 1]$. The power of the intersection method that we present is tractable. We derive a formula for the intersection method's power at $\theta^*$. From the intersection method's description, we reject $H_0$ if and only if $C_n^\text{LRT}(\alpha) \cap (\S_1 \backslash \S_{0.5}) = \emptyset$, where $C_n^\text{LRT}(\alpha) = \left\{\theta\in\Theta: \|\theta - \bar{Y}\|^2 \leq c_{\alpha,d} / n\right\}$. This is equivalent to rejecting $H_0$ if and only if $\hat{\theta}^{\text{proj}} \notin C_n^\text{LRT}(\alpha)$, where
\[\hat{\theta}^{\text{proj}} = \begin{cases}
0.5 \: \bar{Y} / \|\bar{Y}\| &\text{if  } \|\bar{Y}\| < 0.5 \\
\bar{Y} &\text{if  } \|\bar{Y}\| \in [0.5, 1.0] \\
\bar{Y} / \|\bar{Y}\| &\text{if  } \|\bar{Y}\| > 1
\end{cases}.
\]

In \textbf{Case 2}, we have $\|\bar{Y}\|\in [0.5, 1]$. In this setting, it is always true that $\hat{\theta}^{\text{proj}} = \bar{Y} \in C_n^\text{LRT}(\alpha)$. So we will never reject $H_0$ in this case. We consider \textbf{Case 1} ($\|\bar{Y}\| < 0.5$) and \textbf{Case 3} ($\|\bar{Y}\| > 1$). For $\|\theta^*\|\notin [0.5, 1.0]$, the power is given by 
\begin{align*}
\text{Power}(\theta^*) &= \P_{\theta^*}\left(\left\|\bar{Y}/\|\bar{Y}\| - \bar{Y} \right\|^2 > c_{\alpha, d}/n, \: \|\bar{Y}\| > 1 \right) + \\
&\qquad \P_{\theta^*}\left(\left\|0.5\:\bar{Y}/\|\bar{Y}\| - \bar{Y} \right\|^2 > c_{\alpha, d}/n, \: \|\bar{Y}\| < 0.5 \right).
\end{align*} 
We know that $n\|\bar{Y}\|^2 \sim \chi^2(df = d, \lambda = n\|\theta^*\|^2)$. We will use this fact to write $\text{Power}(\theta^*)$ in terms of this non-central $\chi^2$ CDF.
\\[12pt]
\textbf{Case 1.} Note that 
\begin{align*}
\left\|0.5\:\bar{Y}/\|\bar{Y}\| - \bar{Y} \right\|^2 &= \frac{\bar{Y}^T \bar{Y}}{4\|\bar{Y}\|^2} - 2 \frac{\bar{Y}^T \bar{Y}}{2\|\bar{Y}\|} + \|\bar{Y}\|^2 = \frac{1}{4} - \|\bar{Y}\| + \|\bar{Y}\|^2 = \left(\|\bar{Y}\|- \frac{1}{2} \right)^2.
\end{align*}
Then we write
\begin{align*}
\P_{\theta^*}&\left(\left\|0.5\:\bar{Y}/\|\bar{Y}\| - \bar{Y} \right\|^2 > c_{\alpha, d}/n, \: \|\bar{Y}\| < 1/2 \right) \\
&= \P_{\theta^*}\left(\left(\|\bar{Y}\|-1/2 \right)^2 > c_{\alpha,d}/n, \:\|\bar{Y}\| < 1/2 \right) \\
&= \P_{\theta^*}\left(1/2 - \|\bar{Y}\| > (c_{\alpha,d}/n)^{1/2},\: \|\bar{Y}\| < 1/2 \right) \\
&= \P_{\theta^*}\left(\|\bar{Y}\| < 1/2 - (c_{\alpha,d}/n)^{1/2},\: \|\bar{Y}\| < 1/2 \right) \\
&= \P_{\theta^*}\left(\|\bar{Y}\| < 1/2 - (c_{\alpha,d}/n)^{1/2} \right) \\
&= \one\left(c_{\alpha, d}/n < 1/4 \right) \P_{\theta^*}\left(\|\bar{Y}\| < 1/2 - (c_{\alpha,d}/n)^{1/2}\right) \\
&= \one\left(n > 4c_{\alpha, d} \right) \P_{\theta^*}\left(\|\bar{Y}\|^2 < 1/4 - \sqrt{c_{\alpha,d}/n} + c_{\alpha,d}/n\right) \\
&= \one\left(n > 4c_{\alpha, d} \right) \P_{\theta^*}\left(n\|\bar{Y}\|^2 < n/4 - \sqrt{nc_{\alpha,d}} + c_{\alpha,d}\right) \\
&= \one\left(n > 4c_{\alpha, d} \right) F_{d, n\|\theta^*\|^2}\left(n/4 - \sqrt{nc_{\alpha,d}} + c_{\alpha,d}\right)~, \stepcounter{equation}\tag{\theequation} \label{eq:intersect_case3} 
\end{align*}
where $F_{d, n\|\theta^*\|^2}$ is the non-central $\chi^2(df = d, \lambda = n\|\theta^*\|^2)$ CDF.
\\[12pt]
\textbf{Case 3.} Note that 
\begin{align*}
\left\|\bar{Y}/\|\bar{Y}\| - \bar{Y} \right\|^2 &= \frac{\bar{Y}^T \bar{Y}}{\|\bar{Y}\|^2} - 2 \frac{\bar{Y}^T \bar{Y}}{\|\bar{Y}\|} + \|\bar{Y}\|^2 = 1 - 2\|\bar{Y}\| + \|\bar{Y}\|^2 = \left(\|\bar{Y}\|-1 \right)^2.
\end{align*}
Then we write
\begin{align*}
\P_{\theta^*}&\left(\left\|\bar{Y}/\|\bar{Y}\| - \bar{Y} \right\|^2 > c_{\alpha, d}/n, \: \|\bar{Y}\| > 1 \right) \\
&= \P_{\theta^*}\left(\left(\|\bar{Y}\|-1 \right)^2 > c_{\alpha,d}/n, \: \|\bar{Y}\|^2 > 1 \right) \\
&= \P_{\theta^*}\left(\|\bar{Y}\|-1 > (c_{\alpha,d}/n)^{1/2}, \: \|\bar{Y}\|^2 > 1 \right) \\
&= \P_{\theta^*}\left(\|\bar{Y}\| > 1 + (c_{\alpha,d}/n)^{1/2},\: \|\bar{Y}\|^2 > 1 \right) \\
&= \P_{\theta^*}\left(\|\bar{Y}\|^2 > 1 + (2/\sqrt{n})c_{\alpha,d}^{1/2} + c_{\alpha,d}/n,\: \|\bar{Y}\|^2 > 1 \right) \\
&= \P_{\theta^*}\left(\|\bar{Y}\|^2 > 1 + (2/\sqrt{n})c_{\alpha,d}^{1/2} + c_{\alpha,d}/n \right) \\
&= \P_{\theta^*}\left(n\|\bar{Y}\|^2 > n + 2\sqrt{n c_{\alpha,d}} + c_{\alpha,d} \right) \\
&= 1 - F_{d, n\|\theta^*\|^2}(n +   2\sqrt{n c_{\alpha,d}} + c_{\alpha,d}), \stepcounter{equation}\tag{\theequation} \label{eq:intersect_case2} 
\end{align*}
where $F_{d, n\|\theta^*\|^2}$ is the non-central $\chi^2(df = d, \lambda = n\|\theta^*\|^2)$ CDF.

For a given $\|\theta^*\|\notin [0.5, 1]$, our calculation of $\text{Power}(\theta^*)$ is given by (\ref{eq:intersect_case2}) + (\ref{eq:intersect_case3}). That is, 
\begin{align*}
\text{Power}(\theta^*) &= 1 - F_{d, n\|\theta^*\|^2}(n +   2\sqrt{n c_{\alpha,d}} + c_{\alpha,d}) +  \\
&\qquad \one\left(n > 4c_{\alpha, d} \right) F_{d, n\|\theta^*\|^2}\left(n/4 - \sqrt{nc_{\alpha,d}} + c_{\alpha,d}\right).
\end{align*}

Figure~\ref{fig:power_intersect} compares this calculated power to the simulated power of the intersection method from Figure~\ref{fig:power_interval}. The points correspond to the simulated power, and the curves trace out the calculated power. The calculated and simulated powers align.

\begin{figure}
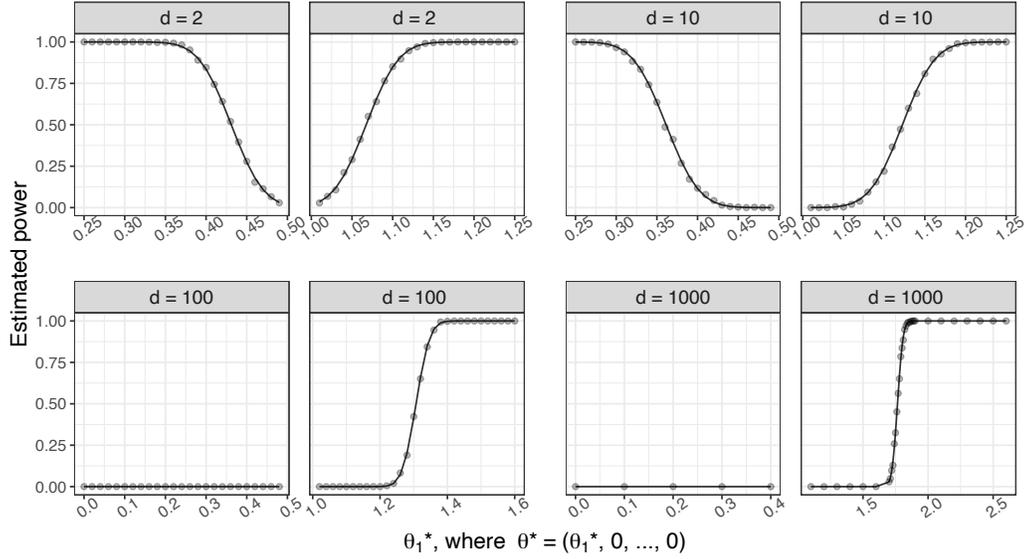

\figuresize{.6}
\figurebox{20pc}{25pc}{}[SuppFigure6.pdf]
\caption{Calculated power of $H_0: \|\theta^*\| \in [0.5, 1.0]$ versus $H_1: \|\theta^*\| \notin [0.5, 1.0]$ using the intersection method. We compare the simulated power to the calculation (\ref{eq:intersect_case2}) + (\ref{eq:intersect_case3}). The points correspond to the simulated power, and the curves trace out the calculated power.}
\label{fig:power_intersect}
\end{figure}

\subsection{Cases of the Subsampled Hybrid LRT}

The subsampled hybrid test of $H_0: \|\theta^*\| \in [0.5, 1]$ versus $H_1: \|\theta^*\| \notin [0.5, 1]$ takes one of three approaches within each repeated subsample:
\begin{enumerate}
\item If $\|\bar{Y}_{1,b}\| < 0.5$, use the split LRT statistic $U_n$ on the $b^{th}$ subsample.
\item If $\|\bar{Y}_{1,b}\| \in [0.5, 1]$, set the $b^{th}$ subsample's test statistic to 1.
\item If $\|\bar{Y}_{1,b}\| > 1$, use the RIPR LRT statistic $R_n$ on the $b^{th}$ subsample.
\end{enumerate}    

Figure~\ref{fig:case_props} shows the proportion of these three cases that make up the hybrid test. We consider all $\|\theta^*\|$ values from Figure~\ref{fig:power_interval} of the main paper, as well as cases where $\|\theta^*\|$ is within the null region. At any given value of $d$ and $\|\theta^*\|$, the three proportions sum to 1. Interestingly, although $\|\bar{Y}_{1,b}\| < 0.5$ approximately 95\% of the time when $\|\theta^*\| = 0$ and $d = 100$, the hybrid test has approximately zero power at that choice of parameters. We derive this fact in section~\ref{sec:hybrid_power_d100}. In addition, when $d = 1000$ we see that $\|\bar{Y}_{1,b}\| > 1$ in all simulations, even at $\theta^* = 0$. In section~\ref{sec:hybrid_power_d1000}, we see why  this setting has approximately zero power as well.

\begin{figure}
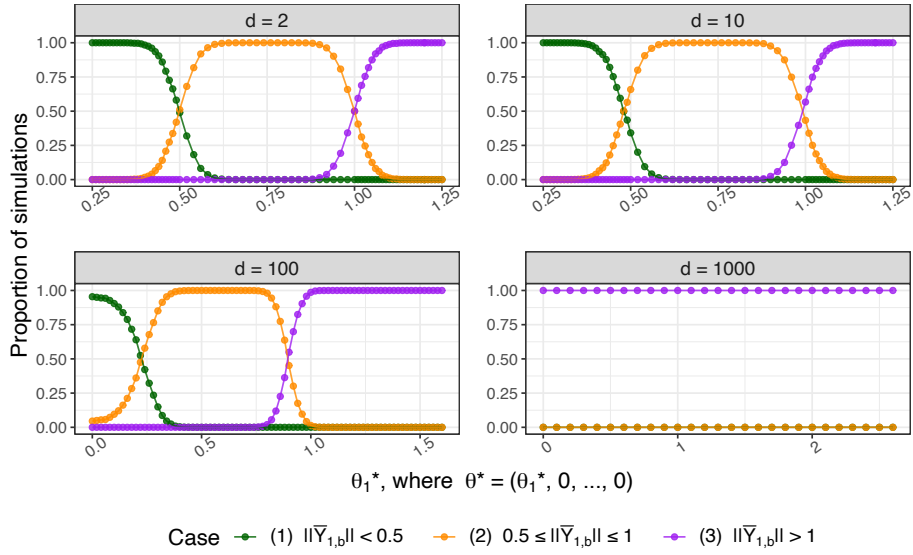

\figuresize{.6}
\figurebox{20pc}{25pc}{}[SuppFigure7.pdf]
\caption{Proportions of three cases that compose the hybrid LRT. We set $\alpha = 0.10$ and $n=1000$, and we perform 1000 simulations at each value of $\|\theta^*\|$. We subsample $B = 100$ times.}
\label{fig:case_props}
\end{figure}

\subsection{Hybrid power when $\theta^* = 0$, $d = 100$, and $n = 1000$} \label{sec:hybrid_power_d100}

When $\theta^* = 0$, $d = 100$, and $n = 1000$, Figure~\ref{fig:case_props} shows that $\|\bar{Y}_{1,b}\| < 0.5$ (case 1) occurs with probability of approximately 0.95, and $\|\bar{Y}_{1,b}\| \in [0.5, 1]$ (case 2) occurs with probability of approximately 0.05. At these parameters, the hybrid method has power of approximately 0, as shown in Figure~\ref{fig:power_interval} in the main paper. We consider the power of the hybrid method at a single split of the data:

{\scriptsize
\begin{align*}
\P&_{\theta^* = 0}(U_n \one(\|\bar{Y}_1\| < 0.5) + \one(\|\bar{Y}_1\| \in [0.5, 1]) + R_n \one(\|\bar{Y}_1\| > 1) \geq 1/\alpha) \\
&= \underbrace{\P_{\theta^* = 0}(\|\bar{Y}_1\| < 0.5, \|\bar{Y}_0\| < 0.5)}_{A_1} \underbrace{\P_{\theta^* = 0}\left(\exp\left(-\frac{n}{4}\|\bar{Y}_0 - \bar{Y}_1\|^2 + \frac{n}{4}\|\bar{Y}_0 - 0.5 \bar{Y}_0 / \|\bar{Y}_0\|\|^2 \right) \geq \frac{1}{\alpha} \: \Big| \: \|\bar{Y}_1\| < 0.5, \|\bar{Y}_0\| < 0.5 \right)}_{A_2} + \\
&\qquad \underbrace{\P_{\theta^* = 0}(\|\bar{Y}_1\| < 0.5, \|\bar{Y}_0\| \in [0.5, 1])}_{B_1} \underbrace{\P_{\theta^* = 0}\left(\exp\left(-\frac{n}{4}\|\bar{Y}_0 - \bar{Y}_1\|^2 + \frac{n}{4}\|\bar{Y}_0 - \bar{Y}_0 \|^2 \right) \geq \frac{1}{\alpha} \: \Big| \: \|\bar{Y}_1\| < 0.5, \|\bar{Y}_0\| \in [0.5, 1] \right)}_{B_2} + \\
&\qquad \underbrace{\P_{\theta^* = 0}(\|\bar{Y}_1\| < 0.5, \|\bar{Y}_0\| > 1)}_{C_1} \underbrace{\P_{\theta^* = 0}\left(\exp\left(-\frac{n}{4}\|\bar{Y}_0 - \bar{Y}_1\|^2 + \frac{n}{4}\|\bar{Y}_0 - \bar{Y}_0/ \|\bar{Y}_0\| \|^2 \right) \geq \frac{1}{\alpha} \: \Big| \: \|\bar{Y}_1\| < 0.5, \|\bar{Y}_0\| > 1 \right)}_{C_2} + \\
&\qquad \underbrace{\P_{\theta^* = 0}(\|\bar{Y}_1\| \in [0.5, 1])}_{D_1} \underbrace{\P_{\theta^* = 0}(1 \geq 1/\alpha \mid \|\bar{Y}_1\| \in [0.5, 1])}_{D_2} + \\
&\qquad \underbrace{\P_{\theta^* = 0}(\|\bar{Y}_1\| > 1)}_{E_1} \underbrace{\P_{\theta^* = 0}\left(\exp\left(-\frac{n}{4}\|\bar{Y}_0 - \bar{Y}_1\|^2 + \frac{n}{4}\|\bar{Y}_0 - \bar{Y}_1 / \|\bar{Y}_1\| \|^2 \right) \geq \frac{1}{\alpha} \: \Big| \: \|\bar{Y}_1\| > 1 \right)}_{E_2}.
\end{align*}
}
The probabilities $B_2$ and $D_2$ equal 0. In addition,
\begin{align*}
\P_{\theta^* = 0}(\|\bar{Y}_0\| > 1) &= \P_{\theta^* = 0}(\|\bar{Y}_1\| > 1) \\
&= \P_{\theta^* = 0}\left(\frac{n}{2} \|\bar{Y}_1\|^2 > \frac{n}{2}\right) \\
&= \P(\chi^2_{df = 100} > 1000/2) \\
&\approx 0.
\end{align*}
So $C_1$ and $E_1$ are also approximately 0. That means we only need to consider $A_1 A_2$. Working with the joint probability, we see

{\small
\begin{align*}
A_1 A_2 &= \P_{\theta^* = 0}\left(\exp\left(-\frac{n}{4}\|\bar{Y}_0 - \bar{Y}_1\|^2 + \frac{n}{4}\|\bar{Y}_0 - 0.5 \bar{Y}_0 / \|\bar{Y}_0\|\|^2 \right) \geq \frac{1}{\alpha}, \: \|\bar{Y}_1\| < 0.5, \: \|\bar{Y}_0\| < 0.5 \right) \\
&\leq \P_{\theta^* = 0} \left(\|\bar{Y}_0 - \bar{Y}_1\|^2 < \|\bar{Y}_0 - 0.5 \bar{Y}_0 / \|\bar{Y}_0\| \|^2, \|\bar{Y}_0\| < 0.5\right) \\
&\leq \P_{\theta^* = 0} (\|\bar{Y}_0 - \bar{Y}_1\|^2 < 0.25) \\
&= \P((4/n)\chi^2_{df = 100} < 1/4) \\
&= \P(\chi^2_{df = 100} < 1000/16) \\
&\approx 0.001.
\end{align*}
}
This means that at a single split of the data, the power at $\|\theta^*\| = 0$, $d = 100$, and $n = 1000$ is
{
\begin{align*}
\P_{\theta^* = 0}(U_n \one(\|\bar{Y}_1\| < 0.5) + \one(\|\bar{Y}_1\| \in [0.5, 1]) + R_n \one(\|\bar{Y}_1\| > 1) \geq 1/\alpha) \leq 0.001.
\end{align*}
}

\subsection{Hybrid power when $\theta^* = 0$, $d = 1000$, and $n = 1000$} \label{sec:hybrid_power_d1000}

When $\theta^* = 0$, $d = 1000$, and $n = 1000$, we see that the hybrid method selects case 3 ($\|\bar{Y}_{1,b}\| > 1$) in all simulations. This is essentially choosing the wrong case, since $\|\theta^*\| = 0 < 0.5$. Numerically, we can show that the hybrid method will have power of approximately 0 at these parameters. Again, we consider a single split of the data. 

{\small
\begin{align*}
\P_{\theta^* = 0}&(U_n \one(\|\bar{Y}_1\| < 0.5) + \one(\|\bar{Y}_1\| \in [0.5, 1]) + R_n \one(\|\bar{Y}_1\| > 1) \geq 1/\alpha) \\
&= \underbrace{\P_{\theta^* = 0}(\|\bar{Y}_1\| < 0.5)}_{A_1} \underbrace{\P_{\theta^* = 0}\left(U_n \geq 1/\alpha \: \Big| \: \|\bar{Y}_1\| < 0.5 \right)}_{A_2} + \\
&\qquad \underbrace{\P_{\theta^* = 0}(\|\bar{Y}_1\| \in [0.5, 1])}_{B_1} \underbrace{\P_{\theta^* = 0}(1 \geq 1/\alpha \mid \|\bar{Y}_1\| \in [0.5, 1])}_{B_2} + \\
&\qquad \underbrace{\P_{\theta^* = 0}(\|\bar{Y}_1\| > 1)}_{C_1} \underbrace{\P_{\theta^* = 0}\left(\exp\left(-\frac{n}{4}\|\bar{Y}_0 - \bar{Y}_1\|^2 + \frac{n}{4}\|\bar{Y}_0 - \bar{Y}_1 / \|\bar{Y}_1\| \|^2 \right) \geq \frac{1}{\alpha} \: \Big| \: \|\bar{Y}_1\| > 1 \right)}_{C_2}.
\end{align*}
}

The probability $B_2$ equals 0. In addition, $A_1$ is approximately 0 because
\begin{align*}
\P_{\theta^* = 0}(\|\bar{Y}_1\| < 0.5) &= \P_{\theta^* = 0}\left((n/2) \|\bar{Y}_1\|^2 < n/8 \right) \\
&= \P(\chi^2_{df = 1000} < 1000/8) \\
&\approx 0.
\end{align*}

So the probability of rejecting $H_0$ at this choice of parameters is approximately
\begin{align*}
C_1 C_2 &= \P_{\theta^* = 0}\left(\exp\left(-\frac{n}{4}\|\bar{Y}_0 - \bar{Y}_1\|^2 + \frac{n}{4}\|\bar{Y}_0 - \bar{Y}_1 / \|\bar{Y}_1\| \|^2 \right) \geq \frac{1}{\alpha}, \|\bar{Y}_1\| > 1 \right) \\
&\leq \P_{\theta^* = 0}\left( \|\bar{Y}_0 - \bar{Y}_1\|^2 < \|\bar{Y}_0 - \bar{Y}_1 / \|\bar{Y}_1\| \|^2, \|\bar{Y}_1\| > 1 \right) \\
&= \P_{\theta^* = 0}\left(\|\bar{Y}_0\|^2 - 2\bar{Y}_0^T \bar{Y}_1 + \|\bar{Y}_1\|^2 < \|\bar{Y}_0\|^2 - 2\bar{Y}_0^T \bar{Y}_1/\|\bar{Y}_1\| + 1, \|\bar{Y}_1\| > 1  \right) \\
&= \P_{\theta^* = 0}\left(2\bar{Y}_0^T \bar{Y}_1 (1 / \|\bar{Y}_1\| - 1) + \|\bar{Y}_1\|^2 < 1, \|\bar{Y}_1\| > 1  \right) \\
&= \P_{\theta^* = 0}\left(2\bar{Y}_0^T \bar{Y}_1 (1 - \|\bar{Y}_1\|) /  \|\bar{Y}_1\| < 1 - \|\bar{Y}_1\|^2, \|\bar{Y}_1\| > 1  \right) \\
&= \P_{\theta^* = 0}\left(2\bar{Y}_0^T \bar{Y}_1 (1 - \|\bar{Y}_1\|) /  \|\bar{Y}_1\| < (1 - \|\bar{Y}_1\|) (1 + \|\bar{Y}_1\|), \|\bar{Y}_1\| > 1  \right) \\
&= \P_{\theta^* = 0}\left(2\bar{Y}_0^T \bar{Y}_1  > \|\bar{Y}_1\| (1 + \|\bar{Y}_1\|), \|\bar{Y}_1\| > 1  \right) \\
&\leq \P_{\theta^* = 0}\left(\bar{Y}_0^T \bar{Y}_1 > 1\right).
\end{align*}
Let $\sigma = 1/\sqrt{500}$. Since $\bar{Y}_0$ and $\bar{Y}_1$ are averages of 500 $N(0, I_d)$ random variables, we see that $\bar{Y}_0 \sim N(0, \sigma^2 I_d)$ and $\bar{Y}_1 \sim N(0, \sigma^2 I_d)$. Let $\lambda = -d/2 + (1/2)\sqrt{d^2 + 4/\sigma^4}$. (This choice of $\lambda$ minimizes $\E[\exp(\lambda \bar{Y}_0^T \bar{Y}_1)] / \exp(\lambda)$ out of $\lambda > 0$.) Let $\nu = \sigma / (1 - \sigma^4 \lambda^2)^{1/2}$. We derive
{\footnotesize
\begin{align*}
\P&_{\theta^* = 0}\left(\bar{Y}_0^T \bar{Y}_1 > 1\right) \\
&= \P_{\theta^* = 0}\left(\exp\left( \lambda \bar{Y}_0^T \bar{Y}_1 \right) > \exp(\lambda) \right) \\
&\leq \E_{\theta^* = 0}\left[\exp\left( \lambda \bar{Y}_0^T \bar{Y}_1 \right)\right] / \exp(\lambda) \\
&= \exp(-\lambda) \int_{\R^d} \int_{\R^d} \frac{1}{(2\pi)^d |\sigma^2 I_d|} \exp\left( -\frac{1}{2\sigma^2} \|\bar{Y}_0\|^2 - \frac{1}{2\sigma^2} \|\bar{Y}_1\|^2 + \lambda \bar{Y}_0^T \bar{Y}_1 \right) d\bar{Y}_0 d\bar{Y}_1 \\
&= \exp(-\lambda) \int_{\R^d} \frac{1}{(2\pi)^{d/2} |\sigma^2 I_d|^{1/2}} \exp\left(-\frac{1}{2\sigma^2} \|\bar{Y}_1\|^2\right) \left\{ \int_{\R^d} \frac{1}{(2\pi)^{d/2} |\sigma^2 I_d|^{1/2}} \exp\left( -\frac{1}{2\sigma^2} \|\bar{Y}_0\|^2 + \lambda \bar{Y}_0^T \bar{Y}_1 \right) d\bar{Y}_0 \right\} d\bar{Y}_1 \\
&= \exp(-\lambda) \int_{\R^d} \frac{1}{(2\pi)^{d/2} |\sigma^2 I_d|^{1/2}} \exp\left(-\frac{1}{2\sigma^2} \|\bar{Y}_1\|^2\right) \left\{ \E\left[\exp((\lambda \bar{Y}_1)^T \bar{Y}_0) \mid \bar{Y}_1 \right] \right\} d\bar{Y}_1 \\
&= \exp(-\lambda) \int_{\R^d} \frac{1}{(2\pi)^{d/2} |\sigma^2 I_d|^{1/2}} \exp\left(-\frac{1}{2\sigma^2} \|\bar{Y}_1\|^2\right) \exp\left(\frac{1}{2} \lambda^2 \sigma^2 \|\bar{Y}_1\|^2 \right) d\bar{Y}_1 \\
&= \exp(-\lambda) \int_{\R^d} \frac{1}{(2\pi)^{d/2} |\sigma^2 I_d|^{1/2}} \exp\left(-\frac{1}{2} \left(\frac{1}{\sigma^2} - \sigma^2 \lambda^2 \right) \|\bar{Y}_1\|^2 \right) d\bar{Y}_1 \\
&= \exp(-\lambda) \int_{\R^d} \frac{1}{(2\pi)^{d/2} |\sigma^2 I_d|^{1/2}} \exp\left(-\frac{1}{2} \left(\frac{1 -\sigma^4 \lambda^2}{\sigma^2}\right) \|\bar{Y}_1\|^2 \right) d\bar{Y}_1 \\
&= \exp(-\lambda) \int_{\R^d} \frac{1}{(2\pi)^{d/2} |\sigma^2 I_d|^{1/2}} \exp\left(-\frac{1}{2\nu^2} \|\bar{Y}_1\|^2 \right) d\bar{Y}_1 \\
&= \exp(-\lambda) \frac{|\nu^2 I_d|^{1/2}}{|\sigma^2 I_d|^{1/2}} \int_{\R^d} \frac{1}{(2\pi)^{d/2} |\nu^2 I_d|^{1/2}} \exp\left(-\frac{1}{2\nu^2} \|\bar{Y}_1\|^2 \right) d\bar{Y}_1 \\
&= \exp(-\lambda)(\nu/\sigma)^d \\
&\approx \exp(-207) (1.1)^{1000} \\
&\approx 0.
\end{align*}
}

At a single split of the data, the power at $\|\theta^*\| = 0$, $d = 1000$, and $n = 1000$ is approximately 0 because
\begin{align*}
\P_{\theta^* = 0}&(U_n \one(\|\bar{Y}_1\| < 0.5) + \one(\|\bar{Y}_1\| \in [0.5, 1]) + R_n \one(\|\bar{Y}_1\| > 1) \geq 1/\alpha) \\
&\approx \P_{\theta^* = 0}\left(\exp\left(-\frac{n}{4}\|\bar{Y}_0 - \bar{Y}_1\|^2 + \frac{n}{4}\|\bar{Y}_0 - \bar{Y}_1 / \|\bar{Y}_1\| \|^2 \right) \geq \frac{1}{\alpha}, \|\bar{Y}_1\| > 1 \right) \\
&\leq \P_{\theta^* = 0}\left(\bar{Y}_0^T \bar{Y}_1 > 1\right) \\
&\approx 0.
\end{align*}

\end{document}